\documentclass[11pt]{article}

\usepackage{arxiv}
\hypersetup{hypertexnames=false}
\crefname{claim}{claim}{claims}
\Crefname{claim}{Claim}{Claims}

\usepackage{titlesec}
\titlespacing*{\paragraph}{0pt}{1.5ex plus .5ex minus .2ex}{1ex plus .3ex}
\usepackage{enumitem}
\setlist[itemize]{itemsep=3pt, topsep=4pt, parsep=0pt, partopsep=0pt}

\title{Quantum Property Testing for Bounded-Degree Directed Graphs}
\author{
  Pan Peng\footnote{
	School of Computer Science and Technology, University of Science and Technology of China.  Email: \href{ppeng@ustc.edu.cn}{ppeng@ustc.edu.cn}}
  \and
  Jingyu Wu\footnote{
	School of Computer Science and Technology, University of Science and Technology of China.  Email: \href{wujingyu@mail.ustc.edu.cn}{wujingyu@mail.ustc.edu.cn}}\\
}
\date{}

\begin{document}

\maketitle

\makeatletter
\newcommand{\rmnum}[1]{\romannumeral #1} %
\newcommand{\Rmnum}[1]{\expandafter\@slowromancap\romannumeral #1@} %
\makeatother

\newcommand{\1}{\mathds{1}}
\newcommand{\vecone}{\mathbf{1}}
\newcommand{\A}{\mathbb{A}}
\newcommand{\E}{\mathbb{E}}
\newcommand{\I}{\mathbf{1}}
\newcommand{\R}{\mathbb{R}}
\newcommand{\N}{\mathbb{N}}
\newcommand{\Z}{\mathbb{Z}}
\newcommand{\One}{\boldsymbol{1}}
\newcommand{\bA}{\boldsymbol{A}}
\newcommand{\ba}{\boldsymbol{a}}
\newcommand{\bB}{\boldsymbol{B}}
\newcommand{\be}{\boldsymbol{e}}
\newcommand{\bI}{\boldsymbol{I}}
\newcommand{\bL}{\boldsymbol{L}}
\newcommand{\bO}{\boldsymbol{O}}
\newcommand{\br}{\boldsymbol{r}}
\newcommand{\bR}{\boldsymbol{R}}
\newcommand{\bS}{\boldsymbol{S}}
\newcommand{\bT}{\boldsymbol{T}}
\newcommand{\bU}{\boldsymbol{U}}
\newcommand{\bu}{\boldsymbol{u}}
\newcommand{\bV}{\boldsymbol{V}}
\newcommand{\bv}{\boldsymbol{v}}
\newcommand{\bw}{\boldsymbol{w}}
\newcommand{\bW}{\boldsymbol{W}}
\newcommand{\bX}{\boldsymbol{X}}
\newcommand{\bY}{\boldsymbol{Y}}
\newcommand{\bx}{\boldsymbol{x}}
\newcommand{\by}{\boldsymbol{y}}
\newcommand{\bz}{\boldsymbol{z}}
\newcommand{\bZ}{\boldsymbol{Z}}
\newcommand{\bpi}{\boldsymbol{\pi}}
\newcommand{\cA}{\mathcal{A}}
\newcommand{\cB}{\mathcal{B}}
\newcommand{\cC}{\mathcal{C}}
\newcommand{\cD}{\mathcal{D}}
\newcommand{\cE}{\mathcal{E}}
\newcommand{\cF}{\mathcal{F}}
\newcommand{\cG}{\mathcal{G}}
\newcommand{\cH}{\mathcal{H}}
\newcommand{\cI}{\mathcal{I}}
\newcommand{\cJ}{\mathcal{J}}
\newcommand{\cK}{\mathcal{K}}
\newcommand{\cL}{\mathcal{L}}
\newcommand{\cM}{\mathcal{M}}
\newcommand{\cN}{\mathcal{N}}
\newcommand{\cO}{\mathcal{O}}
\newcommand{\cP}{\mathcal{P}}
\newcommand{\cQ}{\mathcal{Q}}
\newcommand{\cR}{\mathcal{R}}
\newcommand{\cS}{\mathcal{S}}
\newcommand{\cT}{\mathcal{T}}
\newcommand{\cU}{\mathcal{U}}
\newcommand{\cV}{\mathcal{V}}
\newcommand{\cW}{\mathcal{W}}
\newcommand{\cX}{\mathcal{X}}
\newcommand{\cY}{\mathcal{Y}}
\newcommand{\ttheta}{\tilde{\theta}}

\newcommand{\bp}{\mathbf{p}}

\newcommand{\perm}[2]{\genfrac{[}{]}{0pt}{}{#1}{#2}}

\newcommand{\adeg}{\widetilde{\operatorname{deg}}}
\newcommand{\ubdeg}{\widetilde{\operatorname{ubdeg}}}
\newcommand{\dpdeg}{\widetilde{\operatorname{dpdeg}}}

\newcommand{\bs}{\text{bs}}
\newcommand{\acc}{\operatorname{acc}}
\newcommand{\Adv}{\operatorname{Adv}}
\newcommand{\Dom}{\operatorname{Dom}}
\newcommand{\Advpm}{\operatorname{Adv}^{\pm}}
\newcommand{\Var}{\operatorname{Var}}
\newcommand{\Cov}[1]{{\operatorname{Cov}\left[ #1\right]}}
\newcommand{\Supp}[1]{{\operatorname{Supp}\left( #1\right)}}
\newcommand{\norm}[1]{{\lVert #1 \rVert}}
\newcommand{\abs}[1]{{\left\lvert #1\right\rvert}}
\newcommand{\eps}{\varepsilon}

\newcommand{\identity}{\mathbf{I}}

\newcommand{\dout}{d_{\mathrm{out}}}
\newcommand{\din}{d_{\mathrm{in}}}
\newcommand{\out}{{\mathrm{out}}}
\newcommand{\cnt}{{\mathrm{cnt}}}
\newcommand{\dist}{\mathrm{dist}}
\newcommand{\disc}{\mathrm{disc}}
\newcommand{\head}{\mathrm{head}}
\newcommand{\tail}{\mathrm{tail}}
\newcommand{\Flag}{\mathrm{Flag}}
\newcommand{\BFS}{\mathrm{BFS}}
\newcommand{\est}{\mathrm{est}}
\newcommand{\prob}{\mathrm{prob}}
\newcommand{\dis}{\mathrm{dis}}
\newcommand{\eq}{\mathrm{eq}}
\newcommand{\zero}{\mathrm{zero}}
\newcommand{\one}{\mathrm{one}}
\newcommand{\Grover}{\textnormal{\textsc{Grover}}}
\newcommand{\Count}{\textnormal{\textsc{Count}}}
\newcommand{\poly}{\mathrm{poly}}
\newcommand{\forr}{\mathrm{forr}}
\newcommand{\rt}{\mathrm{rt}}
\newcommand{\type}{\mathrm{Type}}
\newcommand{\pass}{\mathrm{pass}}
\newcommand{\freq}{\operatorname{freq}}
\newcommand{\inn}{\mathrm{in}}
\newcommand{\Unif}{\operatorname{Unif}}
\newcommand{\C}{\mathbb{C}}
\newcommand{\calA}{\mathcal{A}}
\newcommand{\calB}{\mathcal{B}}
\newcommand{\calC}{\mathcal{C}}
\newcommand{\calD}{\mathcal{D}}
\newcommand{\calE}{\mathcal{E}}
\newcommand{\calM}{\mathcal{M}}
\newcommand{\calN}{\mathcal{N}}
\newcommand{\calP}{\mathcal{P}}
\newcommand{\calQ}{\mathcal{Q}}
\newcommand{\calR}{\mathcal{R}}
\newcommand{\calT}{\mathcal{T}}
\newcommand{\calU}{\mathcal{U}}
\newcommand{\Ans}{\mathsf{Ans}}
\newcommand{\str}{\mathrm{str}}
\newcommand{\Sym}{\mathsf{Sym}}

\newcommand{\bool}{\{0,1\}}

\newcommand{\sgn}{\operatorname{sgn}}
\newcommand{\bdeg}{\operatorname{bdeg}}
\newcommand{\phd}{\operatorname{phd}}

\newcommand{\AND}{\mathsf{AND}}
\newcommand{\OR}{\mathsf{OR}}
\newcommand{\OCCU}{\mathsf{ OCCU}}
\newcommand{\dOCCU}{\mathsf{ dOCCU}}
\newcommand{\GapOR}{\mathsf{GapOR}}
\newcommand{\THR}{\mathsf{THR}}
\newcommand{\BTHR}{\mathsf{BTHR}}
\newcommand{\COLL}{\mathsf{Collision}}

\newcommand{\bits}{\{-1,1\}} %
\begin{abstract}
We study quantum property testing of directed graphs whose maximum in-degree and
out-degree are bounded by a fixed constant \(d\). For a proximity parameter
\(\varepsilon\), we prove that every property testable with
\(O_{\varepsilon,d}(1)\) quantum queries in the bidirectional model, where both
incoming and outgoing neighbors are accessible, can also be tested in the quantum
unidirectional model, where only outgoing neighbors are accessible, using $
n^{1/2-\Omega_{\varepsilon,d}(1)}
$
queries. This gives an almost quadratic quantum speedup over the best known
generic classical transformation.

Our proof has two main ingredients. First, we show that, in the bidirectional model
for bounded-degree digraphs, the class of properties testable with constantly many
classical queries coincides with the class testable with constantly many quantum
queries. Second, we give a transformation from classical bidirectional testers to
quantum unidirectional testers by designing a quantum unidirectional algorithm for
estimating the frequency vector of constant-radius rooted discs. The algorithm
combines quantum counting and Grover search with a correction procedure that removes
false local appearances of smaller disc types inside larger neighborhoods.

We further show that this transformation is essentially tight. For every sufficiently
small fixed \(\varepsilon>0\), there is a degree bound \(d=d(\varepsilon)\) and an
explicit property \(P_\varepsilon\) for \(d\)-bounded-degree digraphs that is
\(\varepsilon\)-testable with \(O_\varepsilon(1)\) classical, and hence quantum, 
bidirectional queries, but requires
$
\widetilde{\Omega}_\varepsilon
\left(n^{1/2-f'(\varepsilon)}\right)
$
quantum queries in the unidirectional model, where
\(f'(\varepsilon)\to 0\) as \(\varepsilon\to 0\).
\end{abstract}

 \thispagestyle{empty}

\newpage
\tableofcontents
\thispagestyle{empty}

\newpage
\pagenumbering{arabic}

\section{Introduction}

Graph property testing is a family of fundamental partial decision problems that seeks to distinguish between graphs that satisfy a given property (e.g., being connected) and those that are \emph{far} from having that property (e.g., having many disconnected components), while examining only a small portion of the input graph. A graph property is, by definition, preserved under graph isomorphisms; equivalently, it is invariant under vertex relabeling. The notion of \emph{distance} is quantified by a proximity parameter \(\eps\). In the bounded-degree model, an \(n\)-vertex graph is \(\eps\)-far from a property if at least \(\eps dn\) directed edge additions or deletions are required to obtain a graph in the property with the same degree bounds; see \Cref{sec:query-models} for the precise convention. Given such a property, a randomized algorithm that has query access to the graph and distinguishes between satisfying and $\eps$-far instances with probability at least $2/3$ is called an \emph{$\eps$-tester}; the property is then said to be \emph{$\varepsilon$-testable} by the algorithm. This field provides a foundational framework for understanding how much local information is sufficient to make reliable decisions about global graph properties. Beyond its theoretical significance, graph property testing is motivated by practical needs: it enables efficient analysis of massive networks without reading the entire input and can serve as a lightweight preprocessing step before applying more computationally intensive methods such as learning algorithms. Since the seminal works of Goldreich, Goldwasser, and Ron on dense graphs \cite{10.1145/285055.285060} and of Goldreich and Ron on bounded-degree graphs \cite{goldreich1997property}, the area has seen substantial growth over the past nearly three decades. 
For comprehensive surveys, see \cite{goldreich2017introduction, bhattacharyya2022property}.

The rise of quantum computing has naturally led to the field of quantum property testing, which lies at the intersection of quantum algorithms and property testing. It encompasses three main directions: quantum testing of classical properties, classical testing of quantum properties, and quantum testing of quantum properties. The most studied among these is quantum testing of classical properties (see the survey \cite{DBLP:journals/toc/MontanaroW16}). A central question is whether quantum algorithms can outperform classical testers, and if so, by how much -- ranging from exponential to polynomial speedups, or none at all. This is closely linked to understanding what kinds of symmetries a property can exhibit while still enabling (super-polynomial) quantum speedups (see e.g. \cite{aaronson2014need, chailloux2019note, Ben_David_2020}).

In the \emph{adjacency matrix model}, Chakraborty, Fischer, Matsliah and de Wolf \cite{Chakraborty2010NewRO} studied the problem of testing graph isomorphism for dense graphs. They demonstrated a polynomial quantum speedup in query complexity over classical testers in the setting where one graph is given in advance, and the other can only be accessed via adjacency matrix queries \cite{fischer2008testing}.
{More recently, \cite{Ben_David_2020} proved that in this model, the randomized and quantum query complexities of testing every graph property can differ by at most a polynomial factor.}

In the \emph{adjacency-list model} for \emph{bounded-degree graphs}, Ambainis, Childs and Liu \cite{Ambainis_2011} introduced the first quantum algorithms for testing bipartiteness and expansion. By combining classical and quantum techniques, they achieved\footnote{We use tilde notation to suppress polylogarithmic factors.} $\widetilde{O}(n^{1/3})$ query complexity for both properties, demonstrating a polynomial quantum speedup over the classical lower bound of $\Omega(n^{1/2})$ \cite{goldreich1997property}. They also proved a quantum lower bound of $\Omega(n^{1/4})$ for expansion testing, ruling out exponential speedup in this case. More recently,  \cite{Ben_David_2020} achieved a breakthrough by demonstrating exponential quantum speedups for certain graph property testing problems in the adjacency list model in bounded-degree graphs, showcasing the power of quantum algorithms in this domain. On the other hand, Apers, Magniez, Sen and Szabó \cite{apers2024quantumpropertytestingsparse} showed that testing $3$-colorability in the same model requires a linear number of quantum queries, giving an explicit, natural graph property for which quantum testing still requires a linear number of queries.

Furthermore, \cite{apers2024quantumpropertytestingsparse} initiated the study of quantum property testing for \emph{bounded-out-degree directed graphs in the unidirectional model}, focusing on the $k$-star-freeness property (and, more generally, $k$-source-subgraph-freeness); see \Cref{def:kstarfree} and the discussion below. They established an upper bound of $O(n^{1/2 - 1/(2(2^{k} - 1))})$ and a lower bound of $\widetilde{\Omega}(n^{1/2 - 1/(2k)})$ on the quantum query complexity of testing $k$-source-subgraph-freeness for directed graphs with bounded maximum out-degree (but possibly  unbounded maximum in-degree). This demonstrates an almost quadratic quantum speedup compared to the classical lower bound of $\Omega(n^{1 - 1/k})$ queries \cite{peng2023optimalseparationpropertytesting}.

In this paper, we focus on \emph{quantum testers for bounded-degree directed graphs (digraphs)}, %
where both in- and out-degree are bounded. In the classical setting, two natural models for accessing such graphs were introduced by Bender and Ron \cite{bender2002testing}: the \emph{bidirectional model}, which allows querying both outgoing and incoming edges of a vertex, and the more restrictive \emph{unidirectional model}, which permits querying only the outgoing edges (see, e.g., \cite{bender2002testing, goldreich2010introduction, hellweg2012property}). The bidirectional model closely resembles the standard adjacency list model for undirected graphs and is at least as powerful as the unidirectional model. However, the unidirectional model is often more natural in real-world applications -- such as web crawling, social networks, or recommendation systems -- where querying incoming neighbors may not be feasible. Algorithmically, the unidirectional model is also more challenging. For instance, strong connectivity can be tested using $\widetilde{O}(1/\varepsilon)$ queries in the bidirectional model, but requires $\Omega(\sqrt{n})$ queries for two-sided error testing\footnote{A tester for a property $P$ has \emph{one-sided error} if it always accepts (di)graphs that satisfy $P$ and may err only when the (di)graph is $\varepsilon$-far from $P$. It has \emph{two-sided error} if it may err in both cases.} \cite{bender2002testing} and $\Omega(n)$ queries for one-sided error testing \cite{goldreich2010introduction, hellweg2012property} in the unidirectional model. 

Czumaj, Peng and Sohler \cite{10.1145/2897518.2897575} advanced the study of property testing in bounded-degree digraphs by providing a generic transformation between the bidirectional and unidirectional models. They showed that any property $\Pi$ that can be tested with a constant number of queries (i.e., $O_{\varepsilon,d}(1)$)\footnote{Throughout the paper, we use $O_{\varepsilon,d}(\cdot)$ and $\Omega_{\varepsilon,d}(\cdot)$ to denote bounds that treat $\varepsilon$ and $d$ as constants.} in the bidirectional model can also be tested in the unidirectional model using a sublinear number of queries, specifically $n^{1 - \Omega_{\varepsilon,d}(1)}$. This result was later shown to be essentially tight by Peng and Wang \cite{peng2023optimalseparationpropertytesting}, who constructed a property $P$ that can be tested with $O_{\varepsilon,d}(1)$ queries in the bidirectional model, but requires at least $n^{1 - f(\varepsilon,d)}$ queries in the unidirectional model, where $f(\varepsilon,d) \to 0$ as $\varepsilon \to 0$.

In this work, we investigate whether a similar transformation exists in the setting of quantum property testing, and if so, whether the resulting quantum testers in the unidirectional model can achieve an almost quadratic speedup over their classical counterparts. %

\subsection{Our contributions}
\paragraph{The upper bound} %
Our main contribution is a generic transformation from constant-query quantum
testers in the bidirectional model to \(o(\sqrt n)\)-query quantum testers in the
unidirectional model. Recall that a digraph is \emph{\(d\)-bounded-degree} if both
its maximum in-degree and maximum out-degree are at most the fixed constant \(d\). %
Formally, we establish the following result. 

\begin{restatable}{theorem}{QTester2}\label{thm:main}
    If a graph property is $\eps$-testable with two-sided error and quantum query complexity  $O_{\eps,d}(1)$ in the $d$-bounded-degree digraph bidirectional model, then it is also $\eps$-testable with two-sided error and quantum query complexity $n^{1/2-\Omega_{\eps,d}(1)}$ in the $d$-bounded-degree digraph unidirectional model.
\end{restatable}

The proof consists of two components. The first, which may be of independent
interest, shows that quantum access provides no additional testing power at
constant query complexity in the bidirectional bounded-degree model. In the
adjacency-matrix model, a direct consequence of {\cite[Corollary 31]{Ben_David_2020}} is that the
class of graph properties testable with constantly many classical queries
coincides with the class testable with constantly many quantum queries. We
establish an analogous equivalence in the bidirectional model for bounded-degree digraphs.

\begin{restatable}{theorem}{QBidClassical}\label{thm:qbd-classical}
Fix \(d\) and \(\eps>0\). Let \(\Pi=(\Pi_n)_{n\ge 1}\) be a \(d\)-bounded-degree digraph property. If \(\Pi\) is
\(\eps\)-testable with two-sided error by a
\(T=T(\eps,d)=O_{\eps,d}(1)\)-query quantum tester in the bidirectional model,
then \(\Pi\) is also \(\eps\)-testable with two-sided error by a classical
randomized tester using \(O_{\eps,d}(1)\) bidirectional queries. Consequently,
at constant query complexity, quantum and classical testability coincide in the
bidirectional model for bounded-degree digraphs.
\end{restatable}

Both testers in \Cref{thm:qbd-classical} may be adaptive. %
The proof shows that, after randomizing vertex names and adjacency-list orderings, a constant-query quantum tester has nearly the same acceptance probability on any two graphs with nearly the same distribution of constant-radius neighborhoods. Consequently, the tester's acceptance gap forces these neighborhood distributions to differ by a constant between graphs in the property and graphs that are $\eps$-far from it. A classical tester can detect this difference by sampling and exploring constantly many neighborhoods.

The proof of \Cref{thm:qbd-classical} uses direct classical exploration of bidirectional discs and is independent of the outgoing-query estimator. The second component combines that estimator with a canonical testing reduction to transform classical bidirectional testers into quantum unidirectional testers.

\begin{restatable}{theorem}{QTester}\label{thm:main-classical-quantum}
    If a graph property is $\eps$-testable with two-sided error and classical query complexity  $O_{\eps,d}(1)$ in the $d$-bounded-degree digraph bidirectional model, then it is also $\eps$-testable with two-sided error and quantum query complexity $n^{1/2-\Omega_{\eps,d}(1)}$ in the $d$-bounded-degree digraph unidirectional model.
\end{restatable}

Note that \Cref{thm:main} directly follows from \Cref{thm:qbd-classical} and \Cref{thm:main-classical-quantum}. Furthermore, 
\Cref{thm:main-classical-quantum} already answers affirmatively an open question proposed in \cite{apers2024quantumpropertytestingsparse}, showing that our transformation yields an almost quadratic quantum speedup over the best known generic classical transformation to the unidirectional model~\cite{10.1145/2897518.2897575}.  
We remark that many natural graph properties can be tested with two-sided error using a constant number of classical queries in the bounded-degree bidirectional model. Examples include:
subgraph-freeness \cite{goldreich1997property},
any hyperfinite digraph property \cite[Theorem 1.3]{10.1145/2897518.2897575},
strong connectivity and Eulerianity \cite{bender2002testing, orenstein2011testing},
$k$-(vertex/edge)-connectivity  \cite{forster2020computing}, 
and certain first-order properties \cite{adler2024testability}.
Our transformation in \Cref{thm:main-classical-quantum} implies that all of these properties can be tested in the unidirectional model in the quantum setting with query complexity $n^{1/2 - \Omega_{\varepsilon,d}(1)}$.

As a byproduct of \Cref{thm:main-classical-quantum}, we obtain a quantum algorithm for approximating the number $\#H$ %
of occurrences of any constant-size connected subgraph $H$ (referred to as $H$-instances) in a digraph $G$ with $o(\sqrt{n})$ quantum query complexity in the unidirectional model. Specifically, for any $d$-bounded-degree digraph $H$ of radius at most $r$, we can approximate the number $\#H$ within an additive error $\delta n$ using $O_{\delta,d,r}\left(n^{1/2-2^{-\Theta((2d)^{r+1})}}\right)$ quantum queries in the unidirectional model. We refer to \Cref{thm:subgraph} for the formal statement.

Furthermore, we remark that if $q = q(\varepsilon, d)$ is the classical bidirectional query complexity, our transformation from \Cref{thm:main-classical-quantum} yields quantum unidirectional complexity $n^{1/2 - f(\varepsilon,d)}$, where $f(\varepsilon,d) = 2^{-\Theta((2d)^{q+1})}$ %
(see \Cref{thm:full-main}). Since any non-trivial property requires $\Omega(1/\varepsilon)$ queries \cite{fischer2024basic}, we have $q \to \infty$ as $\varepsilon \to 0$, and thus $f(\varepsilon,d) \to 0$. 

\paragraph{The lower bound} A natural question is whether our transformation is tight. In particular, can the upper bound of the resulting tester in the quantum unidirectional model be improved to $n^{c-f'(\eps,d)}$ for some universal constant $c<1/2$ (and a function $f'$ satisfying $f'\to 0$ as $\eps\to 0$)? Prior to our work, the best known lower bound in our setting was $\Omega(n^{1/3})$. This bound follows from the standard quantum collision lower bound \cite{ambainis2005polynomial} by a direct reduction to $2$-star-freeness. (For completeness, we reproduce its proof in \Cref{sec:BKLB}.) %
We also note that the aforementioned lower bound from \cite{apers2024quantumpropertytestingsparse} applies only to graphs with bounded maximum out-degree but potentially unbounded in-degree, and thus does not cover our setting. %

We show that our transformation is indeed almost tight by constructing a property that is constant-query testable in the classical bidirectional model, yet requires $\widetilde{\Omega}_\eps(n^{1/2 - f'(\varepsilon)})$ queries in the quantum unidirectional model, for some function $f'(\varepsilon)$ that tends to $0$ as $\varepsilon \to 0$.

\begin{theorem}
\label{thm:tight}
For any sufficiently small constant $\eps>0$, there exists a constant $d=O(\log\log(1/\eps))$ and a digraph property $P=P_{\eps}$ for $d$-bounded-degree digraphs such that $P$ is $\eps$-testable with $O_{\eps}(1)$ queries in the classical bidirectional model, while any $\eps$-tester for $P$ in the quantum unidirectional model requires $\widetilde{\Omega}_\eps(n^{1/2-f'(\eps)})$ queries, where $f'(\eps)$ is a function satisfying that $f'(\eps)\to 0$ as $\eps\to 0$.
\end{theorem}

The above theorem follows directly from the next result concerning the testing of \emph{$k$-star-freeness} (which was also used to establish the classical lower bound in \cite{peng2023optimalseparationpropertytesting}). A \emph{$k$-star} is a directed graph on $k + 1$ vertices and $k$ edges, consisting of a central vertex and $k$ source vertices, each having an edge directed toward the center.
A digraph $G$ is said to be $k$-star-free if it does not contain a $k$-star as a subgraph. The following theorem extends the quantum lower bound of \cite{apers2024quantumpropertytestingsparse} to digraphs in which both in-degree and out-degree are bounded, strengthening their result that applied only to the more relaxed regime where only the maximum out-degree was bounded\footnote{We remark that the lower bounds in these two regimes can differ significantly. For instance, in the classical setting, \cite{10.1145/2897518.2897575} showed that in the unidirectional model, a certain property (called $\delta$-incoming, i.e. the fraction of vertices with non-zero in-degree is at least $\delta$) requires $n^{1-O(\sqrt{\log\log n/\log n})}$ queries when only the maximum out-degree is bounded, whereas the same property can be tested with $O(n^{1-1/d})$ queries (by estimating the number of $k$-stars for each $k$ such that $d\geq k \geq 1$) when both the maximum out-degree and in-degree are bounded by $d$.}.

\begin{restatable}{theorem}{Kstar}
\label{thm:k_star_lowerbound}
Let constants $ k\geq 2$ and $0<\eps< 1/(32k)^{k/2+1}$.
Then the quantum query complexity of $\eps$-testing $k$-star-freeness in $k$-bounded-degree digraphs is $\Omega_{k,\eps}(n^{1/2-1/(2k)}/\ln^4 n)$ in the unidirectional model.
\end{restatable}

For each fixed $k$, the lower-bound exponent is $1/2-1/(2k)$.
It approaches $1/2$ across a family in which the degree bound and admissible
proximity parameter vary. Thus the near-optimality statement concerns a
generic transformation covering all fixed degree bounds; it does not assert
an $\Omega(\sqrt n)$ lower bound at one fixed degree or match every parameter
of the estimator.

It is known that $k$-star-freeness is $\eps$-testable in the classical bidirectional model using $O_k(1/\varepsilon)$ queries on $k$-bounded-degree digraphs, by sampling $O_k(1/\varepsilon)$ vertices, querying all their in- and out-neighbors, and rejecting if a $k$-star is detected (see, e.g.  \cite{goldreich1997property}).
We now proceed to prove \Cref{thm:tight}, assuming the validity of  \Cref{thm:k_star_lowerbound}.

\begin{proof}[Proof of \Cref{thm:tight}]
For sufficiently small $\eps>0$, take
$k=d=\lceil\log\log(1/\eps)\rceil\geq2$ and let $P_\eps$ be
$k$-star-freeness. Since
$(k/2+1)\log(32k)=o(\log(1/\eps))$, the condition
$\eps<(32k)^{-k/2-1}$ holds. Apply \Cref{thm:k_star_lowerbound} and
set $f'(\eps)=1/(2k)$, which tends to zero. The classical bound above
is constant for each fixed $\eps$, proving the theorem.
\end{proof}

\subsection{Technical overview}
\subsubsection{The upper bound} 
We first explain \Cref{thm:qbd-classical}: every property admitting a
constant-query quantum tester in the bidirectional bounded-degree model also
admits a constant-query classical tester. The key is a uniform locality bound:
after randomizing the graph representation, two graphs with similar
constant-radius neighborhood distributions have similar acceptance probabilities
for \emph{every} constant-query quantum algorithm.

Fix a \(T\)-query tester \(\mathcal A\), where
\(T=O_{\varepsilon,d}(1)\); the tester may be adaptive. By deferring intermediate measurements, we may write
it as a sequential query circuit
\[
 |\psi_{\mathsf O}\rangle
 =V_TU_{\mathsf O}V_{T-1}\cdots V_1U_{\mathsf O}V_0|\psi_0\rangle,
\]
where the \(V_t\)'s contain arbitrary oracle-independent computation and
workspace, and \(U_{\mathsf O}\) is a bidirectional query. 
We average over a uniformly random vertex relabeling and independently randomize
the valid orderings of all incoming and outgoing adjacency lists. This preserves
the tester's acceptance gap. Let \(\Sym(G)\) denote the resulting distribution
over oracle representations and let
\[
 p_{\mathcal A}(G)
 :=
 \mathbb E_{\mathsf O\sim\Sym(G)}
 \langle\psi_{\mathsf O}|E|\psi_{\mathsf O}\rangle
\]
be the averaged acceptance probability, where
$0\le E\le I$ describes the accepting measurement outcome.

\paragraph{From the quantum computation to small diagrams.}
Expanding the \(T\)-query computation together with its complex conjugate shows
that \(p_{\mathcal A}(G)\) is a linear functional of the doubled oracle moment
\[
 \mathcal M_T(\Sym(G))
 :=
 \mathbb E_{\mathsf O\sim\Sym(G)}
 \left[
 U_{\mathsf O}^{\otimes T}
 \otimes
 \overline{U_{\mathsf O}}^{\otimes T}
 \right].
\]
More precisely,
\[
 p_{\mathcal A}(G)
 =
 \langle\mathcal A,\mathcal M_T(\Sym(G))\rangle.
\]
Here \(\langle\mathcal A,\cdot\rangle\) denotes the linear functional determined
by the initial state, the oracle-independent unitaries \(V_0,\ldots,V_T\), and
the accepting measurement \(E\). Intuitively, the tensor
\(U_{\mathsf O}^{\otimes T}\otimes\overline{U_{\mathsf O}}^{\otimes T}\)
collects the \(T\) oracle factors from the ket and the \(T\) conjugated oracle
factors from the bra in the acceptance-probability expansion, while all
algorithm-dependent factors are absorbed into \(\langle\mathcal A,\cdot\rangle\).
Thus the sequential nature of the quantum algorithm is fully retained; the
moment notation only packages the oracle dependence of its acceptance
probability. Each term imposes at most \(2T\) query--answer requirements, such as
``the \(i\)th outgoing neighbor of \(x\) is \(y\).'' Forgetting the actual
vertex names while retaining their equality relations, directions,
adjacency-list positions, and \(\bot\)-answers leaves only constantly many
\emph{doubled query--answer diagrams} for fixed \(d,T\).

For a diagram \(D\) with \(m\) abstract labels and \(q\) connected components,
choose one root label in each component. We sum the corresponding elementary
tensors over injective assignments of vertex names and divide by
\((n-q)_{m-q}\), the number of ways to assign the \(m-q\) non-root labels once
the \(q\) root names are fixed. This defines the
\emph{root-normalized diagram tensor} \(\mathcal U_D^{(n)}\): the non-root
labels are averaged, while the root labels remain summed.

By relabeling symmetry, for any fixed injective assignment of the abstract
labels, the probability that the query--answer requirements of \(D\) are
satisfied is the same. Denoting this common probability by \(p_D(G)\), the
normalization above gives
\[
 \mathcal M_T(\Sym(G))
 =
 \sum_D c_D(G)\,\mathcal U_D^{(n)},
 \qquad
 c_D(G):=(n-q)_{m-q}\,p_D(G).
\]
Thus the tensors \(\mathcal U_D^{(n)}\) depend only on \(D\) and \(n\), while
all dependence on the graph \(G\) is contained in the scalar coefficients
\(c_D(G)\).

\paragraph{Uniformly bounding quantum interference.}
The main technical issue is that small errors in individual matrix entries need
not imply a small error in acceptance probability: a quantum algorithm may add
many contributions coherently. We therefore use the operational acceptance
seminorm
\[
 \|Z\|_{\str,T}
 :=
 \sup_{\mathcal A}|\langle\mathcal A,Z\rangle|,
\]
where the supremum ranges over all \(T\)-query algorithms with arbitrary
workspace. We prove
\[
 \|\mathcal U_D^{(n)}\|_{\str,T}\le K_T
\]
uniformly in \(n\).

The key idea is an orthogonality saving for each connected component. We
factor every partial query into an input query projection, an answer shift, and
an output query projection, and use temporary registers to ensure that repeated
occurrences of the same abstract label receive the same vertex name. Summing
over one label costs at most \(n\), while a first or last occurrence that is a
query projection saves a factor \(\sqrt n\), because projections onto distinct
vertex names are orthogonal. In every connected component, the first and last
occurrences are query projections, yielding a full factor-\(n\) saving per
component. Hence the unnormalized sum over vertex assignments has norm at most
\(n^{m-q}\). Inclusion--exclusion enforces distinct vertex names at only a
\(T\)-dependent constant loss. Finally, the root-normalized tensor divides this sum by
\((n-q)_{m-q}\). Since \(m,q=O(T)\),
$
 \frac{n^{m-q}}{(n-q)_{m-q}}=O_T(1)
$ 
for all sufficiently large \(n\). Dividing by \((n-q)_{m-q}\) therefore gives
the required bound independent of \(n\): $
 \|\mathcal U_D^{(n)}\|_{\str,T}\le K_T.
$

\paragraph{Recovering the coefficients from local neighborhoods.}
The coefficient \(c_D(G)\) has a local interpretation. Choose distinct roots
uniformly, then count the injective completions of the diagram, weighting each
completion by the probability that the randomized adjacency-list positions
satisfy the requirements of \(D\). Bounded degree makes this weighted count
uniformly bounded.

For \(r=4T+3\), the radius-\(r\) disc around each root determines every
possible realization of its component, including the relevant degree and
\(\bot\)-constraints. The only errors arise when different root neighborhoods
interact or from replacing sampling without replacement by independent root
sampling; both contribute \(O_{d,T}(1/n)\). Let \(\cD_{d,r}\) be the finite set of rooted radius-\(r\) disc types, where
distance is measured in the underlying undirected graph while edge directions
are retained, and let \(\cnt_\Gamma(G)\) be the number of vertices whose rooted
disc has type \(\Gamma\). Then
\[
 \freq_r(G)
 :=
 \frac1n\bigl(\cnt_\Gamma(G)\bigr)_{\Gamma\in\cD_{d,r}}
\]
is the distribution of the rooted disc seen from a uniformly random vertex.
We obtain
\[
 c_D(G)=P_D(\freq_r(G))+O_{d,T}(1/n),
\]
where \(P_D\) is a fixed polynomial depending only on \(D\). Since \(P_D\) is
Lipschitz on the probability simplex, the coefficient estimate together with
the uniform diagram bound yields
\[
 |p_{\mathcal A}(G)-p_{\mathcal A}(H)|
 \le
 O_{d,T}(1)\|\freq_r(G)-\freq_r(H)\|_1
 +
 O_{d,T}(1/n)
\]
for every pair of \(n\)-vertex \(d\)-bounded digraphs and every \(T\)-query
algorithm.

\paragraph{Constructing the classical tester.}
A property tester has an acceptance gap of at least \(1/3\) between
yes-instances and \(\varepsilon\)-far instances. For sufficiently large \(n\),
the locality bound therefore forces their disc-frequency vectors to be
separated by a constant in \(\ell_1\) distance. A classical tester can detect
this separation by sampling constantly many vertices and exploring their
radius-\(r\) neighborhoods in both directions, including the boundary queries
needed to recover the induced discs. It accepts when the empirical frequency
vector is sufficiently close to that of some \(n\)-vertex graph in the
property. Smaller graphs can be read in full.

The final decision rule may use a nonuniform lookup table, since the theorem
bounds query complexity rather than running time. We also stress that \(T\) is
fixed throughout; the constants obtained above do not yield a meaningful bound
when \(T=T(n)\) grows with \(n\).

\vspace{1em}

We now turn to \Cref{thm:main-classical-quantum}. 

\paragraph{The CPS approach}
The framework of Czumaj, Peng and Sohler~\cite{10.1145/2897518.2897575}, which we refer to as the CPS approach, reduces testing a property $\Pi$ with classical bidirectional query complexity $q=q(\varepsilon,d)$ to approximating the disc-frequency vector $\freq_r(G)$ for a radius $r=O(q)$ depending only on $\varepsilon,d$. Since $d$ and $r$ are constants, it is enough to estimate the count $\cnt_\Gamma(G)$ of every nonempty disc type $\Gamma\in\cD_{d,r}^*$ within additive error $\delta(\varepsilon,d)n$. Here nonempty means that the disc has at least one edge, i.e., $\cD_{d,r}^*$ collects all types in $\cD_{d,r}$ but the empty type. The count of the empty type is then determined from the fact that all type counts sum to $n$. Once such estimates are available, the canonical bidirectional tester can be simulated from the approximate frequency vector without further access to the graph.

In the classical unidirectional model, CPS sample $n^{1-\Omega_{d,r}(1)}$ edges and count sampled instances of type $\Gamma$. The resulting count is biased: revealing only some outgoing edges can make a rooted neighborhood look like $\Gamma$ although the full disc is a larger type $\Gamma'\succeq\Gamma$, where $\succeq$ is the rooted-subgraph order. Thus a sampled $\Gamma$-instance may be either genuine or a false appearance contributed by any $\Gamma'$ above it in this order. CPS exploit that the conditional probability of observing $\Gamma$ given the true type $\Gamma'$ is determined by the two types and the sampling probability, rather than by the rest of the input graph. Because their sampling is non-adaptive, all such contributions can be assembled into one constant-dimensional linear system and corrected in a single final step.

\paragraph{Our approach}
We retain the CPS reduction to local frequencies but replace their non-adaptive sampling by an adaptive quantum construction. For each $\Gamma\in\cD_{d,r}^*$, the algorithm estimates $\cnt_\Gamma(G)$ by iteratively applying quantum counting and Grover search. Rather than sampling a complete local view in one shot, it first creates many small partial instances and then searches only for edges capable of extending them. The adaptivity is essential: the marked set at stage $i$ is defined from the instances obtained at stage $i-1$. This table-then-search structure yields an almost quadratic improvement over the CPS query bound, but also requires a stage-by-stage analysis of false appearances instead of a one-shot correction.

\paragraph{The high-level idea}
Let $m=m_{d,r}=2d\sum_{j=0}^r(2d)^j$ be the edge bound for types in $\cD_{d,r}$ used in the preliminaries. The algorithm proceeds through $m_{d,r}$ stages. At stage $i$, it stores rooted subgraphs with $i$ edges and samples edges that extend the subgraphs stored at stage $i-1$. The sample size must be small enough to save queries, but large enough that every frequent type is represented and the contributions of larger types can be estimated. Moreover, sampled instances should usually be far apart, so that one graph edge does not simultaneously extend several stored prefixes; this dispersiveness makes extension counts unambiguous and lets us relate independent sample counts to distinct stored instances. The stage sizes are therefore chosen recursively, balancing the cost of Grover search against the number of prefixes available from the preceding stage. This generalizes the adaptive $k$-star procedure of~\cite{apers2024quantumpropertytestingsparse,liu2019findingquantummulticollisions}: those algorithms only detect an abundant collision pattern, whereas we must estimate all constant-radius disc counts to additive error and separate genuine occurrences from every possible larger type.

\paragraph{The algorithm}
At stage $i$, the algorithm processes every $i$-edge type $\Gamma$ in a fixed order. 
We fix consistent prefix chains that start at the isolated root $\Gamma_{[0]}$
and append one edge at each step; $\Gamma_{[i-1]}$ is the chosen
$(i-1)$-edge prefix of $\Gamma$. A $\Gamma$-instance is an actual rooted
subgraph of type $\Gamma$, not necessarily a full disc. Initially,
$\cH_{\Gamma_{[0]}}$ contains one isolated root at each vertex.
Given the stored collection $\cH_{\Gamma_{[i-1]}}$, we define a Boolean function $f_\Gamma$ on the $dn$ outgoing slots $(v,s)$ of $G$, where $s$ is a position in $v$'s outgoing-neighbor list: a slot is marked if its non-null edge is absent from some stored prefix and extends that prefix to a $\Gamma$-instance. Membership in the marked set can be checked using a constant number of outgoing-neighbor queries together with the classically stored prefix. Let
\[
    X_\Gamma=\abs{f_\Gamma^{-1}(1)}.
\]
Quantum counting produces an estimate $\tilde X_\Gamma$. If the threshold
check below passes and $i<m$, repeated Grover search draws
\[
    \ell_\Gamma=\lfloor{t_i\tilde X_\Gamma}/{t_{i-1}}\rfloor,
    \qquad
    t_i=n^{(2^{m-i}-1)/(2^{m}-1)}
\]
marked edges. At the terminal stage $i=m$ we only count. Combining the samples with their prefixes gives the collection $\cH_\Gamma$ used by later stages. On a \textit{good}  history, an accurate count passing the threshold guarantees $X_\Gamma>0$. We will specify what constitutes a good history later.
Conditioned on this history and count, fresh Grover calls give independent uniform samples from $f_\Gamma^{-1}(1)$ with replacement. Each instance retains its sample index and edge-revelation history, whose successive prefix types follow the chosen chain.
Notice $t_0=n$ and $t_{m}=1$.
The scales $t_i$ decrease with $i$ and satisfy $t_i=o(\sqrt{t_{i-1}})$. Because of the well-known birthday paradox, this relationship ensures a separation property at each stage: at the first stage, two sampled roots are unlikely to lie within weak distance $2r$, and at later stages, two sampled edges are unlikely to extend the same stored prefix. These scales make the total Grover-sampling cost have the same leading exponent at every nonterminal stage, while keeping every counting step within that bound. After all stages, normalized versions of the estimates $\tilde X_\Gamma$ are converted into disc-count estimates by an upper-triangular linear system described below. Notice that $\tilde X_\Gamma$ itself does not estimate $\cnt_\Gamma$ directly: it counts all currently available ordered extensions, including those arising from larger true discs.

\paragraph{Warm-up: estimating the number of vertices with in-degree $k$}
For intuition, let $\cnt_j$ count the vertices of in-degree exactly $j$.
A vertex of in-degree $j$ contributes $\binom ji i!$ ordered incoming
$i$-stars: we choose $i$ edges and then their revelation order.
The warm-up estimator therefore forms the observations
\[
 \tilde x_i=\frac{n\tilde X_i}{i!t_{i-1}}
 \approx\frac{n X_i}{i!t_{i-1}} = x_i \approx \sum_{j\geq i}\binom ji\cnt_j\qquad(1\leq i\leq d).
\]
Here $\tilde X_i$ is the counting output for $X_i$, the marked incoming-star
extensions at stage $i$. The warm-up uses the same stage-scale formula with $d$ in place
of $m$, under the assumption that every positive in-degree class has
density at least $\delta$. The two approximations hold when counting and
sampling satisfy their accuracy and distinctness conditions respectively. The
matrix with entries $\binom ji$ is upper triangular with diagonal
entries one, so we can invert it to remove the contributions from
larger in-degrees. For general discs, we use the same principle with
coefficients that also account for the possible construction histories
and symmetries.

\paragraph{Correcting appearances in general discs}
For $\Gamma\preceq\Gamma'$, let $\mu_{\Gamma,\Gamma'}$ count the
ordered edge histories inside $\Gamma'$ that follow $\Gamma$'s chosen
prefix chain. These are the nodes at level $i=|E(\Gamma)|$ of the
tuple tree in \Cref{fig:tree}; each step appends one eligible edge.
Here $\cW_j$ denotes the set of ordered histories realizing $\Gamma_{[j]}$ inside $\Gamma'$; each history records the root and the edges in revelation order.
We group histories into classes $[W]$ under automorphisms of $\Gamma'$,
meaning vertex renamings that preserve its root and directed edges.
For classes at consecutive levels, $\kappa_{[W],[W']}$ counts the
extensions from any history in the first class into the second.
Symmetry makes this number independent of the representative.
These classes let us track contributions from each true disc type
without requiring the algorithm to identify that unknown type.

\begin{figure}[htbp]
    \centering
    \includegraphics[width=1.0\linewidth]{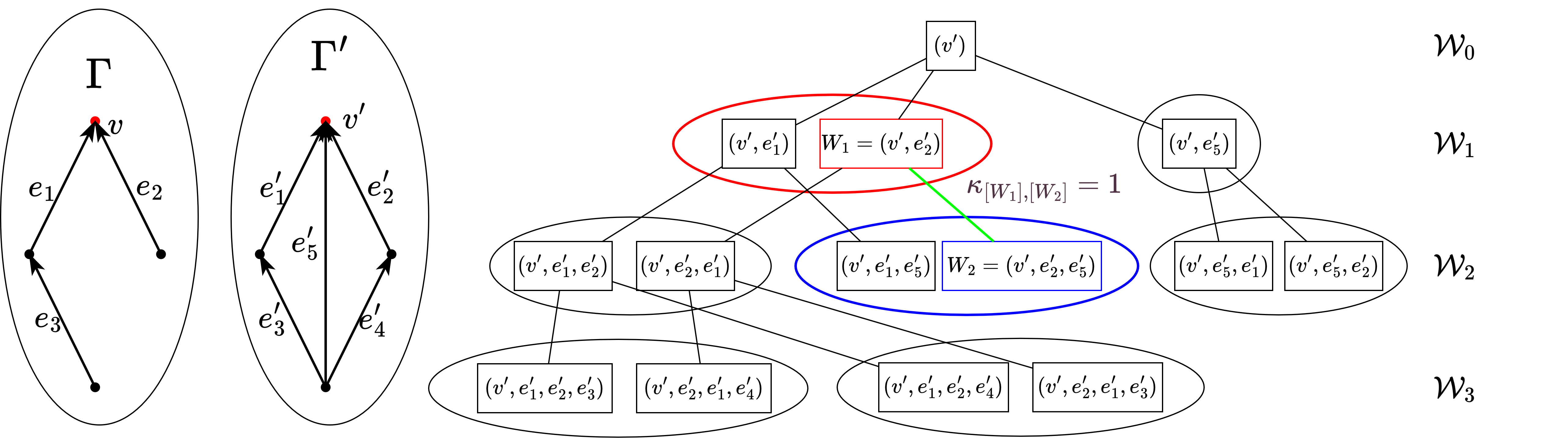}
    \caption{The tuple tree for $\Gamma\preceq\Gamma'$. Level $j$
    contains the ordered realizations of the prefix $\Gamma_{[j]}$
    inside $\Gamma'$, and circles indicate automorphism classes.
    The highlighted transition has $\kappa_{[W_1],[W_2]}=1$:
    each history in the first class has exactly one eligible edge
    leading to the second class. The red vertices are the roots.}
    \label{fig:tree}
\end{figure}

We call the execution history before the current type \textit{good} if all previously completed counting steps satisfy their accuracy bounds, and every previously completed sampling step satisfies the prescribed relative and additive error bounds for the sample count in each history class and has roots at pairwise weak distance greater than $2r$ for distinct sample indices within that step. We analyze each type conditional on such a history.
The separation property makes the sampled discs disjoint and, after
the initial stage, gives each marked edge a unique stored parent.
It also lets us identify sample counts, initially taken with
multiplicity, with counts of distinct stored instances.
Conditioning on the stored prefixes and current accurate count,
we apply concentration to the independent uniform samples and combine
its failure probability with that of separation by a union bound.
Along a branch whose types all meet the frequent cutoffs defined below,
the sampling ratio is close to $t_i/t_{i-1}$,
with rounding absorbed into the allowed error. Iterating the class-count recursion
then gives, up to controlled multiplicative and additive errors,
\[
 X_\Gamma\approx\frac{t_{i-1}}n A_\Gamma,
 \qquad
 A_\Gamma=\sum_{\Gamma'\succeq\Gamma}
                   \mu_{\Gamma,\Gamma'}\cnt_{\Gamma'}.
\]
Thus $A_\Gamma$ counts all ordered appearances in the full input,
whereas $X_\Gamma$ counts only the extensions available in the current
adaptive history. The detailed recursion appears in \Cref{sec:pfP}.

Ordering types compatibly with $\preceq$, we obtain an upper triangular
matrix $M$ with entries $M_{\Gamma,\Gamma'}=\mu_{\Gamma,\Gamma'}$
for $\Gamma\preceq\Gamma'$ and zero otherwise. Its positive diagonal
makes it invertible. We set $\tilde x_\Gamma=n\tilde X_\Gamma/t_{i-1}$ and $x_\Gamma=n X_\Gamma/t_{i-1}$ 
when the threshold passes, and leave truncated coordinates at zero.
On a good execution (i.e., every type is rare or frequent as defined
below, and the run before imposing the query cap completes all
unskipped stages with the counting, sample-count and separation conditions
above satisfied), the extension analysis and the truncation bound below
give
\begin{equation}
\label{eq:crucial}
 \tilde{\boldsymbol{x}}\approx \boldsymbol{x}\approx M\boldsymbol{\cnt},
 \qquad
 \widetilde{\boldsymbol{\cnt}}=M^{-1}\tilde{\boldsymbol{x}}.
\end{equation}
Here $\boldsymbol{\cnt}$ is the vector of true nonempty disc counts. As in the warm-up, the first approximation follows from accurate counting. The second follows from the class-count recursion, supported by concentration and separation in the preceding sampling steps.
Finally we use $\widetilde{\boldsymbol{\cnt}}$ as the estimate of $\boldsymbol{\cnt}$.
The norm $\|M^{-1}\|_1$ bounds how much inversion can amplify the total
absolute coordinate error. Because it depends only on $d,r$, we can
choose the internal accuracies to make the final total additive error
at most $\delta n$.

\paragraph{Handling the case when $\Gamma$-instances are rare}
The preceding multiplicative analysis fails when $X_\Gamma$ is very small: the additive counting error may dominate $X_\Gamma$, and Grover search may cost $\Omega(\sqrt n)$ even to find one extension. The algorithm therefore compares $\tilde X_\Gamma$ with a threshold before storing its normalized observation or running the Grover-sampling step. If it is below the threshold, the $\Gamma$-iteration and its prefix descendants, namely types $\Gamma'$ with $\Gamma'_{[i]}=\Gamma$ for $i=|E(\Gamma)|$, are truncated. These are the types whose chosen construction chain passes through $\Gamma$, and their observations remain zero. If $\Gamma$ is rare, every ordered appearance of such a descendant contains a $\Gamma$-prefix, giving $A_{\Gamma'}\leq m!A_\Gamma$. We choose the cutoffs small enough to absorb this truncation error after applying $M^{-1}$. On good histories, a safely large estimate ensures that $X_\Gamma$ is large enough for the multiplicative estimates and concentration bounds to apply.

There is a narrow gray zone in which neither conclusion is guaranteed. We remove it by parameter filtering. For each accuracy parameter $\delta_j$, we call a type frequent when $A_\Gamma\geq\alpha_\Gamma^{\delta_j}n$, rare when $A_\Gamma<\beta_\Gamma^{\delta_j}n$, and gray otherwise, where $0<\beta_\Gamma^{\delta_j}<\alpha_\Gamma^{\delta_j}$. On a good history, an accurate count passes the threshold in the frequent case and fails it in the rare case. We choose $100D$ geometrically separated values $\delta_1,\ldots,\delta_{100D}$, where $D=|\cD_{d,r}^*|$, so that for each fixed $\Gamma$ the gray intervals $[\beta_\Gamma^{\delta_j}n,\alpha_\Gamma^{\delta_j}n)$ are disjoint. Since $A_\Gamma$ depends only on the input, a type is gray for at most one parameter choice, and at most $D$ choices are bad for at least one type. Selecting $\delta_j$ uniformly therefore gives, with probability at least $0.99$, a parameter for which every type is classified either as rare or as frequent. For such a parameter, the separate counting and sampling analysis makes all truncation decisions and estimates valid simultaneously with high constant probability.

\paragraph{Query complexity}
Counting an $i$-edge type uses $O_{\delta,d,r}(\sqrt{n/t_{i-1}})$
queries. Given the preceding good history and the current accurate
count at a frequent type, each Grover sample has this conditional
expected cost, and a nonterminal sampling step takes $\Theta(t_i)$
samples. The scales balance these costs, giving
\[
 O_{\delta,d,r}\!\left(
   \sum_{i=1}^{m}(t_i+1)\sqrt{\frac n{t_{i-1}}}
 \right)
 =O_{\delta,d,r}\!\left(n^{1/2-1/(2(2^m-1))}\right).
\]
To obtain a bound on every execution, we analyze a run stopped after
the first inaccurate count or failed sampling step, including all
queries in that step. For a selected parameter with no gray types, its
expected cost is at most an explicit $L$
of the displayed order. The actual algorithm checks a global counter
before every oracle call and returns the zero vector at
$\lceil100L\rceil$. Markov's inequality gives at most $0.01$
additional failure probability. Taking the largest constants over
the finite parameter list gives a uniform deterministic query bound;
see \Cref{lem:query-threshold} and \Cref{sssec:pfDisc,sec:pfP*}.

Finally, we amplify the count estimates by independent repetitions
and coordinatewise medians, and recover the empty type count by
subtraction. For sufficiently large $n$, the canonical tester's
sampled neighborhoods overlap or have connecting edges with only
small probability. We can therefore estimate its rejection probability
from the disc counts, with an error smaller than the remaining constant
acceptance gap. \Cref{sec:test} carries out this CPS reduction.
Together with $r=O(q(\varepsilon,d))$, the query bound proves
\Cref{thm:main-classical-quantum}.

\subsubsection{The lower bound}\label{sec:techoverlowerbound}
We now outline the high-level idea behind our quantum lower bound for testing $k$-star-freeness, i.e., the proof of \Cref{thm:k_star_lowerbound}. 

We begin by reviewing the approach of \cite{apers2024quantumpropertytestingsparse}, which builds on the framework of \cite{bun2018polynomial} to establish a lower bound for testing $k$-star-freeness in bounded-out-degree digraphs. 
Their idea is to reduce testing $k$-collision-freeness to testing $k$-star-freeness in $(N+R)$-vertex bounded-out-degree digraphs. The collision problem is defined as follows: for some $R=\Theta(N)$, given query access to a function $f:[N]\to[R]$, %
determine whether $f$ is $k$-collision-free -- that is, no value in $[R]$ is mapped by \emph{at least} $k$ distinct elements of $[N]$ -- or whether $f$ is $\eps_f$-far from satisfying this property. %

For their lower bound, they restrict the far case to inputs with linearly many distinct values of multiplicity at least $k$. This restricted collision problem, denoted by $\COLL^k$, is represented through the block composition $\GapOR\circ\THR^k$ after incidence encoding. Here $\GapOR$ and $\THR^k$ serve as key building blocks for applying the dual polynomial method and the block-composition framework. In particular, $\GapOR$ admits a dual witness $\phi$ of the so-called \emph{pure high degree} (see below) $\Omega(1)$, and $\THR^k$ admits a dual witness $\psi$ of pure high degree $\Omega(N^{1/2-1/(2k)})$. %
Their dual block composition $\phi\star\psi$ then yields a dual witness for $\COLL^k$ of pure high degree $\widetilde{\Omega}(N^{1/2-1/(2k)})$ through the zeroing-out trick of \cite{bun2018polynomial}, which in turn implies the quantum lower bound of $\widetilde{\Omega}(n^{1/2-1/(2k)})$ queries for testing $k$-star-freeness, where $n=N+R=\Theta(N)$.

\paragraph{Our extension and new setting}

We extend the  setting of \cite{apers2024quantumpropertytestingsparse}, where only the maximum out-degree is bounded, to the stricter regime in which both the maximum in-degree and out-degree are bounded.
This generalization necessitates several new structural insights and a refined reduction framework.

In our approach, we study $\OCCU^k$, which corresponds to testing $k$-occurrence-freeness -- that is, for some $R=\Theta(N)$, given query access to a function $f:[N]\to[R+N]$ with the promise that no value in $[R+N]$ is mapped by more than $k$ distinct elements of $[N]$, the task is to decide whether $f$ is $k$-collision-free or $\eps_f$-far from satisfying this property.
We interpret $\OCCU^k$ as a bounded variant of $\COLL^k$, where no value in $[R+N]$ has more than $k$ pre-images, and accordingly replace $\THR^k$ with $\BTHR^k$, which restricts the Hamming weight of the input to at most $k$.

We adapt the sparse-support construction of the $\THR_N^k$ witness from \cite{bun2018polynomial, apers2024quantumpropertytestingsparse} and analyze its composition with $\phi$ under the additional bounded-promise, without constructing a separate witness for $\BTHR_N^k$.
Consequently, we can obtain the quantum lower bound for testing $k$-star-freeness in a bounded-degree digraph via bounding the pure high degree of the dual witness. 
We provide a more detailed explanation below.  %

\paragraph{Reduction from $k$-occurrence-freeness to $k$-star-freeness}
We reduce testing $k$-occurrence-freeness to testing $k$-star-freeness in bounded-degree digraphs. Here $N$ is the list length, $R$ is the number of non-dummy symbols in the intermediate problem, and the ordinary function has alphabet $[R+N]$. We take $N=c_kR$, where $c_k$ is the constant depending only on $k$ in \Cref{thm:lowerbound}, and choose $\eps_f$ in the range of that theorem. For $f:[N]\to[R+N]$, the graph has $2N+R$ vertices: one per list position and one per possible value, with an edge from position $i$ to value $f(i)$. This property mirrors $k$-collision-freeness but imposes an additional constraint, resulting in two promises:
\begin{itemize}
    \item \textit{gap-promise}: either $f$ has no $k$-collisions, or at least $\eps_f N$ function values must be modified to eliminate all $k$-collisions;
    \item \textit{bounded-promise}: no value in $[R+N]$ is mapped by more than $k$ elements in $[N]$, i.e., $f$ has no $(k+1)$-collisions.
\end{itemize}
The second promise captures the stricter bounded in-degree condition inherent to our graph setting. Each value occurring $k$ times gives an incoming $k$-star, and stars for different values are edge-disjoint. The parameter $\eps_f$ measures the fraction of list positions that must change; write $\eps_g$ for the graph proximity parameter, where distance is normalized by $k$ times the vertex count. We choose a fixed
\[
 \eps_g<\frac{\eps_f c_k}{k(2c_k+1)}.
\]
For every sufficiently large target graph size $n$, take $R=\lfloor n/(2c_k+1)\rfloor$ and $N=c_kR$, then pad the $2N+R$-vertex graph with isolated vertices. An $\eps_f$-far function gives graph distance at least $\eps_fN/(kn)$, which tends to $\eps_fc_k/(k(2c_k+1))$ and therefore exceeds $\eps_g$ for sufficiently large $n$. Each outgoing graph query uses at most two coherent queries to $f$; see \Cref{obs:reduction}.

\paragraph{Polynomial method and block composition}
By the polynomial method \cite{beals_poly}, the \emph{approximate bounded degree} of the function $\OCCU_{N,R+N}^{k,\eps_f}$ (that distinguishes whether a function is $k$-occurrence free or $\eps_f$-far from $k$-occurrence-freeness) yields a value-query lower bound up to a logarithmic factor from binary encoding (\Cref{lem:poly_method}). Here the approximating polynomial must also remain bounded outside the promised domain.
Using standard reductions for non-Boolean functions (\cite{Aaronson02, ambainis2005polynomial}), we introduce the dummy-augmented variant $\dOCCU_{N,R}^{k,\eps_f}$, which uses the alphabet $[R]_0=\{0,1,\ldots,R\}$ with $0$ ignored when counting repetitions. Its approximate bounded-degree lower bound transfers to $\OCCU_{N,R+N}^{k,\eps_f}$ up to a logarithmic loss (\Cref{lem:dummy_augment}). For a list $s\in[R]_0^N$, define the incidence block $x_j(s)\in\bits^N$ by setting its $i$-th coordinate to $-1$ if $s_i=j$ and to $+1$ otherwise. Thus its Hamming weight, the number of $-1$ coordinates, counts occurrences of $j$, and on valid lists we have the \emph{block-composition} identity
\[
 \dOCCU_{N,R}^{k,\eps_f}(s)
 = (\GapOR_R^{\eps_f N}\circ\BTHR_N^k)(x_1(s),\ldots,x_R(s)).
\]
Here, $\GapOR_R^{\eps_f N}$ is a gapped $\OR$ function defined on Boolean inputs whose Hamming weight is either $0$ or at least $\eps_f N$, while $\BTHR_N^k$ is a bounded threshold function defined on Boolean inputs of Hamming weight at most $k$, determining whether the Hamming weight of an input is exactly $k$. Every list encoding has total Hamming weight at most $N$; write $F^{\leq N}$ for the restriction of the composite function to this total-weight range.

\paragraph{Dual polynomial method and dual block composition}
For a Boolean function $f$ defined on a domain $D_f$, %
\emph{a dual witness} $\psi$ is a real-valued function of unit $\ell_1$-norm that strongly correlates with $f$, i.e., with large $\sum_{x\in D_f}\psi(x)f(x)-\sum_{x\notin D_f}|\psi(x)|$ value (\Cref{def:dual_witness}).
A good dual witness also has large \emph{pure high degree}, meaning orthogonality to all polynomials below that degree (\Cref{def:phd}). This lower-bounds the corresponding approximate bounded degree, which the polynomial method converts into a query lower bound with the encoding losses above (\Cref{prop:dual_dpdeg}).
The dual block-composition method (\cite{shi2009quantum, lee2009note, sherstov2013intersection}) combines normalized $\phi$ and $\psi$ into $\phi\star\psi$. With a zero-sum inner witness, this composition preserves unit $\ell_1$-norm, and its pure high degree is at least the product of their pure high degrees. Its correlation with the composed partial function must, however, be proved separately (\Cref{def:dual_block_comp,prop:dual_block_composition}).

Consequently, we compose witnesses for $\GapOR_R^{\eps_f N}$ and $\THR_N^k$ and prove that the composition retains sufficient correlation under the additional bounded-promise.

\paragraph{Construction of dual witnesses}
Bun, Kothari and Thaler \cite{bun2018polynomial} constructed dual witnesses $\phi$ and $\psi$ for $\GapOR_R^{\eps_f N}$ and $\THR_N^k$ with pure high degrees $\Omega(1)$ and $\Omega(N^{1/2-1/(2k)})$, and \cite{apers2024quantumpropertytestingsparse} used their composition to lower-bound $\COLL^k$. In our \textit{bounded-promise} setting, however, the composition must remain valid when every block has Hamming weight at most $k$, and the previous correlation analysis alone does not establish this.
Although \cite{bun2018polynomial} also provided a general dual polynomial for $\THR_N^k$ supported on inputs of Hamming weight at most $T$, directly restricting it to $\BTHR_N^k$ by setting $T=k$ yields only constant pure high degree, which is insufficient for our purposes.

Our strategy is to use the same sparse-support construction with adjusted parameters, preserving its high degree while controlling its mass above Hamming weight $k$.
To do so, we need to (1) lower-bound the correlation between $\phi\star\psi$ and $\OR_R\circ\THR_N^k$ over the entire universe, and (2) upper-bound the $\ell_1$-mass of $\phi\star\psi$ outside the domain $\GapOR_R^{\eps_f N}\circ\BTHR_N^k$.
The first bound follows from the correlation argument in \Cref{lem:corr}, but the second one is more challenging: the domain outside $\GapOR_R^{\eps_f N}\circ\BTHR_N^k$ is strictly larger than that outside $\GapOR_R^{\eps_f N}\circ\THR_N^k$.

To overcome this, we consider two outside-promise regions:
let $D_1$ denote the inputs violating the {gap-promise} and $D_2$ those violating the {bounded-promise}, so that the outside-promise region is $D_1\cup D_2$.
Now for a Boolean function $f:D_f\to \bits$, let $D_+(\psi, f) := \{x \in D_f : \psi(x) > 0, f(x) = -1\}$, $D_-(\psi, f) := \{x \in D_f : \psi(x) < 0, f(x) = +1\}$. We identify a structural property of $\psi$: its mass on $D_+(\psi,\THR_N^k)$ is $O(1/N)$, while its mass on $D_-(\psi,\THR_N^k)$ is at most $1/2-2/4^k$. At threshold $k+1$, its mass on $D_+(\psi,\THR_N^{k+1})$ is again $O(1/N)$, while its mass on $D_-(\psi,\THR_N^{k+1})$ is at least $1/2-O(1/N)$
(\Cref{prop:psi}).
Intuitively, the {gap-promise} corresponds to the number of $k$-collisions, while the {bounded-promise} corresponds to the number of $(k+1)$-collisions. \cite{apers2024quantumpropertytestingsparse} used the first two bounds on $D_+(\psi,\THR_N^k)$ and $D_-(\psi,\THR_N^k)$ to control the mass of $\phi\star\psi$ on $D_1$.
Here, we use the latter two bounds on $D_+(\psi,\THR_N^{k+1})$ and $D_-(\psi,\THR_N^{k+1})$ to control the mass on $D_2$.
Indeed, $\phi\star\psi$ is a signed mixture of the product distributions obtained by normalizing the positive or negative part of $\psi$ in each block. The bounds above give per-block mass $O(1/N)$ at weights greater than $k$. A union bound over the $R$ blocks controls the bounded-promise violations, while the threshold-$k$ bounds and concentration control the gap-promise.

This yields our key technical result: without requiring a separate $\BTHR_N^k$ witness, the composition $\phi\star\psi$ serves as a good dual witness for $\GapOR_R^{\eps_f N}\circ\BTHR_N^k$ (\Cref{lem:corr}). Applying the zeroing-out trick removes all mass on inputs of total Hamming weight greater than $N$, with only a polylogarithmic loss in pure high degree, and gives a valid witness for $F^{\leq N}$. Its degree lower bound transfers back to occurrence testing by \Cref{lem:dp-encode,lem:dummy_augment,lem:poly_method}, establishing the desired $\widetilde{\Omega}(N^{1/2-1/(2k)})$ quantum lower bound.

\subsection{Other related work}
\label{sec:related}%

The \emph{$k$-collision finding} problem asks, given $s\in[R]^N$, for
$k$ distinct indices with the same value, or a report that none exist.
Without an additional promise, this is the search version of
\emph{$k$-distinctness}. Classically, even $k=2$ requires $\Omega(N)$
queries in the worst case. For fixed $k$, Belovs \cite{6375298} gave a
quantum algorithm using $O(N^{3/4-1/(4(2^k-1))})$ queries.
Belovs \cite[Theorem~1.1]{belovs2026tight} recently proved a matching
lower bound for alphabets of size $R=\Omega(N^2)$, improving the earlier
$\widetilde\Omega(N^{3/4-1/(4k)})$ lower bound of Mande, Thaler, and
Zhu for fixed $k\geq4$ \cite{mande2020improved}.

For a uniformly random function $f:[N]\to[R]$ in the compressing
regime $N=kR$, collision finding has classical query complexity
$\Theta_k(R^{1-1/k})$ and quantum query complexity
$\Theta_k(R^{1/2-1/(2(2^k-1))})$
\cite{liu2019findingquantummulticollisions}. These bounds use the range
size $R$; since $N=kR$ and $k$ is fixed, the same exponents apply to $N$.
The testing problem studied by Apers et al.\
\cite{apers2024quantumpropertytestingsparse} instead asks, on a general
input string, whether it is collision-free or requires at least
$\eps N$ changes to become collision-free. Disjoint collisions give
distance witnesses: eliminating $h$ collisions on pairwise disjoint
sets of positions requires at least $h$ changes.
Apers et al.\ give a quantum testing upper bound
$O(N^{1/2-1/(2(2^k-1))})$ and a lower bound
$\widetilde\Omega(N^{1/2-1/(2k)})$.

\emph{Quantum counting}, introduced by Brassard, Høyer, Mosca and Tapp \cite{Brassard_2002}, is a useful technique for estimating the number of marked items defined by a Boolean function. Formally, given a Boolean function $F: [N] \to \{0,1\}$, quantum counting can estimate $M$, the size of $F^{-1}(1)$, with additive error $\delta N$ using only $O(1/\delta)$ quantum queries. This yields a quadratic speedup compared to the $O(1/\delta^2)$ classical query complexity.
Moreover, it can also output the exact value of $M=|F^{-1}(1)|$ with high probability using $O(\sqrt{(M+1)(N-M+1)})$ quantum queries, which achieves a quantum speedup as long as $M=o(N)$.

\emph{Grover search}, introduced by Grover \cite{grover1996fastquantummechanicalalgorithm}, is one of the most widely used quantum algorithms. Given a Boolean function $F: [N] \to \{0,1\}$, it finds a marked element uniformly at random from $F^{-1}(1)$ using $O(\sqrt{N/M})$ queries in expectation, where $M = |F^{-1}(1)|>0$. This achieves a quadratic speedup over the classical $O(N/M)$ query complexity.
When the value of $M$ is unknown in advance, two common strategies are employed. One approach is the \emph{exponential search algorithm} \cite{Boyer_1998}, which gradually increases the number of iterations exponentially. The other approach leverages quantum counting \cite{Brassard_2002} to first estimate the number of marked items. Both strategies allow Grover search to be applied effectively without increasing the query complexity by more than a constant factor.

Czumaj, Fichtenberger, Peng, and Sohler \cite[Definitions 7--8 and Theorem 10]{czumaj2019testablepropertiesgeneralgraphs} gave canonical testers for all constant-query testable properties in the random-neighbor model for general graphs. Their testers sample roots and inspect \emph{$r$-bounded discs}, the random subgraphs obtained by breadth-first exploration to depth $r$, sampling $r$ neighbors independently with replacement at each explored nonisolated vertex. A corresponding exploration in a digraph can sample incoming and outgoing neighbors; extending their characterization to general digraphs is a natural direction for further work.

\subsection{Open problems}
We highlight two open problems arising from this work for future investigation. 
\begin{itemize}
    \item \textbf{Generalization to unbounded-degree digraphs:} It would be interesting to extend our transformation to general digraphs without any restriction on the degrees, or even to digraphs with unbounded maximum in-degree but bounded maximum out-degree. One possible approach is to leverage the canonical testers for all constant-query testable properties in unbounded-degree graphs given in \cite{czumaj2019testablepropertiesgeneralgraphs}, which operate by sampling and inspecting the \emph{$r$-bounded $r$-discs} (see \Cref{sec:related}). However, adapting these testers to the unidirectional model, in both classical and quantum settings, appears to require fundamentally new ideas. In particular, $r$-bounded $r$-discs are random subgraphs that may overlap significantly at high-degree vertices, and it remains unclear how to efficiently estimate their frequencies or exploit them for property testing in the unidirectional setting.
    
\item \textbf{The role of adaptivity:} As discussed earlier, our framework crucially relies on adaptive strategies, which distinguishes it from the CPS approach. For example, in the warm-up problem of counting in-degree-$k$ vertices, we use $d$ adaptive counting stages and sample between consecutive stages. This design enables the use of Grover search and quantum counting, leading to a quantum advantage. A natural and fundamental question is whether such full adaptivity is necessary to achieve these performance gains. More generally, it would be interesting to study the tradeoff between the number of adaptive rounds and the query complexity of quantum property testing algorithms for (di)graphs.

\end{itemize}

\subsection{Organization of the remainder of the paper}
\Cref{sec:pre} covers preliminary definitions, notations and useful subroutines. We give the detailed proof of the constant-query quantum-classical equivalence from \Cref{thm:qbd-classical} in \Cref{app:qbd-classical}. To give the transformation from classical bidirectional testers to quantum unidirectional testers, \Cref{sec:star} first studies a simple case: estimating the number of vertices of in-degree $k$, which illustrates the counting and sampling mechanism used in the general estimator. 
In \Cref{sec:disc}, we extend the result to estimate the number of occurrences of any $d$-bounded-degree $r$-disc isomorphism type.
In \Cref{sec:test}, we state the canonical tester and prove \Cref{thm:main-classical-quantum}. We prove our lower bound in \Cref{sec:LB}. We give supplementary tools, omitted proofs and additional results in the appendices. %

\subsection*{AI Disclosure}
The proof strategy underlying \Cref{thm:qbd-classical} was obtained through extensive interaction with ChatGPT 5.5 and 5.6 Sol; all other algorithms and proofs were developed entirely by the authors. Subsequent AI interactions were used to refine and clarify the proof and to assist with writing. The authors independently verified, developed, and simplified the AI-proposed argument and take full responsibility for its correctness, exposition, and attribution.

\section{Preliminaries}
\label{sec:pre}

\subsection{Notation and basic definitions}
Here we introduce some notation that will be used throughout the paper.

We denote by $[N]$ the set $\{1,2,\dots, N\}$, and by $[N]_0$ the set $\{0,1,2,\dots, N\}$.
We write $f:X\to Y$ for a map, $f:X\hookrightarrow Y$ for an
injective map, and $f:X\xrightarrow{\sim}Y$ for a bijection.
Given a function $f : X \to Y$, we write $f^{-1}(y)$ for the set of pre-images of $y\in Y$, that is, $f^{-1}(y)=\{x\in X:f(x)=y\}$.
We say $f$ is \emph{$k$-bounded} if, for any $y\in Y$, the pre-image $f^{-1}(y)$ has size at most $k$. 
For a multiset $T$, let $\Supp{T}$ denote the set of distinct elements in $T$.
Given a set $S$ and an integer $n \leq \abs{S}$, we use $\binom{S}{n}$ to denote the collection of $n$-element subsets of $S$, and $\perm{S}{n}$ to denote the set of $n$-tuples of distinct elements from $S$, that is,
\[
\binom{S}{n} = \{S'\subseteq S:\abs{S'} = n\}, \quad \perm{S}{n}=\{(e_1,\dots,e_n)\in S^n: e_i\neq e_j\text{ for all }i\neq j\}.
\]
For a real number $a$ and a nonnegative integer $s$, write $(a)_s=a(a-1)\cdots(a-s+1)$, with $(a)_0=1$.
For a finite nonempty set $S$, $\Unif(S)$ denotes the uniform distribution
on $S$. We write $\cB(N,p)$ for the binomial distribution with $N$ trials
and success probability $p$, and $\I[\mathcal E]$ for the indicator of
an event or condition $\mathcal E$.
For nonnegative $a,b$ and $0\le\varepsilon<1$, the notation
$x\in(1\pm\varepsilon)a\pm b$ means
$(1-\varepsilon)a-b\le x\le(1+\varepsilon)a+b$;
we use the same interval convention when either error term is absent.

For a vector $x=(x_i)$, let $\|x\|_1=\sum_i|x_i|$; a function on a
finite set is viewed as a vector indexed by its domain.
For a matrix $A$, $\|A\|_1=\max_j\sum_i|A_{ij}|$ is the induced
matrix $1$-norm. An unsubscripted norm denotes the Euclidean norm
for vectors and the induced operator norm for matrices.
Inequalities between real vectors are coordinatewise.
For Hermitian operators, $A\le B$ means that $B-A$ is positive semidefinite,
and $I$ denotes the identity operator on the relevant space.

For a complex matrix $U$, denote by $U^*$ its adjoint, i.e. conjugate transpose, and denote by $\overline{U}$ the entrywise complex conjugation.

For a simple directed graph $G(V,E)$ where the self-loops are not allowed, we denote by $n = |V|$ the number of vertices and by $m = |E|$ the number of edges in $G$.
For a vertex $v \in V$, let $N_\mathrm{in}(v)$ and $N_\out(v)$ denote its sets of in-neighbors and out-neighbors, respectively, and define its in-degree and out-degree as $\deg_\mathrm{in}(v) = |N_\mathrm{in}(v)|$ and $\deg_\out(v) = |N_\out(v)|$.
We say that \( G \) is \emph{\( d \)-bounded-degree} if for every vertex \( v \in V \), the in-degree and out-degree satisfy \( \deg_{\mathrm{in}}(v) \leq d \) and \( \deg_{\mathrm{out}}(v) \leq d \), respectively.
For an edge $e \in E$, we represent $e$ either as a vertex-index pair $(v, i)$, referring to the $i$-th outgoing edge of $v$, or as a vertex pair $(u, v)$, depending on the context.
We also use $\head(e)=v$ and $\tail(e)=u$ to denote the head and tail vertices of $e=(u,v)$, respectively.

Next we introduce some definitions and some more related notation.

\begin{definition}[rooted graph]
    A rooted graph $G(V,E,v)$ is a graph with vertex set $V$, edge set $E$, with a vertex $v\in V$ marked as root.  
\end{definition}

For a rooted digraph $G$, we denote by $V(G)$, $E(G)$, and $\rt(G)$ its vertex set, edge set, and root vertex, respectively.
For any $v\in V(G)$ and  $e_1,\dots,e_k\in E(G)$, we write $G[{v, e_1, \dots, e_k}]$ for the subgraph of $G$ induced by the edges $e_1, \dots, e_k$ and rooted at $v$.
In particular, we define $G[{v}]$ as the graph consisting solely of the root vertex $v$, rather than the empty graph.
For a tuple $W=(v,e_1,\ldots,e_k)$, we abbreviate
$G[W]:=G[{v,e_1,\ldots,e_k}]$.
For a rooted subgraph $H$, $H+e$ is obtained by adding the directed edge
$e$ and its endpoints, retaining the root of $H$.
If $e$ is specified by an outgoing slot $(v,i)$, it is first resolved
using the graph oracle; an unused slot cannot extend $H$.

In any directed graph $G$, the distance $\dist(u, v)$ between vertices $u$ and $v$ is defined as the length of the shortest path between $u$ and $v$ in the underlying undirected graph.

\subsection{\texorpdfstring{$r$}{r}-discs}

\begin{definition}[$r$-disc]\label{def:qdisc}
    Given a parameter $r \geq 0$ and a $d$-bounded-degree digraph $G=(V,E)$, an $r$-disc rooted at a vertex $v \in V$, denoted by $\disc_r(v)$, is the rooted subgraph (with $v$ marked as root) of $G$ induced by the vertices that are at distance at most $r$ from $v$. We also write $\disc_r^G(v)$ when the ambient graph needs to be specified.

\end{definition}

Let $\cD_{d,r}$ be the set of all possible isomorphism types
$\Gamma$ of $d$-bounded-degree $r$-discs.
That is, each type $\Gamma \in\cD_{d,r}$ is a rooted digraph with exactly one vertex marked as root, maximum in-degree and out-degree at most $d$, and all vertices of $\Gamma$ are within distance at most $r$ to the root.
Write $s_{d,r}:=\sum_{j=0}^r(2d)^j$ and $m_{d,r}:=2d s_{d,r}$.
These are upper bounds on the numbers of vertices and edges, respectively,
in every type $\Gamma\in\cD_{d,r}$.
We partition $\cD_{d,r}$ by the number of edges. Specifically, let $\cD_{d,r}^i$ be the subset of $\cD_{d,r}$ consisting of types with exactly $i$ edges. In particular, define $\cD_{d,r}^* := \cD_{d,r} \setminus \cD_{d,r}^0$, and let $D_{d,r} := |\cD_{d,r}^*|$.

\begin{definition}[isomorphism]
Two rooted digraphs $H_1,H_2$ are said to be isomorphic (denote as $H_1\cong H_2$) if and only if there exists a bijection $\phi:V(H_1)\xrightarrow{\sim}V(H_2)$ such that (denote as $H_1\cong_\phi H_2$):
    \begin{itemize}
        \item \emph{(Isomorphism)} $(u,v)\in E(H_1)$ if and only if $(\phi(u),\phi(v))\in E(H_2)$;
        \item \emph{(Root-preservation)} $\phi(\rt(H_1)) = \rt(H_2)$. 
    \end{itemize}
\end{definition}

Let $\Phi_{H_1,H_2}$ denote the set of isomorphisms from $H_1$ to $H_2$, and let $\Phi_H$ denote the automorphism group of $H$.
Although isomorphisms are bijections between vertex sets, we extend the notation $\phi$ to act on edge sets by defining $\phi((u,v)) = (\phi(u),\phi(v))$.
For a rooted type $\Gamma$, a $\Gamma$-instance is a rooted subgraph
isomorphic to $\Gamma$; it need not be an induced subgraph.

Define a binary relation $\preceq$ on $\cD_{d,r}$ such that $\Gamma \preceq \Gamma'$ if and only if $\Gamma'$ contains a root-preserving copy of $\Gamma$.
The relation $\preceq$ defines a partial order on the elements of $\cD_{d,r}$.
We fix an arbitrary linear extension of this partial order and index the elements of $\cD_{d,r}$ accordingly, so that $\cD_{d,r} = \{\Gamma_0, \Gamma_1, \dots, \Gamma_{D_{d,r}}\}$ with $\Gamma_i \preceq \Gamma_j$ implying $i \leq j$.
In particular, $\Gamma_0$ is the rooted graph consisting of a single vertex (the root) and no edges.

We fix the edge ordering of each isomorphism type with $i$ edges as the following iterative procedure for $i$ from $1$ to $m_{d,r}$.
Firstly, for each isomorphism type $\Gamma$ consisting of a single edge (i.e., $i=1$), the edge ordering is uniquely determined. 
For any $i\geq 2$, and a type  $\Gamma$ with $i$ edges, we arbitrarily select an isomorphism type $\Gamma'$ with $i-1$ edges such that $\Gamma \succeq \Gamma'$, meaning that $\Gamma$ can be obtained from $\Gamma'$ by adding a single edge. 
We then define the edge ordering of $\Gamma$ by appending the newly added edge to the edge ordering of $\Gamma'$ and define $\Gamma_{[j]}$ as the isomorphism type formed by the first $j$ edges of $\Gamma$ in our prefix order.
Such a parent always exists: keep a breadth-first spanning tree
rooted at the root, delete a non-tree edge if possible, and otherwise
delete a non-root leaf and its incident edge. The remaining rooted
graph still has radius at most $r$. The recursive choice also ensures
prefix consistency, $(\Gamma_{[j]})_{[h]}=\Gamma_{[h]}$ for $h\leq j$.
In particular, $\Gamma_{[0]}=\Gamma_0$ and $\Gamma_{[i]}=\Gamma$.

Given a digraph $G$ and an isomorphism type $\Gamma$, let $\cnt_\Gamma(G)$ denote the number of vertices that the $r$-discs rooted at them are of isomorphism type $\Gamma$, i.e.,
\[
\cnt_\Gamma(G) = \abs{\{v\in V(G):\disc_r(v)\cong \Gamma\}}.
\]
We may write it as $\cnt_\Gamma$ when context is clear.

\subsection{Query models}
\label{sec:query-models}
In query complexity, we study algorithms that access an unknown input \( x \in \Sigma^I \), where \( \Sigma \) is a finite alphabet and \( I \) is an index set. The algorithm initially has no knowledge of \( x \) and can only interact with it through an oracle. The nature of this oracle differs between the classical and quantum settings. In the classical case, the oracle is a black box \( \mathcal{O}_x \) that, upon receiving a query index \( i \in I \), returns the corresponding value \( x_i \). In the quantum setting, query access is modeled by a unitary operator \( \mathcal{O}_x \) acting on a quantum state. Specifically, for \( i \in I \) and \( z\in\Sigma \), the oracle maps:
$\mathcal{O}_x \ket{i}\ket{z} = \ket{i}\ket{z \oplus x_i}$, 
where \( \oplus \) denotes addition modulo \( |\Sigma| \) under a fixed identification of \( \Sigma \) with \( \{0,\ldots,|\Sigma|-1\} \). The answer register is \( \C^\Sigma \), typically initialized to \( \ket{0} \). The inverse oracle uses modular subtraction and can be implemented with one oracle call by negating the answer register before and after that call.

When the input is a \( d \)-bounded-degree digraph, there are two standard query models for accessing it. In the \emph{bidirectional model}, one has access to both the outgoing and incoming neighbors of each vertex, via the oracles \( \mathcal{O}^{\mathrm{out}}_G \) and \( \mathcal{O}^{\mathrm{in}}_G \), respectively. In contrast, the \emph{unidirectional model} restricts access to only the outgoing neighbors, i.e., one can query \( \mathcal{O}^{\mathrm{out}}_G \) but not \( \mathcal{O}^{\mathrm{in}}_G \).

We formally define the oracles \( \mathcal{O}^{\mathrm{out}}_G \) and \( \mathcal{O}^{\mathrm{in}}_G \) as follows:
\begin{align*}
    \mathcal{O}^{\mathrm{out}}_G(v, i) &=
    \begin{cases}
        w, & \text{if } w \in V \text{ is the } i\text{-th out-neighbor of } v; \\
        \perp, & \text{if } \deg_{\mathrm{out}}(v) < i,
    \end{cases} \\
    \mathcal{O}^{\mathrm{in}}_G(v, i) &=
    \begin{cases}
        w, & \text{if } w \in V \text{ is the } i\text{-th in-neighbor of } v; \\
        \perp, & \text{if } \deg_{\mathrm{in}}(v) < i.
    \end{cases}
\end{align*} 
\paragraph{Distance and testing conventions.}
For $d\geq1$, let $\mathcal G_{n,d}$ be the simple digraphs on $[n]$
whose in-degrees and out-degrees are at most $d$. For
$\Pi_n\subseteq\mathcal G_{n,d}$, define
\[
 \operatorname{dist}_d(G,\Pi_n)
 :=\frac1{dn}\min_{H\in\Pi_n}|E(G)\mathbin{\triangle}E(H)|,
\]
with value $+\infty$ if $\Pi_n=\varnothing$. Each directed edge
addition or deletion counts once, independently of the adjacency-list
representation. Thus both query models use the same distance, and the
target graph must satisfy the same degree bounds.
We call $G$ \emph{$\eps$-far} when this distance is at least $\eps$.
An $\eps$-tester accepts every graph in $\Pi_n$ and rejects every
$\eps$-far graph, in each case with probability at least $2/3$; its
behavior on the remaining graphs is unrestricted. For $d=0$, the input
is the unique edgeless graph on $[n]$, so membership requires no queries.

\paragraph{Model conventions.}
Both in-degree and out-degree are bounded by $d$, and graph properties
are invariant under isomorphism. Testers must succeed for every valid
adjacency-list ordering and may make adaptive queries. We count oracle
calls separately from computational running time. In particular, the
classical decision rule in \Cref{thm:qbd-classical} may depend on $n$ and
the property through a nonuniform lookup table.

\subsection{Testing \texorpdfstring{$k$}{k}-star-freeness and \texorpdfstring{$k$}{k}-occurrence-freeness}
We study the problem of testing $k$-star-freeness to establish our lower bound, and give its formal definition below.

\begin{definition}{($k$-star-freeness)}\label{def:kstarfree}
    For a $d$-bounded-degree digraph $G$,
    \begin{itemize}
        \item $G$ is said to be \emph{$k$-star-free} if $k$-star does not appear as a subgraph in $G$.
        \item $G$ is said to be \emph{${\varepsilon}$-far from $k$-star-freeness}, if 
     at least $\varepsilon dn$ edge deletions are needed to make it $k$-star-free. This agrees with the distance above, since additions cannot remove an existing star.
    \end{itemize}
\end{definition}
The $k$-star-freeness testing problem is defined as follows: given parameters $k \ge 2$, $\varepsilon > 0$, and $d$, together with query access to a $d$-bounded-degree digraph $G$ (in the unidirectional model), the goal is to distinguish whether $G$ is $k$-star-free or $\varepsilon$-far from being $k$-star-free. This problem can be viewed as a special case of testing \emph{$k$-source-subgraph-freeness}, where the $k$-star is replaced by an arbitrary constant-size connected digraph $H$ with at most $k$ source components.

We will relate the problem of testing $k$-star-freeness in $k$-bounded-degree digraphs to that of testing $k$-occurrence-freeness in functions. We formally define the latter below.
\begin{definition}{(k-occurrence-freeness)}\label{def:koccurrencefree}
    Let $R=cN$ for some constant $c\geq1$, and let $f:[N]\to[R]$ be a function promised to be $k$-bounded for $k\geq 2$. We define: 
    \begin{itemize}
        \item $f$ is \emph{$k$-occurrence-free} if  $|f^{-1}(r)|<k$  for every $r\in[R]$.%
        \item $f$ is \emph{${\varepsilon}$-far from $k$-occurrence-freeness} if at least $\varepsilon N$ function values must be modified to make it {$k$-occurrence-free}; equivalently, there exist at least $\eps N$ distinct  $r\in[R]$ such that $|f^{-1}(r)|=k$.
    \end{itemize}
\end{definition}

The alphabet-size assumption ensures enough unused capacity to move
one preimage from each value occurring $k$ times to values occurring
fewer than $k-1$ times. Thus the minimum number of changes is exactly
the number of values occurring $k$ times, proving the equivalence in
the definition. The later hard instances have alphabet size $R+N>N$.

Note that the $k$-occurrence-freeness testing problem is a variant of the $k$-collision-freeness testing studied in \cite{apers2024quantumpropertytestingsparse}, obtained by restricting the function to be $k$-bounded.

\subsection{Useful subroutines}
Grover search (\Grover) and quantum counting (\Count) are two fundamental subroutines used in this paper.
We briefly describe their functionality here and treat them as black-box components in our algorithms.

\begin{lemma}[Grover search {\cite[Theorem 2.5]{van2024basic}}\cite{grover1996fastquantummechanicalalgorithm,Boyer_1998, Brassard_2002}]\label{lem:grover}
    Given a Boolean function $f:[N]\to \bool$, and suppose $t=\abs{f^{-1}(1)}>0$.  
    There is a quantum algorithm $\emph{\Grover}(f)$ that outputs $x^* \in f^{-1}(1)$ uniformly with an expected number of quantum queries to $f$ at most  $O\left(\sqrt{N/t}\right)$ even without knowing $t$ in advance. 
\end{lemma}

We use independent runs of a verified search routine: a candidate is
returned only after checking that it is marked. If the marked set is empty,
the routine keeps searching until the enclosing algorithm's query counter
reaches its threshold. Thus no claim of uniform sampling from an empty
set is made.

\begin{lemma}[Quantum counting \cite{Brassard_2002} Theorem 13]
\label{lem:count}
    Given a Boolean function $f:[N]\to \bool$, and let $t=\abs{f^{-1}(1)}$.
    There is a quantum algorithm with $M$ queries that outputs an estimate $\tilde{t}$ with probability at least $8/\pi^2$ such that
    \begin{equation*}
        \abs{\tilde{t}-t} \leq 2\pi \frac{\sqrt{t(N-t)}}{M} + \pi^2\frac{N}{M^2}.
    \end{equation*}
\end{lemma}

Using the median trick introduced in  \Cref{lem:mt}, we can amplify the success probability from $8/\pi^2$ to  $1-\eta$ for any constant $\eta\in(0,1)$, %
hence the following corollary holds immediately.

\begin{corollary}
\label{cor:count}
For an integer $M\geq1$ and $\eta\in(0,1)$, put
$b(\eta)=2\lceil250\ln(2/\eta)\rceil+1$.
The routine $\Count(f,M,\eta)$ takes the median of $b(\eta)$ independent
quantum counting estimates, each with precision parameter $M$.
It uses $O(b(\eta)M)$ queries, returns a value in $[0,N]$, and satisfies
\[
 |\tilde t-t|\leq 2\pi\frac{\sqrt{t(N-t)}}{M}
                  +\pi^2\frac{N}{M^2}
\]
with probability at least $1-\eta$.
\end{corollary}
Clipping the median to $[0,N]$ cannot increase its error. The odd number of
repetitions and the probability bound follow from \Cref{lem:mt}.

\paragraph{Query threshold.}
Every estimator below maintains a counter of graph-oracle calls, including
inverse calls and calls made during predicate evaluation and verification.
Before making a call that would exceed its specified integer query
threshold $B$, it stops and returns the zero vector. Classical computation
and stored-data access do not count as graph queries.
Each extension predicate uses $O(1)$ graph queries: query the proposed
outgoing edge, compare it with the stored instances, and uncompute. A null
answer never defines an edge. Search calls use fresh independent randomness.

\paragraph{The run used in the analysis.}
View each estimator as a finite sequence of counting and sampling steps.
A counting step is one call to $\Count$; a sampling step is the entire
group of $\Grover$ calls that generates one sample multiset. A count must
satisfy its stated error bound. Sample counts must satisfy the stated
concentration bounds; in addition, sampled heads must be distinct in the
warm-up estimator, and sampled roots must be sufficiently far apart in the
disc estimator. Each estimator's analysis specifies these conditions
precisely. They are not tests performed by the algorithm.

For the proof only, remove the query threshold and stop immediately after
the first step whose required conditions fail; if none fails, run to
completion. Call this the \emph{stopped run}, and let $T$ be its total
number of graph queries. In particular, $T$ includes all queries in the
first step whose conditions fail. A sampling step is reached only if the
earlier conditions and the current counting accuracy hold. In our
applications these imply a nonempty marked set, so the sampling step
terminates almost surely by \Cref{lem:grover}.

\begin{lemma}[Query threshold]
\label{lem:query-threshold}
Suppose the stopped run satisfies $\E[T]\leq L$ for some $L>0$ and
finishes with all required conditions satisfied with probability at least
$1-\rho$. Then the estimator with query threshold $B=\lceil100L\rceil$
uses at most $B$ queries on every execution and finishes with all those
conditions satisfied with probability at least $1-\rho-1/100$.
\end{lemma}
\begin{proof}
Use the same randomness and measurement outcomes for both runs.
Whenever the stopped run finishes with all conditions satisfied and
$T\leq B$, the estimator with threshold $B$ follows the same computation
and returns the same output. The probability of any other outcome is at
most $\rho+\Pr[T>B]\leq\rho+L/B\leq\rho+1/100$, by Markov's inequality.
The query counter enforces the bound $B$ on every execution.
\end{proof}

\section{Constant-Query Quantum Bidirectional Testers Are Classical}\label{app:qbd-classical}
This section proves \Cref{thm:qbd-classical}. The main point is that, after we randomize vertex names and adjacency-list orderings, a constant-query quantum algorithm can distinguish two graphs only if their local neighborhoods have noticeably different distributions. We derive this statement directly from the usual quantum query circuit. %

\QBidClassical*

Recall that a radius-$r$ disc is the induced neighborhood of a root, where distance is measured in the underlying undirected graph and edge directions are retained. Its possible rooted isomorphism types form the finite set $\cD_{d,r}$ defined in the preliminaries. We write
\[
  \freq_r(G)_\Gamma
  :=\frac{\cnt_\Gamma(G)}{n}
  =\Pr_{v\sim\Unif(V(G))}[\disc_r(v)\cong\Gamma],
  \qquad \Gamma\in\cD_{d,r}.
\]
Thus $\freq_r(G)$ is the distribution of the neighborhood seen from a uniformly random vertex. It has a constant number of coordinates when $d,r$ are fixed. For two such vectors, $\|x-y\|_1=\sum_\Gamma|x_\Gamma-y_\Gamma|$.

\paragraph{Proof plan.}
We first express acceptance probabilities as linear functions of a tensor containing the oracle's contributions %
to the circuit and its complex conjugate. Next, we group those contributions according to constant-size patterns of query locations and answers. We prove that each group has a bounded contribution to every quantum algorithm, and that its coefficient is determined approximately by $\freq_r(G)$. These two facts give the locality bound
\[
 |p_\cA(G)-p_\cA(H)|
 \le O_{d,T}(1)\|\freq_r(G)-\freq_r(H)\|_1+O_{d,T}(1/n),
\]
where $p_\cA(G)$ denotes the acceptance probability after randomizing the representation of $G$. Finally, this bound shows that sampling constantly many discs suffices to distinguish graphs in the property from graphs that are $\eps$-far from it.

\subsection{Randomizing the graph representation}
\label{sec:qbd-symmetrization}

We may assume $d,T\ge1$: when $d=0$, the graph is edgeless and can be identified without queries, and a zero-query tester already has an input-independent classical simulation. All quantum algorithms in this section may use their oracle calls sequentially, with arbitrary computation between calls. This includes intermediate measurements and choices depending on earlier outcomes: by storing outcomes in extra workspace and deferring measurement, we obtain the usual unitary query-circuit form. The classical tester constructed at the end also makes adaptive queries when exploring neighborhoods.

A graph property does not depend on vertex names or on the ordering of adjacency lists. We therefore choose a uniformly random permutation of $[n]$ and, independently at every vertex, uniformly randomize the order of its actual incoming neighbors and of its actual outgoing neighbors. 
Only actual neighbors are permuted: positions above the relevant degree continue to return $\bot$. This preserves the convention for valid adjacency lists. A tester that succeeds on every valid representation retains its success guarantee under this randomization.

Let $\Sym(G)$ denote the resulting distribution over deterministic bidirectional oracles. We analyze the original $T$-query circuit with an oracle sampled once from $\Sym(G)$ and used for all $T$ calls. This is an average over valid representations, so it preserves the acceptance gap. If one implements this randomization using the given graph oracle, each randomized call can be simulated coherently using $O(d)$ original calls: read all $d$ entries in the requested direction, compute and add the permuted answer modulo $n+1$, and uncompute using inverse oracle calls. Thus the implementation still uses constantly many queries. Throughout the proof, $T$ counts calls to the randomized oracle.

\paragraph{The query unitary.}
The query locations and the answer-register alphabet are
\[
Q_n=\{\out,\inn\}\times[n]\times[d],
\qquad
\Ans_n=[n]\cup\{\bot\}.
\]
A query $q=(\mathrm{dir},v,i)$ requests the $i$-th neighbor
of $v$ in direction $\mathrm{dir}$.
Here $\oplus$ denotes addition modulo $n+1$, treating $\bot$ as $0$,
so $z\oplus\bot=z$. The answer register is $\C^{\Ans_n}$, with basis
$\{|z\rangle:z\in\Ans_n\}$.
A deterministic oracle $\mathsf O:Q_n\to\Ans_n$ acts by
\[
U_{\mathsf O}|q,z\rangle
=
|q,z\oplus\mathsf O(q)\rangle.
\]
Inverse calls fit the same circuit form: if $J|z\rangle=|-z\rangle$
on the answer register modulo $n+1$, then
$U_{\mathsf O}^{-1}=(I\otimes J)U_{\mathsf O}(I\otimes J)$.
Thus a call to either direction uses one forward call and
oracle-independent gates.
For later use, define
\[
X_a|z\rangle=|z\oplus a\rangle,
\qquad
E_{q,a}=|q\rangle\langle q|\otimes X_a.
\]
The operator $E_{q,a}$ keeps the part of a state whose query
location is $q$, adds $a$ to the answer register modulo $n+1$,
and sends all other query locations to zero.
We refer to $E_{q,a}$ as a \emph{partial query}.
Since $|q\rangle\langle q|$ is an orthogonal projection and
$X_a$ is unitary, $E_{q,a}$ is a \emph{contraction}: its operator
norm is at most one, so it cannot increase the norm of a state vector.
Summing over the answers specified by the oracle gives
\begin{equation}
U_{\mathsf O}
=
\sum_{q\in Q_n}\sum_{a\in\Ans_n}
\mathbf{1}[\mathsf O(q)=a]E_{q,a}.
\label{eq:qbd-oracle-expansion}
\end{equation}

\subsection{From a query circuit to a linear acceptance functional}
\label{sec:qbd-acceptance}

We start with a fixed $T$-query algorithm $\cA$. Let $|\psi_0\rangle$ be its initial unit vector and let $V_0,\ldots,V_T$ be its oracle-independent unitaries. Here $V_0$ is the computation before the first query, $V_t$ for $1\le t<T$ is the computation between queries $t$ and $t+1$, and $V_T$ is the computation after the last query. These operators may act on arbitrary private workspace. Every oracle operator below acts as the identity on that workspace.
Let $E$ be the operator for the accepting outcome of the final two-outcome measurement. Thus $0\le E\le I$, the rejecting operator is $I-E$, and the acceptance probability of a state $|\psi\rangle$ is $\langle\psi|E|\psi\rangle$. 
Then the algorithm $\cA$ can be described by the tuple $(\ket{\psi_0},V_0,\dots,V_T,E)$.
Working on oracle $\mathsf O$, the final state and acceptance probability are
\begin{equation}\label{eq:qbd-circuit-acceptance}
 \begin{aligned}
 |\psi_{\mathsf O}\rangle
   =V_TU_{\mathsf O}V_{T-1}\cdots V_1U_{\mathsf O}V_0|\psi_0\rangle \qquad \text{and} \qquad p_\cA(\mathsf O)=\langle\psi_{\mathsf O}|E|\psi_{\mathsf O}\rangle.
 \end{aligned}
\end{equation}

Each final-state amplitude is a sum of products containing one matrix entry from each of the $T$ oracle calls. The acceptance probability multiplies these amplitudes by their complex conjugates. Every resulting term therefore contains $T$ entries of $U_{\mathsf O}$ and $T$ conjugate entries. The remaining factors depend only on the algorithm. We collect the oracle-dependent factors in
\[
 U_{\mathsf O}^{\otimes T}\otimes\overline{U_{\mathsf O}}^{\otimes T}.
\]
Here the overline means entrywise complex conjugation in the computational basis. The first $T$ factors correspond to the state vector (the ket), and the other $T$ to its conjugate (the bra). 
Although \(U_\mathsf O=\overline{U_\mathsf O}\) for our oracle according to \eqref{eq:qbd-oracle-expansion}, we retain the conjugate notation to keep track of the ket and bra factors in the expansion of the acceptance probability.
All intervening computation remains in the algorithm-dependent factors.

Let $\mathcal H_n=\C^{Q_n}\otimes\C^{\Ans_n}$ be the query--answer Hilbert space and let $\operatorname{End}(\mathcal H_n)$ denote its space of linear operators. Then $U_\mathsf O\in \operatorname{End}(\mathcal H_n)$ and our tensors therefore belong to the finite-dimensional complex vector space
\[
 \mathcal Z_{n,T}
 :=\operatorname{End}(\mathcal H_n)^{\otimes T}
   \otimes\overline{\operatorname{End}(\mathcal H_n)}^{\otimes T}.
\]
The bar on the second space records the conjugated factors. We call this the \emph{oracle-moment space}; it is just the ambient space for the products of oracle entries described above.

For arbitrary operators $K_1,\ldots,K_T,L_1,\ldots,L_T$ on $\mathcal H_n$, put
\[
 Z=(K_1\otimes\cdots\otimes K_T)
    \otimes(\overline{L_1}\otimes\cdots\otimes\overline{L_T}).
\]
Run the circuit with $K_t$ substituted for call $t$ to obtain $|\psi_K\rangle$, and with $L_t$ substituted for call $t$ to obtain $|\psi_L\rangle$:
\[
 |\psi_K\rangle=V_TK_TV_{T-1}\cdots V_1K_1V_0|\psi_0\rangle,
 \qquad
 |\psi_L\rangle=V_TL_TV_{T-1}\cdots V_1L_1V_0|\psi_0\rangle.
\]
Define
\begin{equation}\label{eq:qbd-pairing}
 \langle\cA,Z\rangle:=\langle\psi_L|E|\psi_K\rangle
\end{equation}
and extend linearly to all tensors in $\mathcal Z_{n,T}$. This is well-defined because the right-hand side is linear in each $K_t$ and each $\overline{L_t}$. 

For a probability distribution $\mathsf D$ over deterministic oracles, define its \emph{oracle moment} by the finite average
\begin{equation}\label{eq:qbd-oracle-moment}
 \cM_T(\mathsf D)
 :=\sum_{\mathsf O}\Pr_{\mathsf D}[\mathsf O]\,
 U_{\mathsf O}^{\otimes T}\otimes\overline{U_{\mathsf O}}^{\otimes T}.
\end{equation}
Taking $K_t=L_t=U_{\mathsf O}$ in \eqref{eq:qbd-pairing} recovers \eqref{eq:qbd-circuit-acceptance}; linearity then gives
\begin{equation}\label{eq:qbd-average-acceptance}
 \langle\cA,\cM_T(\mathsf D)\rangle
 =\E_{\mathsf O\sim\mathsf D}[p_\cA(\mathsf O)].
\end{equation}
In particular, from now on
\[
 p_\cA(G):=\E_{\mathsf O\sim\Sym(G)}[p_\cA(\mathsf O)]
 =\langle\cA,\cM_T(\Sym(G))\rangle.
\]

This circuit description is related to the quantum-strategy formalism of Gutoski and Watrous~\cite{gutoski2007toward} and the quantum-comb formalism of Chiribella, D'Ariano, and Perinotti~\cite{chiribella2008quantum,chiribella2009theoretical}. These frameworks describe how a circuit interacts with operations inserted into its query positions. Abbott, Mhalla, and Pocreau~\cite{abbott2024quantumquery} also use this viewpoint to study quantum query complexity, including generalizations with indefinite causal order. We use ordinary circuits with a fixed order of oracle calls. The linear identity \eqref{eq:qbd-average-acceptance} follows from the circuit expansion above and does not require a characterization theorem for quantum combs.

\paragraph{Measuring error by its effect on acceptance.}
Small errors in individual tensor entries need not give small acceptance errors, because an algorithm can combine many entries coherently. We therefore measure an error tensor by the largest change it can produce in an algorithm's acceptance expression.

\begin{definition}[Operational acceptance seminorm]\label{def:qbd-strategy-norm}
For $Z\in\mathcal Z_{n,T}$, define
\[
 \|Z\|_{\str,T}:=\sup_{\cA}|\langle\cA,Z\rangle|.
\]
The supremum ranges over all $T$-query algorithms. Each may use arbitrary workspace, unitaries independent of the oracle, and a final measurement with two outcomes.
\end{definition}
For a difference of oracle moments, this quantity is exactly the largest difference in averaged acceptance probabilities obtainable by such algorithms. We retain the subscript $\str$ for its relation to strategy norms\footnote{The strategy norm was introduced in \cite{gutoski2012measure} to capture the distinguishability of quantum strategies.}. Homogeneity and the triangle inequality follow directly from the definition. We call it \textit{seminorm} because these are the only norm properties needed below, while the nondegeneracy statement about the ambient tensor space is \textit{not} required.

\subsection{Grouping queries into small patterns}
\label{sec:qbd-patterns}

\paragraph{Forgetting vertex names, but remembering equalities.}
Expanding each of the $T$ oracle factors and each of the $T$
conjugate factors using \eqref{eq:qbd-oracle-expansion} gives
\begin{equation}
\label{eq:oracle-dependent-factor-expansion}
U_{\mathsf O}^{\otimes T}
 \otimes\overline{U_{\mathsf O}}^{\otimes T}
=
\sum_{\substack{
q_t,q'_t\in Q_n \\a_t,a'_t\in\Ans_n
}}
\left(
  \prod_{t=1}^T
  \mathbf{1}[\mathsf O(q_t)=a_t]\,
  \mathbf{1}[\mathsf O(q'_t)=a'_t]
\right) \cdot
\left(\bigotimes_{t=1}^T E_{q_t,a_t}\right)
\otimes
\left(\bigotimes_{t=1}^T
  \overline{E_{q'_t,a'_t}}\right).
\end{equation}
An individual term specifies a direction, an adjacency-list position, a query vertex,
and an answer for each of these $2T$ positions.
The query--answer sequences on the ket and bra sides may differ,
since acceptance includes interference between different contributions
to the final state.
After averaging over $\mathsf O\sim\Sym(G)$, its coefficient
in the resulting expansion of $\cM_T(\Sym(G))$ is
\[
\Pr_{\mathsf O\sim\Sym(G)}
\left[
  \mathsf O(q_t)=a_t,\;
  \mathsf O(q'_t)=a'_t
  \quad\text{for all }t=1,\ldots,T
\right].
\]
All requirements concern the same randomized oracle.
Any term with incompatible requirements has coefficient zero;
for example, the same query cannot return two different answers.

Replace the vertex names in such a term by abstract labels,
retaining exactly which names are equal.
Assign labels in order of first appearance, reading the ket-side
positions first and then the bra-side positions, in increasing
query order on each side.
Within each position, read the query vertex before the answer.
The symbol $\bot$ is left unchanged and is not treated as a
vertex label.
The resulting \emph{doubled query--answer diagram} $D$ records
two ordered query--answer sequences,
\[
\bigl((\tilde q_t,\tilde a_t)\bigr)_{t=1}^T
\qquad\text{and}\qquad
\bigl((\tilde q'_t,\tilde a'_t)\bigr)_{t=1}^T,
\]
corresponding to the ket and bra sides, respectively.
The tildes indicate that vertex names have been replaced by
abstract labels; directions and adjacency-list positions are unchanged.
We call a diagram \emph{consistent} if all its requirements can hold
simultaneously for some valid bidirectional oracle.

Let $A(D)$ be the set of abstract vertex labels occurring in $D$,
excluding $\bot$.
Define the \emph{constraint graph} of $D$ to have vertex set
$A(D)$, and join a query label to its answer label by an
undirected edge for every non-$\bot$ answer.
A requirement with answer $\bot$ adds no edge, but its query
label remains a vertex of the constraint graph.
Directions and adjacency-list positions remain part of $D$, although we ignore
them when defining connected components.
Choose one query label as the root of each component, and put
\[
m(D)=|A(D)|,\qquad
q(D)=\text{number of components},\qquad
s(D)=m(D)-q(D).
\]
Thus $s(D)$ counts the non-root labels.
There are at most $4T$ labels. Every connected component contains at
least one query label: a non-$\bot$ answer label is joined to its query
label, while a $\bot$ requirement retains its query label as a vertex.
There are only $2T$ query positions across the ket and bra sides, so
$q(D)\leq2T$.
For fixed $d,T$, the number of possible diagrams is bounded
by a constant independent of $n$.

For an injective assignment $\phi:A(D)\hookrightarrow[n]$,
extend $\phi$ to the symbol $\bot$ and to abstract queries by
\[
\phi(\bot)=\bot,
\qquad
\phi((\mathrm{dir},x,i))
:=(\mathrm{dir},\phi(x),i),
\]
where $\mathrm{dir}\in\{\mathrm{out},\mathrm{in}\}$ and $x\in A(D)$.
Define
\[
W_{D,\phi}^{(n)}
:=
\left(
  \bigotimes_{t=1}^{T}
  E_{\phi(\tilde q_t),\,\phi(\tilde a_t)}
\right)
\otimes
\left(
  \bigotimes_{t=1}^{T}
  \overline{
    E_{\phi(\tilde q'_t),\,\phi(\tilde a'_t)}
  }
\right).
\]
Thus $W_{D,\phi}^{(n)}$ is the tensor product of the specified
$E_{q,a}$ operators on the ket side and the specified
$\overline{E_{q,a}}$ operators on the bra side,
after replacing each abstract vertex label $x$ by $\phi(x)$.
Injectivity ensures that distinct abstract labels receive
distinct vertex names.

Since every $E_{q,a}$ is a contraction, the ket- and bra-side
substitutions specified by $W_{D,\phi}^{(n)}$ in
\eqref{eq:qbd-pairing} produce vectors $|\psi_K\rangle$ and
$|\psi_L\rangle$ of norm at most one. Thus, for every algorithm $\cA$,
\[
 \left|\left\langle\cA,W_{D,\phi}^{(n)}\right\rangle\right|
 =|\langle\psi_L|E|\psi_K\rangle|
 \le\|\psi_L\|\,\|\psi_K\|\le1,
\]
by Cauchy--Schwarz and $0\le E\le I$.
Taking the supremum over algorithms gives
$\|W_{D,\phi}^{(n)}\|_{\str,T}\le1$.

\paragraph{Normalization of diagram tensors.}
Once a component root is placed in a bounded-degree graph, there are only constantly many ways to place its other labels while satisfying all its edges. Accordingly, a pattern with $m$ labels and $q$ components typically has probability of order $n^{q-m}$. The natural normalization therefore averages the $s=m-q$ non-root names. For $n\ge m(D)$, define the root-normalized diagram tensor
\[
 \cU_D^{(n)}:=\frac{1}{(n-q(D))_{s(D)}}
       \sum_{\phi:A(D)\hookrightarrow[n]}W_{D,\phi}^{(n)},
\]
For each fixed assignment of distinct root names, the denominator counts the possible assignments of the non-root names. The root names are summed, rather than averaged; the next lemma explains why this does not introduce a growing acceptance error.

\begin{lemma}[Bound for a root-normalized diagram]\label{lem:anchored-diagram}
For each fixed $T\ge1$, set $K_T:=2^{1+\binom{4T}{2}}$. For every consistent doubled query--answer diagram $D$ of length $T$ and all sufficiently large $n$,
\begin{equation}
\label{eq:universal-diagram-norm}
 \|\cU_D^{(n)}\|_{\str,T}\le K_T.   
\end{equation}
The bound holds for every quantum algorithm and is independent of the graph and its size.
\end{lemma}

Applying the triangle inequality with $\|W_{D,\phi}^{(n)}\|_{\str,T}\le1$ would give $\|\cU_D^{(n)}\|_{\str,T}\le(n)_{q(D)}$, which grows with $n$. The additional ingredient is that projections onto different query vertex names are orthogonal. Each connected component supplies two uses of this orthogonality, together saving one factor of $n$. We give the details next.

\begin{proof}
Fix a pure \(T\)-query tester $\cA$.  Write its initial unit vector as \(|\psi_0\rangle\), its oracle-independent inter-query unitaries as \(V_0,\ldots,V_T\), and its accepting measurement operator as \(0\le E\le I\).  For an assignment \(\phi\), let \(K_t(\phi)\) and \(L_t(\phi)\) be the ket- and bra-side partial query operators specified by \(D\), i.e.
\[
K_t(\phi) = E_{\phi(\tilde q_t),\,\phi(\tilde a_t)}, \qquad \text{and} \qquad L_t(\phi) = E_{\phi(\tilde q'_t),\,\phi(\tilde a'_t)}.
\]
The pairing with the elementary tensor is
\[
 \left\langle\cA,W_{D,\phi}^{(n)}\right\rangle = \langle \psi_L(\phi)|E|\psi_K(\phi)\rangle
 =\langle\psi_0|S_\phi|\psi_0\rangle,
\]
where
\[
\begin{aligned}
 S_\phi
 :=\;&V_0^*L_1(\phi)^*V_1^*\cdots
       L_T(\phi)^*V_T^* E V_T
       K_T(\phi)V_{T-1}\cdots K_1(\phi)V_0 
\end{aligned}
\]
Here $V_t^*$ denotes the adjoint of $V_t$, and similarly for the partial queries.
Every factor in this product is a contraction.  It therefore suffices to bound the operator norm of the corresponding normalized sum of the products \(S_\phi\).

\paragraph{Step 1: allow repeated vertex names.}
We first allow arbitrary maps \(\phi:A(D)\to[n]\), without imposing injectivity, and prove
\begin{equation}\label{eq:all-map-diagram-bound}
 \left\|
   \frac1{n^{s(D)}}
   \sum_{\phi:A(D)\to[n]}S_\phi
 \right\|\le1.
\end{equation}
For a query location \(q=(\mathrm{dir},v,i)\), write
\[
  P_{\mathrm{dir},i,v}
  :=|\mathrm{dir},v,i\rangle\langle\mathrm{dir},v,i|\otimes I_{\Ans_n}.
\]
Here $I_{\Ans_n}$ is the identity on $\C^{\Ans_n}$. The projections $P_{\mathrm{dir},i,v}$, as $v$ varies, are pairwise orthogonal. In the next formula, $X_a$ acts as the identity on the query register, and all operators act as the identity on private workspace. Since the answer shift acts on a different register, every partial query factor may be split into three factors
\[
 E_{q,a}=P_{\mathrm{dir},i,v}X_aP_{\mathrm{dir},i,v},
\]
and the same factorization holds for its adjoint, with $X_a$ replaced by $X_a^*$.  We order all occurrences of labels in these factors in the order in which they act on the state, namely from right to left in the displayed product \(S_\phi\); within a partial query factor the order is
\[
  \text{input query projection},
  \quad\text{answer shift},
  \quad\text{output query projection}.
\]
For a label \(x\in A(D)\), let \(a_x\in\{0,1,2\}\) be the number of the following two events that hold: the first occurrence of \(x\) is a query projection, and the last occurrence of \(x\) is a query projection.

To sum consistently over repeated occurrences of a label, we temporarily store its assigned vertex name in an auxiliary register. These registers are a device for bounding the operator norm, not extra queries made by the algorithm. Here is the construction.  Split \(S_\phi\) into label-independent contractions and factors $B_j(\phi(x_j))$ depending on a single label, where each \(B_j(v)\) is either a query projection \(P_{\mathrm{dir},i,v}\) or an answer shift \(X_v\) or \(X_v^*\).  At a cut in this product, call a label \emph{active} if its first occurrence has already acted but its last occurrence has not.  Add one register \(\C^n\) for every active label.  The value of that register is the public name assigned to the label.

\paragraph{Step 2: bound the cost of summing each label.}
At each position $j$ involving a label $x$, replace the factor
$B_j(\phi(x))$ by a local linear map according to whether
this is the first and/or last occurrence of $x$:
\[
\begin{array}{c|c|c|c}
\text{Occurrence of }x
& \text{Map}
& \text{Definition}
& \text{Norm bound}\\
\hline
\text{First, not last}
& J_x
& \displaystyle\sum_{v=1}^n B_j(v)\otimes|v\rangle
& \begin{cases}
    1, & \text{projection},\\
    \sqrt n, & \text{shift}
  \end{cases}
\\[4mm]
\text{Neither first nor last}
& C_{x,j}
& \displaystyle\sum_{v=1}^n
    B_j(v)\otimes|v\rangle\langle v|
& 1
\\[4mm]
\text{Last, not first}
& R_x
& \displaystyle\sum_{v=1}^n B_j(v)\otimes\langle v|
& \begin{cases}
    1, & \text{projection},\\
    \sqrt n, & \text{shift}
  \end{cases}
\\[4mm]
\text{Only occurrence}
& F_x
& \displaystyle\sum_{v=1}^n B_j(v)
& n
\end{array}
\]
Here ``projection'' means that $B_j(v)$ is a query projection,
and ``shift'' includes both $X_v$ and $X_v^*$.
We show only the register for $x$ in the table; all other active
label registers are left unchanged and are kept in a fixed order.

The map $J_x$ introduces a register storing the value assigned
to $x$. Each subsequent map $C_{x,j}$ applies the corresponding
factor while retaining that value, and $R_x$ removes the register
at the last occurrence.
If $x$ occurs only once, $F_x$ sums over its value directly,
without introducing a register.
Thus every occurrence contributes its factor $B_j(v)$
exactly once.

We now verify the norm bounds.
At the first occurrence $j$ of a non-singleton label $x$,
\[
J_x^*J_x
=
\sum_{v=1}^n B_j(v)^*B_j(v)
\le nI,
\]
so $\|J_x\|\le\sqrt n$.
Similarly, at its last occurrence $j$,
\[
R_xR_x^*
=
\sum_{v=1}^n B_j(v)B_j(v)^*
\le nI,
\]
so $\|R_x\|\le\sqrt n$.
If the corresponding occurrence is a query projection, the
operators $B_j(v)$ are pairwise orthogonal projections, and
the relevant sum is at most $I$.
Hence the corresponding map has norm at most one.

For an internal occurrence, the controlled map is block
diagonal in the label register, giving
\[
\|C_{x,j}\|
=
\max_{v\in[n]}\|B_j(v)\|
\le 1.
\]
For a singleton label, the triangle inequality gives
\[
\|F_x\|
\le
\sum_{v=1}^n\|B_j(v)\|
\le n.
\]
Every query slot contains an input and an output query
projection with the same query label.
Consequently, a singleton label can occur only in an answer
shift, and therefore has $a_x=0$.

Lift each label-independent contraction by the identity on
the active label registers.
Composing all the local maps and these lifted contractions
in the original action order gives exactly
\[
\sum_{\phi:A(D)\to[n]}S_\phi.
\]
Indeed, each label register ensures that all occurrences of
its label use the same value, and removing the register sums
over that value. Singleton values are summed by $F_x$.
Every original factor is applied exactly once, with its
position in the product preserved, and no label registers
remain at the end.

For a label with at least two occurrences, the entry and exit
maps together contribute at most
\[
\|J_x\|\,\|R_x\|
\le n^{1-a_x/2},
\]
while its internal maps have norm at most one.
A singleton label obeys the same bound because $a_x=0$
and $\|F_x\|\le n$.
By submultiplicativity of the operator norm,
\[
\left\|
  \sum_{\phi:A(D)\to[n]}S_\phi
\right\|
\le
\prod_{x\in A(D)}n^{1-a_x/2}
=
n^{m(D)-\frac12\sum_{x\in A(D)}a_x}.
\]
Using $s(D)=m(D)-q(D)$, we obtain
\[
\frac{1}{n^{s(D)}}
\left\|
  \sum_{\phi:A(D)\to[n]}S_\phi
\right\|
\le
n^{q(D)-\frac12\sum_{x\in A(D)}a_x}.
\]

\paragraph{Step 3: save one factor of $n$ per component.}
For every connected component of the constraint graph, its first occurrence in the action order is necessarily the input query projection of a slot belonging to that component: an answer occurrence in a slot is preceded by the input projection of the adjacent query label.  Similarly, the last occurrence of the component is necessarily an output query projection.  These two boundary occurrences are distinct, even for a component appearing in only one slot.  Therefore the labels in each component contribute at least two to \(\sum_xa_x\).  Thus \(\sum_xa_x\ge2q(D)\), and \eqref{eq:all-map-diagram-bound} follows.

\paragraph{Step 4: require distinct vertex names.}
It remains to remove assignments that send two different abstract labels to the same vertex. Ordinary inclusion--exclusion suffices. Let $\mathcal P_D$ be the set of unordered pairs of distinct labels in $A(D)$. Then
\[
 \mathbf1[\phi\text{ is injective}]
 =\prod_{\{x,y\}\in\mathcal P_D}
    \bigl(1-\mathbf1[\phi(x)=\phi(y)]\bigr)
 =\sum_{B\subseteq\mathcal P_D}(-1)^{|B|}
    \mathbf1[\phi(x)=\phi(y)\text{ for all }\{x,y\}\in B].
\]
For each choice of $B$, identify all labels forced equal by its pairs. The resulting diagram has, say, $m_B$ labels and $q_B$ components. It may no longer describe a valid oracle, but the factorization into query projections and answer shifts, the orthogonality of projections onto distinct query names, and the component-boundary argument remain unchanged. Hence the preceding operator bound still applies. Moreover,
$s_B:=m_B-q_B\le s(D)$, because identifying two labels reduces the number of labels by one and the number of components by at most one. Thus \eqref{eq:all-map-diagram-bound}, applied to this identified diagram, bounds its unnormalized sum by $n^{s_B}\le n^{s(D)}$. There are $2^{\binom{m(D)}2}$ subsets $B$. The triangle inequality therefore gives
\[
 \left\|\frac1{n^{s(D)}}
       \sum_{\phi:A(D)\hookrightarrow[n]}S_\phi\right\|
 \le 2^{\binom{m(D)}2}.
\]
Finally, $m(D)\leq4T$ implies that the normalization ratio
$n^{s(D)}/(n-q(D))_{s(D)}$ is at most two for all sufficiently large
$n$, uniformly over all diagrams of length $T$. Multiplying the
preceding bound by this ratio and taking the supremum over algorithms
gives the stated constant $K_T$. The circuit form already includes
mixed states and classical randomness by purification.
\end{proof}

\subsection{How local neighborhoods determine the coefficients}
\label{sec:qbd-locality}

We now connect the diagram expansion to graph neighborhoods. Our goal is the following statement: two graphs with nearly the same distribution of rooted discs have nearly the same averaged acceptance probability for every $T$-query algorithm. 

Intuitively, \Cref{lem:uniform-locality} is natural: with only $T$ queries,
an algorithm can witness only query--answer patterns involving a bounded
number of vertices and edges. Each connected part of such a pattern is
contained in a neighborhood of bounded radius. The vector $\freq_r(G)$
records the frequency of every rooted neighborhood type at radius $r$,
so it captures the local information relevant to these patterns when $r$
is sufficiently large. After averaging over vertex names and adjacency-list position
orderings, we should therefore expect two sufficiently large graphs
with similar disc-frequency vectors to have similar acceptance
probabilities.

\begin{lemma}[Uniform quantum locality]\label{lem:uniform-locality}
Fix constants \(d,T\) and \(\eta>0\).  There exist constants
\[
  r=r(d,T),\qquad
  \Delta=\Delta(d,T,\eta)>0,
  \qquad
  n_0=n_0(d,T,\eta),
\]
such that for every \(n\ge n_0\), every \(T\)-query bidirectional quantum algorithm \(\calA\), and every pair of \(n\)-vertex \(d\)-bounded digraphs \(G,H\),
\[
  \|\freq_r(G)-\freq_r(H)\|_1\le \Delta
  \quad\Longrightarrow\quad
  |p_\calA(G)-p_\calA(H)|\le \eta,
\]
where \(p_\calA(G)=\langle\calA,\calM_T(\Sym(G))\rangle\).  One may take \(r=4T+3\).
\end{lemma}

We prove the following quantitative bound first. It separates the argument into two tasks: \Cref{lem:anchored-diagram} bounds the diagram tensors, while the proof below estimates their scalar coefficients from local neighborhoods. The tensors depend on $n$, but their number and all bounds depend only on $d,T$.

\begin{lemma}[Quantitative locality bound]\label{lem:finite-dimensional}
Fix \(d,T\), and let \(r=4T+3\). For all sufficiently large \(n\) and all \(n\)-vertex \(d\)-bounded digraphs \(G,H\),
\[
  \|\calM_T(\Sym(G))-\calM_T(\Sym(H))\|_{\str,T}
  \le
  O_{d,T}(1)\|\freq_r(G)-\freq_r(H)\|_1+O_{d,T}(1/n).
\]
\end{lemma}

\begin{proof}
The implicit constants in $O_{d,T}(\cdot)$ may differ between occurrences, but depend only on $d,T$, uniformly over the diagrams, the graph size, the graphs, and the quantum tester.

\medskip

\paragraph{Step 1: group the expansion by patterns of equal vertex names.}
For $\mathbf q=(q_1,\ldots,q_T)\in Q_n^T$ and
$\mathbf a=(a_1,\ldots,a_T)\in\Ans_n^T$, write
\[
E_{\mathbf q,\mathbf a}
:=
E_{q_1,a_1}\otimes\cdots\otimes E_{q_T,a_T}.
\]
Averaging \eqref{eq:oracle-dependent-factor-expansion} according to
\eqref{eq:qbd-oracle-moment} gives
\[
\mathcal M_T(\Sym(G))
=
\sum_{\mathbf q,\mathbf a,\mathbf q',\mathbf a'}
p_G(\mathbf q,\mathbf a;\mathbf q',\mathbf a')
E_{\mathbf q,\mathbf a}
\otimes\overline{E_{\mathbf q',\mathbf a'}},
\]
where $p_G(\mathbf q,\mathbf a;\mathbf q',\mathbf a')$ is the
probability that all $2T$ specified query--answer requirements
hold for the same oracle $\mathsf O\sim\Sym(G)$.
Terms with inconsistent requirements have coefficient zero
and can be omitted.

For each remaining choice of indices, replace vertex names
by abstract labels in order of first appearance.
This gives a unique consistent doubled query--answer diagram
$D$ and a unique injection $\phi:A(D)\hookrightarrow[n]$.
Writing the two sequences recorded by $D$ as
\(
\bigl((\tilde q_t,\tilde a_t),(\tilde q'_t,\tilde a'_t)\bigr)_{t=1}^T,
\)
the original indices are recovered by
\[
\begin{aligned}
q_t=\phi(\tilde q_t), \quad
 a_t=\phi(\tilde a_t), \quad
 q'_t=\phi(\tilde q'_t), \quad
 a'_t&=\phi(\tilde a'_t).
\end{aligned}
\]
Conversely, these equations determine a unique choice of
indices from $D$ and $\phi$.
Thus we may reindex the expansion by pairs $(D,\phi)$.

As the indices vary, only $O_{d,T}(1)$ distinct diagrams can
arise: a diagram records the directions and adjacency-list positions at $2T$
positions, which answers are $\bot$, and the equalities among
at most $4T$ vertex occurrences.
We let $D$ range over the fixed collection of consistent
diagrams, allowing coefficients to be zero for the particular
graph $G$.

For a diagram $D$ and an injection $\phi$, temporarily write
\[
p_{D,\phi}(G)
:=
\Pr_{\mathsf O\sim\Sym(G)}
\left[
\begin{array}{l}
\mathsf O(\phi(\tilde q_t))=\phi(\tilde a_t),\\
\mathsf O(\phi(\tilde q'_t))=\phi(\tilde a'_t)
\end{array}
\quad\text{for all }t\in[T]
\right].
\]
Under the correspondence above,
\[
p_G(\mathbf q,\mathbf a;\mathbf q',\mathbf a')
=
p_{D,\phi}(G),
\]
and the definition of $W_{D,\phi}^{(n)}$ gives
\[
\begin{aligned}
E_{\mathbf q,\mathbf a}
\otimes\overline{E_{\mathbf q',\mathbf a'}}
&=
\left(
\bigotimes_{t=1}^T
E_{\phi(\tilde q_t),\,\phi(\tilde a_t)}
\right)
\otimes
\left(
\bigotimes_{t=1}^T
\overline{
E_{\phi(\tilde q'_t),\,\phi(\tilde a'_t)}
}
\right)=W_{D,\phi}^{(n)}.
\end{aligned}
\]
Consequently, the expansion can be regrouped exactly as
\[
\mathcal M_T(\Sym(G))
=
\sum_D
\sum_{\phi:A(D)\hookrightarrow[n]}
p_{D,\phi}(G)\,W_{D,\phi}^{(n)}.
\]

For a fixed $D$, the coefficient $p_{D,\phi}(G)$ is independent
of $\phi$.
To prove this, let $\phi,\phi':A(D)\hookrightarrow[n]$ be two
injections.
The bijection sending $\phi(x)$ to $\phi'(x)$ for every
$x\in A(D)$ extends to a permutation $\pi$ of $[n]$.
Thus
\[
\phi'(x)=\pi(\phi(x))
\qquad\text{for all }x\in A(D).
\]
Extend $\pi$ by setting $\pi(\bot)=\bot$.
For an oracle $\mathsf O$, define its relabeling by
\[
\mathsf O^\pi(\mathrm{dir},v,i)
:=
\pi\!\left(
\mathsf O(\mathrm{dir},\pi^{-1}(v),i)
\right).
\]
If $\mathsf O\sim\Sym(G)$, then $\mathsf O^\pi\sim\Sym(G)$:
composing the uniformly random vertex relabeling with the
fixed permutation $\pi$ preserves its uniform distribution,
while retaining all adjacency-list positions.

Moreover, for each ket-side requirement,
\[
\mathsf O^\pi(\phi'(\tilde q_t))
=
\pi\!\left(\mathsf O(\phi(\tilde q_t))\right).
\]
Since $\phi'(\tilde a_t)=\pi(\phi(\tilde a_t))$ and $\pi$ is
a bijection on $[n]\cup\{\bot\}$, it follows that
\[
\mathsf O^\pi(\phi'(\tilde q_t))=\phi'(\tilde a_t)
\quad\Longleftrightarrow\quad
\mathsf O(\phi(\tilde q_t))=\phi(\tilde a_t).
\]
The same equivalence holds for every bra-side requirement.
Hence
\[
\begin{aligned}
p_{D,\phi'}(G)
&=
\Pr_{\mathsf O\sim\Sym(G)}
\left[
\begin{array}{l}
\mathsf O^\pi(\phi'(\tilde q_t))=\phi'(\tilde a_t),\\
\mathsf O^\pi(\phi'(\tilde q'_t))=\phi'(\tilde a'_t)
\end{array}
\quad\text{for all }t\in[T]
\right]\\
&=
\Pr_{\mathsf O\sim\Sym(G)}
\left[
\begin{array}{l}
\mathsf O(\phi(\tilde q_t))=\phi(\tilde a_t),\\
\mathsf O(\phi(\tilde q'_t))=\phi(\tilde a'_t)
\end{array}
\quad\text{for all }t\in[T]
\right]=p_{D,\phi}(G).
\end{aligned}
\]
The first equality uses the invariance of $\Sym(G)$ under
$\pi$, and the second uses the equivalence of the displayed
requirements.
Denote this common coefficient by $p_D(G)$.

Recall that
\[
\mathcal U_D^{(n)}
=
\frac{1}{(n-q(D))_{s(D)}}
\sum_{\phi:A(D)\hookrightarrow[n]}W_{D,\phi}^{(n)},
\]
and define
\[
c_D(G):=(n-q(D))_{s(D)}p_D(G).
\]
For sufficiently large $n$, all these denominators are
nonzero, since $m(D)\le4T$.
We therefore obtain
\begin{equation}
\label{eq:exact-diagram-expansion}
\begin{aligned}
\mathcal M_T(\Sym(G))
&=
\sum_D
\sum_{\phi:A(D)\hookrightarrow[n]}
p_{D,\phi}(G)\,W_{D,\phi}^{(n)}\\
&=
\sum_D
\sum_{\phi:A(D)\hookrightarrow[n]}
p_D(G)\,W_{D,\phi}^{(n)}\\
&=
\sum_D p_D(G)
\left(
\sum_{\phi:A(D)\hookrightarrow[n]}W_{D,\phi}^{(n)}
\right)\\
&=
\sum_D c_D(G)\,\mathcal U_D^{(n)}.
\end{aligned}
\end{equation}

\medskip
\paragraph{Step 2: interpret the coefficient \(c_D(G)\).}
Fix a diagram $D$. The common coefficient $p_D(G)$ is a probability for specified public vertex names. It is easier to compute by asking which original graph vertices received those names under the random relabeling. Order the roots and non-root labels once and for all.  For an ordered tuple of distinct vertices
\(\mathbf v=(v_1,\ldots,v_{q(D)})\), define
\[
  F_D^G(\mathbf v)
  :=
  \sum_{\mathbf u\in
  \perm{V(G)\setminus\{v_1,\ldots,v_{q(D)}\}}{s(D)}}
  \kappa_D^G(\mathbf v,\mathbf u),
\]
where \(\kappa_D^G(\mathbf v,\mathbf u)\in[0,1]\) is the probability, over the independent uniformly random valid incoming and outgoing adjacency-list position orderings, that all position constraints of \(D\) hold under the assignment sending the roots to \(\mathbf v\) and the non-roots to \(\mathbf u\).  We set \(\kappa_D^G(\mathbf v,\mathbf u)=0\) when an underlying directed-edge or \(\bot\)-constraint fails.

Thus $F_D^G(\mathbf v)$ is the weighted number of ways to complete the pattern once its roots are fixed; unlike $p_D(G)$, it is not divided by the number of possible non-root assignments.

Conditioned on the preimages of the public root names being \(\mathbf v\), the preimages of the public non-root names form a uniformly random ordered injection from the remaining \(n-q(D)\) vertices.  The conditional probability of the pattern is therefore $F_D^G(\mathbf v)/(n-q(D))_{s(D)}$. Multiplying by the normalization in $c_D(G)$ cancels this denominator, giving
\begin{equation}\label{eq:coefficient-as-root-average}
  c_D(G)
  =
  \frac1{(n)_{q(D)}}
  \sum_{\mathbf v\in \perm{V(G)}{q(D)}}
  F_D^G(\mathbf v)
  =
  \E_{\mathbf v\sim\Unif\left(\perm{V(G)}{q(D)}\right)}
  F_D^G(\mathbf v),
\end{equation}

The quantity \(F_D^G(\mathbf v)\) is uniformly bounded.  To see this, choose a spanning tree in every component of the constraint graph of \(D\), rooted at its designated root.  Once the roots are fixed, every non-root image can be exposed along these trees, and at every step there are at most \(2d\) choices in the underlying undirected graph.  Hence
\begin{equation}\label{eq:F-uniform-bound}
  0\le F_D^G(\mathbf v)\le R_{d,T}
\end{equation}
Here and below, we fix $R_{d,T}:=(2d)^{4T}$, since $s(D)\leq4T$ and each realization has weight at most one.

\medskip

\paragraph{Step 3: dependence only on rooted disc types, up to $O(1/n)$.}
Let $r=4T+3$ and write $q=q(D)$. Recall that a radius-$\ell$ neighborhood
contains at most
\(
 s_{d,\ell}=\sum_{j=0}^{\ell}(2d)^j
\)
vertices, since the underlying undirected degree is at most $2d$.
Call an ordered tuple $\mathbf v=(v_1,\ldots,v_q)$ of distinct roots
\emph{good} if their radius-$r$ discs are pairwise disjoint.
Two such discs can intersect only if their roots are at distance at
most $2r$. Thus, for uniformly sampled distinct roots, a union bound gives
\begin{equation}
 \Pr[\mathbf v\text{ is not good}]
 \le \binom{q}{2}\frac{s_{d,2r}}{n-q+1}
 =O_{d,T}(1/n).
 \label{eq:qbd-step3-bad-roots}
\end{equation}
Every label of $D$ is connected to its component root by at most $2T$
non-$\bot$ constraints. Consequently, every assignment satisfying the
required directed edges places that label within distance $2T$ of its
root. On a good root tuple, the labels of different components therefore
lie in disjoint rooted discs. Moreover, every queried vertex is within
distance $2T$ of its root, so all its neighbors lie within distance
$2T+1<r$. Its disc thus determines its full incoming and outgoing degrees,
as well as all the edges needed to evaluate the adjacency-list position and $\bot$ requirements.

For each rooted type $\Gamma\in\mathcal D_{d,r}$, fix a concrete rooted
graph $H_\Gamma$ of that type. These representatives are chosen independently
of $G$; a type need not occur in the current graph.
Order the components of the constraint graph of $D$ as
$C_1,\ldots,C_q$, with designated roots $x_1,\ldots,x_q$.
Given a tuple of rooted types
$\vec\Gamma=(\Gamma_1,\ldots,\Gamma_q)$, let $H_j$ be a copy of
$H_{\Gamma_j}$ with root $o_j$, using pairwise disjoint vertex sets even
when some of the types coincide.
A \emph{local realization} of $D$ in these copies is an injective map
\[
 \rho:A(D)\hookrightarrow\bigsqcup_{j=1}^{q}V(H_j)
\]
that maps the labels of $C_j$ into $H_j$, sends $x_j$ to $o_j$, and
preserves every directed-edge requirement of $D$.
For example, an $\mathrm{out}$ requirement with query label $x$ and
answer label $y\ne\bot$ requires the edge $\rho(x)\to\rho(y)$;
an $\mathrm{in}$ requirement requires the edge $\rho(y)\to\rho(x)$.
The map need not use every vertex of $H_j$ or preserve nonedges.
Adjacency-list positions and $\bot$ requirements will be accounted for by the weights below.
Let $\mathcal R(D,\vec\Gamma)$ denote the set of these maps. A concrete example follows.

\medskip
\noindent\emph{Example.}
Take $T=2$ and $d\ge2$, and let $D$ have the requirements
\[
\begin{array}{ll}
 \text{ket side:}
 &\mathsf O(\mathrm{out},x_1,1)=y_1,\qquad
  \mathsf O(\mathrm{out},x_2,1)=y_2,\\[2pt]
 \text{bra side:}
 &\mathsf O(\mathrm{out},x_3,1)=y_3,\qquad
  \mathsf O(\mathrm{out},x_3,2)=z_3,
\end{array}
\]
All seven labels are distinct. The components are two directed paths of
length one and an outward two-star, rooted at $x_1,x_2,x_3$, respectively.
Let $\Gamma_{\mathrm p}$ be the rooted type of a single edge directed from
the root to a leaf, and let $\Gamma_{\mathrm s}$ be the rooted type of an
outward two-star. The representatives have no other vertices or edges.
For $\vec\Gamma=(\Gamma_{\mathrm p},\Gamma_{\mathrm s},\Gamma_{\mathrm s})$,
use one copy $H_1$ of the path representative and two disjoint copies
$H_2,H_3$ of the star representative, all on disjoint vertex sets.
Figure~\ref{fig:qbd-local-realization} shows the local realization
\[
 \rho(x_j)=o_j\quad(j=1,2,3),\qquad
 \rho(y_j)=a_j\quad(j=1,2,3),\qquad
 \rho(z_3)=b_3.
\]
The vertex $b_2$ is unused, so a local realization need not use every vertex
of a representative. There are four local realizations: $\rho(y_2)$ can be
$a_2$ or $b_2$, and independently $y_3,z_3$ can be mapped to $a_3,b_3$ in
either order. Although $H_2$ and $H_3$ have the same rooted type, their
vertex sets are disjoint.

\begin{figure}[htbp]
 \centering
 \resizebox{0.96\linewidth}{!}{%
 \begin{tikzpicture}[
   x=1cm,y=1cm,font=\small,
   vertex/.style={circle,draw=black!75,fill=white,
                  minimum size=7.5mm,inner sep=1pt},
   root/.style={vertex,fill=blue!12,line width=1pt},
   unused/.style={vertex,fill=orange!18},
   edge/.style={-{Stealth[length=2.2mm]},line width=0.8pt},
   mapping/.style={-{Stealth[length=2.1mm]},blue!65!black,
                   dashed,line width=0.8pt},
   copybox/.style={draw=black!25,fill=black!2,rounded corners=3pt}
 ]
  \draw[copybox] (-0.65,-3.85) rectangle (2.15,-2.10);
  \draw[copybox] (2.80,-3.85) rectangle (6.60,-1.95);
  \draw[copybox] (7.40,-3.85) rectangle (11.20,-2.10);
  \node[anchor=west,font=\small\bfseries] at (-0.65,1.25)
    {Components of $D$};
  \node at (0.75,0.65) {$C_1$};
  \node at (5.35,0.65) {$C_2$};
  \node at (9.30,0.65) {$C_3$};
  \node[root]   (x1) at (0,0)    {$x_1$};
  \node[vertex] (y1) at (1.5,0)  {$y_1$};
  \node[root]   (x2) at (4.7,0)  {$x_2$};
  \node[vertex] (y2) at (6.0,0)  {$y_2$};
  \node[vertex] (y3) at (8,0)    {$y_3$};
  \node[root]   (x3) at (9.3,0)  {$x_3$};
  \node[vertex] (z3) at (10.6,0) {$z_3$};
  \draw[edge] (x1) -- node[above,font=\scriptsize] {$1$} (y1);
  \draw[edge] (x2) -- node[above,font=\scriptsize] {$1$} (y2);
  \draw[edge] (x3) -- node[above,font=\scriptsize] {$1$} (y3);
  \draw[edge] (x3) -- node[above,font=\scriptsize] {$2$} (z3);
  \node[root]   (o1) at (0,-2.8)    {$o_1$};
  \node[vertex] (a1) at (1.5,-2.8)  {$a_1$};
  \node[unused] (b2) at (3.4,-2.8)  {$b_2$};
  \node[root]   (o2) at (4.7,-2.8)  {$o_2$};
  \node[vertex] (a2) at (6.0,-2.8)  {$a_2$};
  \node[font=\scriptsize,text=black!70] at (3.4,-2.17) {unused};
  \node[vertex] (a3) at (8,-2.8)    {$a_3$};
  \node[root]   (o3) at (9.3,-2.8)  {$o_3$};
  \node[vertex] (b3) at (10.6,-2.8) {$b_3$};
  \draw[edge] (o1) -- (a1);
  \draw[edge] (o2) -- (a2);
  \draw[edge] (o2) -- (b2);
  \draw[edge] (o3) -- (a3);
  \draw[edge] (o3) -- (b3);
  \foreach \source/\target in {x1/o1,y1/a1,x2/o2,y2/a2,y3/a3,x3/o3,z3/b3}
    \draw[mapping] (\source.south) -- (\target.north);
  \node[blue!65!black,fill=white,inner sep=2pt] at (0,-1.35) {$\rho$};
  \node[blue!65!black,fill=white,inner sep=2pt] at (4.7,-1.35) {$\rho$};
  \node[blue!65!black,fill=white,inner sep=2pt] at (9.3,-1.35) {$\rho$};
  \node[align=center] at (0.75,-3.48)
    {$H_1$: type $\Gamma_{\mathrm p}$};
  \node[align=center] at (4.70,-3.48)
    {$H_2$: type $\Gamma_{\mathrm s}$};
  \node[align=center] at (9.30,-3.48)
    {$H_3$: type $\Gamma_{\mathrm s}$};
 \end{tikzpicture}%
 }
 \caption{A local realization of a diagram with three components.
 Blue-shaded vertices are roots. Solid arrows are directed edges; the numbers
 on the upper edges are the required outgoing adjacency-list positions.
 Dashed arrows show the map $\rho$ from abstract labels to vertices of
 the representative copies. The lower graphs have no adjacency-list position ordering yet.
 The orange vertex $b_2$ is not in the image of $\rho$: a local realization
 need not use every vertex of a representative. Although $H_2$ and $H_3$
 have the same rooted type, their vertex sets are disjoint.}
 \label{fig:qbd-local-realization}
\end{figure}

For a local realization $\rho\in\mathcal R(D,\vec\Gamma)$, let
$\kappa(D,\vec\Gamma,\rho)$ be the probability that independent uniform
valid adjacency-list position orderings satisfy all requirements of $D$ under $\rho$.
To compute this probability, fix a queried vertex and a direction,
let $\deg$ be its degree in that direction, and merge repeated identical
requirements. A $\bot$ requirement at adjacency-list position $i$ is feasible exactly when
$i>\deg$. If the remaining non-$\bot$ requirements prescribe $m$ distinct
neighbors at $m$ distinct adjacency-list positions in $\{1,\ldots,\deg\}$, their probability is
\[
 \frac{(\deg-m)!}{\deg!}=\frac{1}{(\deg)_m}.
\]
Requiring different neighbors at the same adjacency-list position, the same neighbor at
different adjacency-list positions, or a non-$\bot$ answer at an adjacency-list position larger than the degree
gives probability zero. If any $\bot$ requirement is infeasible, the
probability is also zero. Adjacency-list position orderings at distinct vertex--direction
pairs are independent, so the product of these probabilities is
$\kappa(D,\vec\Gamma,\rho)$.
Define
\[
 \theta_D(\vec\Gamma)
 :=\sum_{\rho\in\mathcal R(D,\vec\Gamma)}
       \kappa(D,\vec\Gamma,\rho).
\]

Intuitively, $\theta_D(\vec\Gamma)$ is the expected number of ways to
realize $D$ inside the prescribed rooted discs after their adjacency-list position orderings
are chosen uniformly at random. The component roots are fixed at the
disc roots. Each injective placement satisfying the required directed
edges contributes the probability that the random adjacency-list position orderings also
satisfy all adjacency-list position and $\bot$ requirements. Thus $\theta_D$ is a weighted
count, rather than a probability, and may exceed one. On disjoint root
neighborhoods, this count depends only on their rooted types; averaging
over those types recovers the diagram coefficient up to the overlap error.

In the example above, the first path requirement has probability one,
the second has probability $1/2$, and the two star requirements jointly
have probability $1/(2)_2=1/2$. Hence each of the four local realizations
has weight $1\cdot(1/2)\cdot(1/2)=1/4$, and
$\theta_D(\vec\Gamma)=1$.

We first bound $\theta_D(\vec\Gamma)$ for every type tuple
$\vec\Gamma$, whether or not it occurs in $G$.
Choose a spanning tree of each component $C_j$, rooted at $x_j$.
In a local realization, the image of $x_j$ is fixed at $o_j$.
Assign the remaining labels in an order in which each parent precedes
its children. The image of each new label must be adjacent to the image
of its parent, so there are at most $2d$ choices. Injectivity and the
remaining directed-edge requirements can only reduce this number.
There are $s(D)$ non-root labels in total. Since each realization has
weight between zero and one, we obtain
\[
 0\le\theta_D(\vec\Gamma)
 \le |\mathcal R(D,\vec\Gamma)|
 \le (2d)^{s(D)}
 \le R_{d,T},
\]
using the choice $R_{d,T}=(2d)^{4T}$ from Step~2.

We next show that $F_D^G(\mathbf v)=\theta_D(\vec\Gamma)$ whenever
$\mathbf v=(v_1,\ldots,v_q)$ is a good root tuple of types
$\vec\Gamma$. Recall that
\[
 F_D^G(\mathbf v)
 =\sum_{\mathbf u\in\perm{V(G)\setminus\{v_1,\ldots,v_q\}}{s(D)}}
      \kappa_D^G(\mathbf v,\mathbf u).
\]
For each tuple $\mathbf u$ in this sum, let $\psi_{\mathbf u}$ be the
assignment of all labels of $D$ to vertices of $G$: it sends $x_j$ to
$v_j$ and sends the non-root labels to $\mathbf u$ in the order fixed
in Step~2. This assignment is injective by the definition of the sum.
If it fails any directed-edge requirement, its weight is zero.
We may therefore restrict the sum to assignments satisfying all these
requirements. We do not yet impose the adjacency-list position or $\bot$ requirements;
these are accounted for by the weights.

For each $j$, choose a rooted digraph isomorphism
\[
 h_j:\disc_r^G(v_j)\xrightarrow{\sim}H_j,
 \qquad h_j(v_j)=o_j.
\]
Consider an assignment $\psi_{\mathbf u}$ satisfying the directed-edge
requirements. Every label $x$ of $C_j$ is connected to $x_j$ by an
undirected constraint path of length at most $2T$. Its image under
$\psi_{\mathbf u}$ is therefore within distance $2T$ of $v_j$, and in
particular belongs to $\disc_r^G(v_j)$.
Thus we can define
\[
 \rho_{\mathbf u}(x):=h_j\bigl(\psi_{\mathbf u}(x)\bigr)
 \qquad\text{for each label }x\text{ of }C_j.
\]
This map sends each component root to the corresponding disc root and
preserves all required directed edges. It is injective within each
component because both $\psi_{\mathbf u}$ and $h_j$ are injective,
and images from different components cannot coincide because the copies
$H_j$ are disjoint. Hence
$\rho_{\mathbf u}\in\mathcal R(D,\vec\Gamma)$.

Conversely, given $\rho\in\mathcal R(D,\vec\Gamma)$, define
\[
 \psi(x):=h_j^{-1}\bigl(\rho(x)\bigr)
 \qquad\text{for each label }x\text{ of }C_j.
\]
This assignment sends $x_j$ to $v_j$ and satisfies every required
directed edge. It is injective within each component; it is also
injective across components because $\mathbf v$ is good, so the actual
discs are disjoint. Reading its non-root images in the fixed order gives
a unique tuple $\mathbf u$ in the restricted sum above.
The two constructions are inverse to each other, establishing a
bijection between its assignments and $\mathcal R(D,\vec\Gamma)$.

We now check that this bijection preserves the weights.
Let $w=\psi_{\mathbf u}(x)$ be the image of a queried label in $C_j$.
Since $w$ is within distance $2T$ of $v_j$, every incoming and outgoing
neighbor of $w$ is within distance $2T+1<r$ of $v_j$.
Consequently, $h_j$ identifies the full incoming and outgoing neighbor
sets of $w$ with those of $h_j(w)$ in $H_j$, and their respective
degrees are equal.
For either direction, an ordering $(w_1,\ldots,w_{\deg})$ of these
neighbors corresponds bijectively to the ordering
$(h_j(w_1),\ldots,h_j(w_{\deg}))$ in $H_j$.
Thus a uniform valid adjacency-list position ordering in $G$ is carried to a uniform valid
adjacency-list position ordering in the representative. The prescribed neighbor occurs
at position $i$ in one ordering exactly when its image occurs at position $i$
in the other. Likewise, a $\bot$ requirement holds in one graph exactly
when it holds in the other, since in both cases it means $i>\deg$.
Applying this correspondence independently at each queried
vertex--direction pair gives
\[
 \kappa_D^G(\mathbf v,\mathbf u)
 =\kappa(D,\vec\Gamma,\rho_{\mathbf u}).
\]
This includes assignments whose weight is zero because their adjacency-list position or
$\bot$ requirements are incompatible.

We can now reindex the restricted sum using the bijection above:
\begin{equation}
 \begin{aligned}
 F_D^G(\mathbf v)
 =\sum_{\substack{\mathbf u: \;
          \psi_{\mathbf u}\text{ satisfies all}\\
          \text{directed-edge requirements}}}
       \kappa_D^G(\mathbf v,\mathbf u)
 =\sum_{\rho\in\mathcal R(D,\vec\Gamma)}
       \kappa(D,\vec\Gamma,\rho)
 =\theta_D(\vec\Gamma).
 \end{aligned}
 \label{eq:qbd-step3-local-equality}
\end{equation}
Here the first sum uses the same domain for $\mathbf u$ as in the
definition of $F_D^G(\mathbf v)$, subject to the displayed edge requirements.

Finally, consider replacing $h_j$ by another rooted isomorphism $h'_j$.
For labels in $C_j$, the corresponding local realization changes by
composition with $h'_j\circ h_j^{-1}$, a rooted automorphism of $H_j$.
This operation permutes $\mathcal R(D,\vec\Gamma)$ and preserves each
weight by the same argument for the adjacency-list position ordering. Hence the sum defining
$\theta_D(\vec\Gamma)$ is unchanged. Similarly, replacing a representative
by a rooted isomorphic copy gives a bijection between the two sets of
local realizations and preserves their weights. Thus the equality above
does not depend on the chosen rooted isomorphisms or on the vertex names
used in the representatives.

We now average over the roots. For any tuple
$\mathbf v=(v_1,\ldots,v_q)\in V(G)^q$, including tuples with repeated
vertices, let $\vec\Gamma(\mathbf v)$ denote the tuple of rooted
isomorphism types of $\disc_r^G(v_1),\ldots,\disc_r^G(v_q)$.
Define the two distributions on type tuples by
\[
 \begin{aligned}
 H_{\vec\Gamma}^{\mathrm{without}}(G)
 &:=\Pr_{\mathbf v\sim\Unif(\perm{V(G)}{q})}
       [\vec\Gamma(\mathbf v)=\vec\Gamma],\\
 H_{\vec\Gamma}^{\mathrm{with}}(G)
 &:=\Pr_{\mathbf v\sim\Unif(V(G)^q)}
       [\vec\Gamma(\mathbf v)=\vec\Gamma].
 \end{aligned}
\]
These correspond to sampling the root vertices uniformly without and
with replacement, respectively. Until \eqref{eq:qbd-step3-type-average},
expectations over $\mathbf v$ use the without-replacement distribution
$\Unif\bigl(\perm{V(G)}{q}\bigr)$.
Grouping the root tuples according to their type tuples gives the exact identity
\[
 \begin{aligned}
 \E_{\mathbf v}\theta_D\bigl(\vec\Gamma(\mathbf v)\bigr)
 &=\sum_{\vec\Gamma\in(\mathcal D_{d,r})^q}
     \theta_D(\vec\Gamma)\,
     \Pr_{\mathbf v}[\vec\Gamma(\mathbf v)=\vec\Gamma]\\
 &=\sum_{\vec\Gamma\in(\mathcal D_{d,r})^q}
     \theta_D(\vec\Gamma)\,H_{\vec\Gamma}^{\mathrm{without}}(G).
 \end{aligned}
\]
Here $H_{\vec\Gamma}^{\mathrm{without}}(G)$ is not conditioned on the roots being good:
it includes all ordered tuples of distinct roots, including those with
overlapping discs.

By \eqref{eq:coefficient-as-root-average},
$c_D(G)=\E_{\mathbf v}F_D^G(\mathbf v)$.
Equation~\eqref{eq:qbd-step3-local-equality} shows that
$F_D^G(\mathbf v)-\theta_D(\vec\Gamma(\mathbf v))$ vanishes on every
good tuple. Therefore the entire difference of expectations comes from
bad tuples:
\[
 \begin{aligned}
 &\left|c_D(G)-\E_{\mathbf v}\theta_D\bigl(\vec\Gamma(\mathbf v)\bigr)\right|\\
 &\quad=\left|\E_{\mathbf v}\left[
    \bigl(F_D^G(\mathbf v)-\theta_D(\vec\Gamma(\mathbf v))\bigr)
    \I[\mathbf v\text{ is not good}]\right]\right|\\
 &\quad\leq\E_{\mathbf v}\left[
    \left|F_D^G(\mathbf v)-\theta_D(\vec\Gamma(\mathbf v))\right|
    \I[\mathbf v\text{ is not good}]\right]\\
 &\quad\leq R_{d,T}\Pr[\mathbf v\text{ is not good}]\\
 &\quad\leq R_{d,T}\binom{q}{2}\frac{s_{d,2r}}{n-q+1}
    =O_{d,T}(1/n).
 \end{aligned}
\]
The pointwise difference is at most $R_{d,T}$:
\eqref{eq:F-uniform-bound} places $F_D^G$ in $[0,R_{d,T}]$, and the
spanning-tree argument above gives the same bound for $\theta_D$.
The last inequality is \eqref{eq:qbd-step3-bad-roots}. Here $q=q(D)$ counts the connected
components of the constraint graph and satisfies $q\leq2T$.
Since $r=4T+3$, the numerator is bounded solely in terms of $d,T$,
while $n-q+1\geq n-2T+1$ grows linearly with $n$.
Combining the two displays yields
\begin{equation}
 c_D(G)
 =\sum_{\vec\Gamma\in(\mathcal D_{d,r})^q}
    \theta_D(\vec\Gamma)\,H_{\vec\Gamma}^{\mathrm{without}}(G)
   +O_{d,T}(1/n).
 \label{eq:qbd-step3-type-average}
\end{equation}

We next replace sampling without replacement by independent sampling.
For any probability vector $x=(x_\Gamma)_{\Gamma\in\mathcal D_{d,r}}$,
let $\vec\Lambda=(\Lambda_1,\ldots,\Lambda_q)\sim x^{\otimes q}$ mean
that the types $\Lambda_1,\ldots,\Lambda_q$ are sampled independently
from $x$. Define
\begin{equation}
 P_D(x)
 :=\sum_{\vec\Gamma\in(\mathcal D_{d,r})^q}
       \theta_D(\vec\Gamma)\,
       \Pr_{\vec\Lambda\sim x^{\otimes q}}[\vec\Lambda=\vec\Gamma].
 \label{eq:qbd-step3-polynomial}
\end{equation}
Thus $P_D(x)$ is the expected value of $\theta_D$ on independently sampled
types. Since
$\Pr_{\vec\Lambda\sim x^{\otimes q}}[\vec\Lambda=\vec\Gamma]
 =\prod_{j=1}^q x_{\Gamma_j}$,
it is a polynomial of degree at most $q$.

For $x=\freq_r(G)$, independent uniform vertex samples have independent
types with distribution $x$. Hence
\[
 P_D\bigl(\freq_r(G)\bigr)
 =\sum_{\vec\Gamma\in(\mathcal D_{d,r})^q}
       \theta_D(\vec\Gamma)\,H_{\vec\Gamma}^{\mathrm{with}}(G).
\]
This has the same form as the sum in
\eqref{eq:qbd-step3-type-average}, with
$H_{\vec\Gamma}^{\mathrm{without}}(G)$ replaced by
$H_{\vec\Gamma}^{\mathrm{with}}(G)$.

To compare these distributions, sample
$\mathbf v\sim\Unif(V(G)^q)$. Conditioned on no repeated vertex,
$\mathbf v$ is uniform among ordered tuples of distinct vertices.
We can therefore couple sampling with and without replacement by
keeping $\mathbf v$ when its entries are distinct, and otherwise
resampling a uniform tuple of distinct vertices. The resulting type
tuples can disagree only when the original sample contains a repeated
vertex. Consequently,
\[
 \sum_{\vec\Gamma\in(\mathcal D_{d,r})^q}
 \left|H_{\vec\Gamma}^{\mathrm{without}}(G)
       -H_{\vec\Gamma}^{\mathrm{with}}(G)\right|
 \leq 2\Pr_{\mathbf v\sim\Unif(V(G)^q)}
       [\text{some vertex is repeated}]
 \leq \frac{2\binom q2}{n}.
\]
The factor two appears because this is the sum of the absolute
probability differences over all type tuples.
Combining \eqref{eq:qbd-step3-type-average} with the expression for
$P_D(\freq_r(G))$ above, and factoring out the common weight
$\theta_D(\vec\Gamma)$, we obtain
\begin{equation}
 \begin{aligned}
 &\left|c_D(G)-P_D\bigl(\freq_r(G)\bigr)\right|\\
 &\quad\leq
 \left|\sum_{\vec\Gamma}\theta_D(\vec\Gamma)
   \left(H_{\vec\Gamma}^{\mathrm{without}}(G)
        -H_{\vec\Gamma}^{\mathrm{with}}(G)\right)\right|
   +O_{d,T}(1/n)\\
 &\quad\leq R_{d,T}\sum_{\vec\Gamma}
   \left|H_{\vec\Gamma}^{\mathrm{without}}(G)
        -H_{\vec\Gamma}^{\mathrm{with}}(G)\right|
   +O_{d,T}(1/n)\\
 &\quad\leq \frac{2R_{d,T}\binom q2}{n}+O_{d,T}(1/n)\\
 &\quad\leq \frac{2R_{d,T}\binom{2T}{2}}{n}+O_{d,T}(1/n)
   =O_{d,T}(1/n).
 \end{aligned}
 \label{eq:qbd-step3-coefficient-error}
\end{equation}
Here both sums range over $\vec\Gamma\in(\mathcal D_{d,r})^q$.
The additive error is the overlap error already present in
\eqref{eq:qbd-step3-type-average}, and the second inequality uses
$0\leq\theta_D(\vec\Gamma)\leq R_{d,T}$.

\paragraph{Step 4: turn coefficient estimates into an acceptance bound.}
Combining \eqref{eq:exact-diagram-expansion} and \eqref{eq:qbd-step3-coefficient-error}, define
\[
  \calR_G
  :=
  \sum_D
  \bigl(c_D(G)-P_D(\freq_r(G))\bigr)\calU_D^{(n)}.
\]
The number of diagrams is \(O_{d,T}(1)\), and \Cref{lem:anchored-diagram} therefore gives
\begin{equation}\label{eq:final-rem-bound-new}
  \|\calR_G\|_{\str,T}\le O_{d,T}(1/n).
\end{equation}
Thus
\begin{equation}\label{eq:moment-expansion-new}
  \calM_T(\Sym(G))
  =
  \sum_D P_D(\freq_r(G))\calU_D^{(n)}+\calR_G.
\end{equation}

Finally, the coupling lemma allows us to couple one type sampled from
$x$ with one sampled from $y$ so that they disagree with probability
$\|x-y\|_1/2$. Repeating the coupling independently for the $q$ roots,
a union bound
shows that the two type tuples disagree with probability at most
$q\|x-y\|_1/2$. By \eqref{eq:qbd-step3-polynomial} and the bound
$0\leq\theta_D\leq R_{d,T}$, we obtain
\begin{equation}\label{eq:PD-Lipschitz}
 \begin{aligned}
  |P_D(x)-P_D(y)|
  &\leq \frac{R_{d,T}q}{2}\|x-y\|_1 \leq T R_{d,T}\|x-y\|_1
   =O_{d,T}(1)\|x-y\|_1.
 \end{aligned}
\end{equation}
Subtracting \eqref{eq:moment-expansion-new} for \(G\) and \(H\), and using the finite number of diagrams, \eqref{eq:universal-diagram-norm}, \eqref{eq:final-rem-bound-new}, and \eqref{eq:PD-Lipschitz}, gives
\[
  \|\calM_T(\Sym(G))-\calM_T(\Sym(H))\|_{\str,T}
  \le
  O_{d,T}(1)\|\freq_r(G)-\freq_r(H)\|_1
  +O_{d,T}(1/n).
\]
This proves the lemma.
\end{proof}

\begin{proof}[Proof of Lemma~\ref{lem:uniform-locality}]
By the definition of the acceptance seminorm, for every \(T\)-query algorithm \(\calA\),
\[
 |p_\calA(G)-p_\calA(H)|
 \le
 \|\calM_T(\Sym(G))-\calM_T(\Sym(H))\|_{\str,T}.
\]
Apply Lemma~\ref{lem:finite-dimensional} with $r=4T+3$.
Its implicit constants and its lower threshold for $n$ depend only on
$d,T$. Choose $\Delta>0$ sufficiently small that the term
$O_{d,T}(1)\|\freq_r(G)-\freq_r(H)\|_1$ is at most $\eta/2$
whenever $\|\freq_r(G)-\freq_r(H)\|_1\leq\Delta$.
Then choose $n_0$ above that threshold and sufficiently large that
the additive $O_{d,T}(1/n)$ term is at most $\eta/2$ for all $n\geq n_0$.
These choices depend only on $d,T,\eta$, and prove the lemma.
\end{proof}

\subsection{Constructing the classical tester}
The locality lemma applies to every quantum algorithm. For a property tester, its consequence is especially simple: an acceptance gap forces a gap between local frequency vectors.

\begin{lemma}[Disc-frequency separation]\label{lem:separation}
Assume \(\Pi\) has a two-sided-error \(T\)-query quantum tester in the bidirectional model for proximity parameter \(\eps\).  Then there exist constants \(r=r(d,T)\), \(\Delta>0\), and \(n_0\) such that for all \(n\ge n_0\), for every \(G\in\Pi_n\) and every \(n\)-vertex \(d\)-bounded digraph \(H\) that is \(\eps\)-far from \(\Pi_n\),
\[
  \|\freq_r(G)-\freq_r(H)\|_1>\Delta .
\]
\end{lemma}

\begin{proof}
Let \(\calA\) be the assumed quantum tester and symmetrize it as above.  The symmetrized tester still accepts every \(G\in\Pi_n\) with probability at least \(2/3\) and accepts every \(\eps\)-far \(H\) with probability at most \(1/3\).  Apply Lemma~\ref{lem:uniform-locality} with \(\eta=1/6\), obtaining \(r,\Delta,n_0\).  If some yes-instance \(G\) and far instance \(H\) had \(\|\freq_r(G)-\freq_r(H)\|_1\le\Delta\), then their acceptance probabilities would differ by at most \(1/6\), contradicting the acceptance gap at least \(1/3\).  Therefore their frequency vectors must be at distance strictly larger than $\Delta$.
\end{proof}

\begin{proof}[Proof of Theorem~\ref{thm:qbd-classical}]
The classical-to-quantum direction is immediate, so it suffices to prove the quantum-to-classical direction.  Let \(\calA\) be a \(T\)-query quantum bidirectional tester for \(\Pi\).  By Lemma~\ref{lem:separation}, there exist constants \(r,\Delta,n_0\) such that yes-instances and \(\eps\)-far instances are \(\Delta\)-separated in \(\ell_1\) distance between their normalized radius-\(r\) disc-frequency vectors, for every \(n\ge n_0\).

For \(n<n_0\), a classical bidirectional tester can query the whole graph by querying all incoming and outgoing slots, using \(O(dn_0)\) queries, which is a constant.  Assume \(n\ge n_0\).

The classical tester samples
\[
  s=O\!\left(\frac{|\cD_{d,r}|^2\log(2|\cD_{d,r}|)}{\Delta^2}\right)
\]
vertices independently and uniformly from \([n]\).  Around each sampled vertex it performs breadth-first search in both directions, stopping the discovery of vertices at distance $r$. It retains all answers obtained during the search and queries any remaining incoming and outgoing slots of vertices in this set, including those at distance exactly $r$, and retains the edges with both endpoints in the set. This recovers the induced disc, including edges between boundary vertices. Each sample costs at most $2d s_{d,r}=O_{d,r}(1)$ queries.  Let \(\widehat f\) be the empirical distribution of the sampled rooted disc types.  Hoeffding's inequality bounds each coordinate error by $\Delta/(10|\cD_{d,r}|)$ with sufficiently high probability; the logarithmic factor in $s$ allows a union bound over all coordinates. Consequently, with probability at least \(5/6\),
\begin{equation}\label{eq:classical-empirical-good}
  \|\widehat f-\freq_r(G)\|_1\le \Delta/10.
\end{equation}

The tester accepts iff there exists an \(n\)-vertex \(d\)-bounded digraph \(G'\in\Pi_n\) such that
\[
  \|\widehat f-\freq_r(G')\|_1\le \Delta/3.
\]
For fixed $n$, this is a decision rule depending on $\Pi_n$ and the sampled types. It may require an arbitrarily large lookup table or computation; the theorem bounds oracle queries and makes no efficiency claim for this final step.

If \(G\in\Pi_n\), then on the good event \eqref{eq:classical-empirical-good}, choosing \(G'=G\) makes the tester accept.  If \(G\) is \(\eps\)-far from \(\Pi_n\) and the tester accepts on the good event, then for some \(G'\in\Pi_n\),
\[
  \|\freq_r(G)-\freq_r(G')\|_1
  \le
  \|\freq_r(G)-\widehat f\|_1+
  \|\widehat f-\freq_r(G')\|_1
  \le \Delta/10+\Delta/3<\Delta,
\]
contradicting Lemma~\ref{lem:separation}.  Thus the tester has two-sided error at most \(1/6\) after adjusting constants, and it uses a number of bidirectional queries depending only on $d,\eps,T$. Since $T=T(\eps,d)$, this is $O_{\eps,d}(1)$.  This proves the quantum-to-classical implication.  Since every classical randomized tester can be simulated quantumly, the two constant-query classes coincide.
\end{proof}

\begin{remark}[Dependence on the number of queries]
The radius, sampling accuracy, and threshold $n_0$ may depend rapidly on $T$. The conclusion concerns $T$ independent of $n$; it does not give a comparable simulation when $T$ grows with $n$. Both the original quantum tester and the resulting classical neighborhood explorations may be adaptive.
\end{remark}

\section{Warm-Up for the Quantum Unidirectional Model: Counting In-Degree-\texorpdfstring{$k$}{k} Vertices}
\label{sec:star}
We first estimate $\cnt_k=|\{v:\din(v)=k\}|$, for $k\in[d]$.
This warm-up isolates the successive sampling idea used for general discs. We assume $d\geq1$; when $d=0$, the list of positive in-degree counts is empty and no queries are needed.
\begin{assumption}\label{asmp:Omega(n)}
The input is a $d$-bounded-degree digraph with $d\geq1$, $\delta\in(0,1)$ is fixed,
and $\cnt_k\geq\delta n$ for every $k\in[d]$.
\end{assumption}

\subsection{Estimating \texorpdfstring{$\cnt_k$}{cntk}}
\label{sec:cntSk}
Set $\delta'=\delta/(d!2^d d^{d+1})$, $\eps=\delta'/(24d)$,
$\eta=\delta'/(16d)$, and
\[
 t_i=n^{(2^{d-i}-1)/(2^d-1)}\quad(0\leq i\leq d).
\]
For $1\leq i\leq d$ define the positive constants
\begin{align*}
 a_i&=\delta(1-\eta)^{i-1}(1-\eps)^{i-1}
          \sum_{k=i}^d\frac{k!}{(k-i)!},\\
 b_i&=(1+\eta)^{i-1}(1+\eps)^{i-1}
          \sum_{k=i}^d\frac{k!}{(k-i)!},\\
 c_i&=\max\left\{\frac{8\pi\sqrt{db_i}}{\eps a_i},
                        \sqrt{\frac{4\pi^2d}{\eps a_i}}\right\}.
\end{align*}
The count has error allowance $\eps a_i t_{i-1}/2$; the remaining half of
the allowance absorbs rounding of the sample count.

\begin{algorithm}[ht]
\caption{Estimating the number of vertices of in-degree $k$}
\label{alg:ApproxSd}
\begin{algorithmic}[1]
\Procedure{EstVertices}{$G,n,d,\delta$}
\State Use the parameters above and the query threshold $B_\star$ in \Cref{eq:query-threshold-star}.
Initialize a global query counter $Q=0$. Before every graph-oracle call,
including calls within subroutines, return the zero vector if $Q=B_\star$;
otherwise increment $Q$ and make the call.
\State Set $R_0=V$ and initialize all edge multisets as empty.
\For{$i=1,\ldots,d$}
\State Let $f_i(v,s)=1$ exactly when the non-null edge at $(v,s)$ points into
$\Supp{R_{i-1}}$ and $(v,s)\notin\bigcup_{j<i}\Supp{T_j}$.
\State Set $\tilde X_i=\Count(f_i,\lceil c_i\sqrt{n/t_{i-1}}\rceil,1/(100d))$.
\label{line:estXi}
\If{$i<d$}
\State Set $\ell_i=\lfloor t_i\tilde X_i/t_{i-1}\rfloor$.
\State Run $\Grover(f_i)$ independently $\ell_i$ times to obtain $T_i$.
\label{line:Grover}
\State Let $R_i$ be the multiset of heads of the sampled edges.
\label{line:shrink}
\EndIf
\EndFor
\State \Return $\widetilde{\cnt}_k=\sum_{i=k}^d(-1)^{i-k}\binom{i}{k}
                         n\tilde X_i/(i!t_{i-1})$, for $k\in[d]$.
\EndProcedure
\end{algorithmic}
\end{algorithm}
No growth is performed at level $d$, since its samples would not be used.
The predicate is defined even if earlier samples contain repetitions.

\begin{theorem}\label{thm:Star}
Under \Cref{asmp:Omega(n)}, for sufficiently large $n$,
\Cref{alg:ApproxSd} has the deterministic quantum query bound
\[
 O_{d,\delta}\left(n^{1/2-1/(2(2^d-1))}\right).
\]
With probability at least $2/3$, its output satisfies
$\sum_{k=1}^d|\widetilde{\cnt}_k-\cnt_k|\leq\delta n$.
\end{theorem}

\subsection{Analysis of \texorpdfstring{\Cref{alg:ApproxSd}}{EstVertices}}
\label{sec:correctSd}
Write $X_i=|f_i^{-1}(1)|$ and let $S_{i,k}$ count distinct in-degree-$k$
vertices in $R_i$, with $S_{0,k}=\cnt_k$.
At each nonterminal level, the required conditions are: the count is accurate to
$\eps a_i t_{i-1}/2$; all sampled heads are distinct; and, writing
$\widehat S_{i,k}$ for the number of sampled edges whose head has in-degree
$k$ (with multiplicity),
\begin{equation}\label{eq:warm-concentration}
 \widehat S_{i,k}\in(1\pm\eta)
           \frac{\ell_i(k-i+1)S_{i-1,k}}{X_i}\quad(k\geq i).
\end{equation}
At level $d$ only counting accuracy is required.
The stopped run defined before \Cref{lem:query-threshold} stops after the
first inaccurate count, or after the first sampling step with a repeated
head or a violation of \Cref{eq:warm-concentration}.
All probability bounds below hold after freezing the entire preceding
history. Bounds for sampling also freeze the current accurate count.

If all earlier counts and samples satisfy these conditions, each vertex
in $R_{i-1}$ has exactly one distinct
incoming edge in each of $T_1,\ldots,T_{i-1}$. Indeed the supports of the
$R_j$ are nested, and distinct sampled heads give one edge per retained
vertex at each level; the predicate excludes all earlier edges.
Consequently,
\begin{equation}\label{eq:warm-X}
 X_i=\sum_{k=i}^d(k-i+1)S_{i-1,k}.
\end{equation}
For all sufficiently large $n$, an accurate count and
$a_i t_{i-1}\leq X_i\leq b_i t_{i-1}$ imply, whenever $i<d$,
\begin{equation}\label{eq:warm-rounding}
 (1-\eps)\frac{t_i}{t_{i-1}}
 \leq\frac{\ell_i}{X_i}
 \leq(1+\eps)\frac{t_i}{t_{i-1}}.
\end{equation}
The lower bound uses $1\leq\eps a_i t_i/2$; this is possible for every
$i<d$ because $t_i\to\infty$. The upper bound follows directly from the
floor. Combining \Cref{eq:warm-X,eq:warm-concentration,eq:warm-rounding}
inductively gives
\begin{equation}\label{eq:warm-recursion}
 X_i\in (1\pm\eta)^{i-1}(1\pm\eps)^{i-1}
      \frac{t_{i-1}}{n}\sum_{k=i}^d\frac{k!}{(k-i)!}\cnt_k,
 \qquad a_i t_{i-1}\leq X_i\leq b_i t_{i-1}.
\end{equation}
Here and below paired interval factors mean the corresponding lower and
upper endpoints, with the inequalities interpreted separately.

\begin{lemma}\label{lem:pr}
For sufficiently large $n$, the stopped run completes all $d$ levels
with all counting and sampling conditions satisfied with probability
at least $0.98$.
\end{lemma}
\begin{proof}
Fix the outcomes of all earlier steps satisfying the required conditions.
\Cref{cor:count} and the definition of
$c_i$ bound the two counting error terms by $\eps a_i t_{i-1}/4$ each.
The conditional counting failure probability is at most $1/(100d)$.
Given an accurate count, the samples are independent and uniform on the
fixed nonempty marked set. For any pair of samples the probability of a
common head is at most $d/X_i$. Thus a union bound gives failure at most
\[
 \binom{\ell_i}{2}\frac d{X_i}
 \leq C_i\frac{t_i^2}{t_{i-1}}
 =C_i n^{-1/(2^d-1)}.
\]
Also $\widehat S_{i,k}$ is binomial with mean
$\ell_i(k-i+1)S_{i-1,k}/X_i\geq c_{i,k}t_i$ for some positive constant,
by the inductive bound and $\cnt_k\geq\delta n$.
Chernoff's inequality (\Cref{lem:B}) gives failure probability $O_{d,\delta}(1/t_i)$
for \Cref{eq:warm-concentration}, summed over $k$.
The distinctness and concentration conditions concern the same samples, so their failure
probabilities are added, not multiplied. Sum the conditional probabilities
of the first failed condition over all levels. Counting contributes at most
$0.01$ and sampling contributes $o(1)$, since only $i<d$ is sampled.
Their sum is at most $0.02$ for sufficiently large $n$.
\end{proof}

\begin{lemma}\label{lem:norm}
If all counting and sampling conditions hold, then
$\sum_{k=1}^d|\widetilde{\cnt}_k-\cnt_k|\leq\delta n$.
\end{lemma}
Its proof, including the normalization of the binomial inversion, appears
in \Cref{sec:pfLem}.

\paragraph{Query bound.}\label{sec:QCSd}
Choose a universal constant $K$ large enough to bound the graph-query
overheads of the counting and verified search routines, and set
\begin{equation}\label{eq:query-threshold-star}
 \begin{split}
 L_\star&=K\sum_{i=1}^d
 \left[b(1/(100d))(c_i+1)
       +\mathbf1_{i<d}(1+\eps)b_i\sqrt{d/a_i}\,t_i\right]
                 \sqrt{n/t_{i-1}},\\
 B_\star&=\lceil100L_\star\rceil.
 \end{split}
\end{equation}
Let $T$ count the graph queries in the stopped run. Each counting step
uses at most the corresponding counting term in $L_\star$, including
when it returns the first inaccurate count. A sampling step is reached
only after all earlier conditions and the current counting accuracy hold.
For every fixed history with these properties,
$X_i\geq a_it_{i-1}>0$ and $\ell_i\leq(1+\eps)b_it_i$, so
\Cref{lem:grover} bounds its conditional expected number of queries by
the corresponding sampling term in $L_\star$. This bound includes all
queries in that step, even if its samples violate distinctness or
concentration. Steps not reached make no queries. Averaging over the
earlier outcomes and summing over levels gives $\E[T]\leq L_\star$. For
$a=1/2-1/(2(2^d-1))$, the largest counting cost and each nonterminal
$t_i\sqrt{n/t_{i-1}}$ are $O(n^a)$.
Apply \Cref{lem:query-threshold,lem:pr,lem:norm}: every execution uses at most $B_\star$ queries,
and a correct output has probability at least $0.98-0.01>2/3$.
This proves \Cref{thm:Star}, including $d=1$, when no sampling is needed.

\section{Estimating Counts of \texorpdfstring{$r$}{r}-Disc Isomorphism Types in the Quantum Unidirectional Model}
\label{sec:disc}
In this section, we shift our focus from estimating the number of vertices with in-degree $k$ to estimating the number of occurrences of any $d$-bounded-degree $r$-disc isomorphism type, without relying on strong assumptions such as those in \Cref{asmp:Omega(n)}.

\subsection{Estimating \texorpdfstring{$\cnt_\Gamma$}{cntGamma}}%

We now present our algorithm for estimating the counts of all isomorphism types $\{\cnt_\Gamma\}_{\Gamma \in \cD_{d,r}^*}$. The core procedure is described in \Cref{alg:ApproxP}, denoted as \textsc{EstDisc}.
However, for technical reasons, we can only provide guarantees when the input graph is assumed to be \emph{$\delta$-explicit} (see \Cref{def:explicit}).
We further employ a parameter filtering technique, ensuring that $G$ is $\delta$-explicit for some $\delta$ with high probability. We present the final version of the algorithm in \Cref{alg:ApproxP*}, denoted as \textsc{EstDisc*}.

We treat $d,r\geq1$ in the analysis below. For brevity, throughout this section write $m=m_{d,r}$ and $D=D_{d,r}$.
Let $\mu_\Gamma=\sum_{\Gamma'\succeq\Gamma}\mu_{\Gamma,\Gamma'}$,
where the integer extension counts are defined in \Cref{sec:isotype}.
Types are processed in a fixed order compatible with their number of edges.
All samples are stored with their indices and ordered edge histories.

\paragraph*{The key subroutine EstDisc} 
The basic algorithm \Cref{alg:ApproxP} iterates from $1$ to $m_{d,r}$.
At the $i$-th iteration, we deal with all types $\Gamma$ with $i$ edges.
Based on the $\Gamma_{[i-1]}$-instances collected previously, we create a Boolean function $f_\Gamma$ marking the edges that can extend any one of them and complete it to a $\Gamma$-instance. 
Then we apply quantum counting on it to get an estimate $\tilde{X}_\Gamma$ of $X_\Gamma$, the number of marked edges, and we repeat applying Grover search to sample $\ell_\Gamma$ many of them randomly where $\ell_\Gamma$ is proportional to the estimate $\tilde{X}_\Gamma$.
For $i<m$, these samples form $\cH_\Gamma$ for later iterations. At $i=m$, only the count is computed.
Notice $\cH_\Gamma$ only collects the $\Gamma$-instances whose edges are revealed with specific order.

After all iterations are finished, we solve a system of linear equations about $\{\widetilde{\cnt}_\Gamma\}_{\Gamma\in\cD_{d,r}^*}$ based on the estimates $\{\tilde{X}_\Gamma\}_{\Gamma\in\cD_{d,r}^*}$ and output them as the result.

\begin{algorithm}[ht]
\caption{Estimating $r$-disc counts when $G$ is $\delta$-explicit}
\label{alg:ApproxP}
\begin{algorithmic}[1]
\Procedure{EstDisc}{$G,n,d,r,\delta$}
\If{$d=0$ or $r=0$}\State \Return the empty list of nonempty counts.\EndIf
\State Set $\eps=\eps_{d,r}\delta$, $\eta=\eta_{d,r}\delta$ as in \Cref{eq:eps&eta},
and $t_i=n^{(2^{m-i}-1)/(2^m-1)}$ for $0\leq i\leq m$.
\State In increasing edge order, define
\[
 \overline{c^X_\Gamma}=\begin{cases}d,&i=1,\\
 m(\overline{c^X_{\Gamma_{[i-1]}}}+\eps\underline{c^X_{\Gamma_{[i-1]}}}),&i>1,
 \end{cases}\qquad
 \underline{c^X_\Gamma}=\frac{4i\eps(1+\eps)\overline{c^X_\Gamma}}{(1+2\eps)^i m!}.
\]
\State Set $c^M_\Gamma=\max\{8\pi\sqrt{d\overline{c^X_\Gamma}}/(\eps\underline{c^X_\Gamma}),
\sqrt{4\pi^2d/(\eps\underline{c^X_\Gamma})}\}$.
\State Set the query threshold $B_{\rm disc}$ as in \Cref{eq:query-threshold-disc} and initialize a global query counter $Q=0$.
Before every graph-oracle call, including calls within subroutines, return the zero vector if
$Q=B_{\rm disc}$; otherwise increment $Q$ and make the call.
\State Initialize $\tilde{\boldsymbol{x}}=0$, the skipped set $\mathcal K=\varnothing$,
and $\cH_{\Gamma_0}=\{G[v]:v\in V\}$.
\For{$i=1,\ldots,m$}
\For{each $\Gamma\in\cD_{d,r}^i\setminus\mathcal K$ in the fixed order}
\State Let $f_\Gamma(v,s)=1$ exactly when $(v,s)$ returns a non-null edge $e$
and some stored $H\in\cH_{\Gamma_{[i-1]}}$ satisfies $e\notin E(H)$ and $H+e\cong\Gamma$.
\State Set $\tilde X_\Gamma=\Count(f_\Gamma,
\lceil c^M_\Gamma\sqrt{n/t_{i-1}}\rceil,1/(100mD))$.
\If{$\tilde X_\Gamma\geq(1-\eps)\underline{c^X_\Gamma}t_{i-1}$}\label{line:0}
\State Set $\tilde{\boldsymbol{x}}_\Gamma=n\tilde X_\Gamma/t_{i-1}$.
\If{$i<m$}
\State Set $\ell_\Gamma=\lfloor t_i\tilde X_\Gamma/t_{i-1}\rfloor$ and obtain
$T_\Gamma$ by $\ell_\Gamma$ independent calls to $\Grover(f_\Gamma)$.\label{line:_Grover}
\State For each returned edge, extend the first eligible stored parent
(in index order) and append the edge to its recorded history; store these
instances, with their sample indices, in $\cH_\Gamma$.\label{line:_H}
\EndIf
\Else
\State Add every $\Gamma'$ with $\Gamma'_{[i]}=\Gamma$ to $\mathcal K$;
their coordinates in $\tilde{\boldsymbol{x}}$ remain zero.
\EndIf
\EndFor
\EndFor
\State \Return $M^{-1}\tilde{\boldsymbol{x}}$, where
$M_{\Gamma,\Gamma'}=\mu_{\Gamma,\Gamma'}$.\label{line:returnP}
\EndProcedure
\end{algorithmic}
\end{algorithm}
The raw count $\tilde X_\Gamma$ is retained even when the branch is
truncated. No search is required at the terminal level. The fixed
tie-breaking rule makes the algorithm well-defined outside the good
events; on a good history each marked edge has a unique eligible parent.
For $r=0$ (or $d=0$), the only disc type is $\Gamma_0$, so the list of
nonempty counts is empty and is returned using zero queries.

\paragraph{The final algorithm EstDisc*} The final algorithm \Cref{alg:ApproxP*} first chooses a list of error parameters (all at most the given one) and samples one of them uniformly at random. Then it invokes \Cref{alg:ApproxP} as a subroutine with the error parameter set as the sampled one.

\begin{algorithm}[ht]
    \caption{Estimating $r$-disc counts}
    \label{alg:ApproxP*}
    \begin{algorithmic}[1] 
        \Procedure{EstDisc*}{$G, n, d, r ,  \delta$}
            \If{$d=0$ or $r=0$}\State \Return the empty list of nonempty counts.\EndIf
            \State Let $\alpha,\beta$ be the constants in \Cref{clm:alpha&beta}, and set $\delta_i=(\beta/\alpha)^{i-1}\delta$ for $1\leq i\leq100D_{d,r}$.
            \State Choose an index $i$ from $1$ to $100D_{d,r}$ uniformly at random.
            \State \Return \textsc{EstDisc}$(G,n,d,r,\delta_i)$.
        \EndProcedure
    \end{algorithmic}
\end{algorithm}

We will prove the following performance guarantee of \Cref{alg:ApproxP*}.

\begin{restatable}{theorem}{algorithmfull}
\label{thm:disc*}    
Given an $n$-vertex $d$-bounded-degree digraph $G$, integers $d,r\geq1$ and an error parameter $\delta \in (0,1)$, for sufficiently large $n$, \Cref{alg:ApproxP*} outputs a list of estimates $
\{\widetilde{\cnt}_\Gamma\}_{\Gamma\in \cD_{d,r}^*}$ satisfying
\[
\sum_{\Gamma\in \cD_{d,r}^*}\abs{\widetilde{\cnt}_{\Gamma} - \cnt_{\Gamma}} \leq \delta n,
\]
with probability at least $5/9$. On every execution it uses at most
$O_{\delta,d,r}\left(n^{1/2-1/(2(2^{m_{d,r}}-1))}\right)$ quantum queries.
\end{restatable}

\paragraph{Outline of the remaining subsections}
To prove \Cref{thm:disc*}, we begin by establishing the performance guarantee of \Cref{alg:ApproxP}, namely \Cref{thm:disc}.
In \Cref{sec:explicit}, we define the notion of explicitness for input graphs with respect to a given error parameter.
\Cref{sec:pfP} presents the proof of \Cref{thm:disc}, assuming \Cref{clm:E^X} and \Cref{clm:Eintersection}.
We then prove \Cref{thm:disc*} in \Cref{sec:pfP*}.
Finally, in \Cref{sssec:pfproperty}, we provide full proofs for the claims used in \Cref{sec:pfP}.

\subsection{Explicitness of input graph with respect to error parameter}
\label{sec:explicit}
Since we no longer assume conditions like \Cref{asmp:Omega(n)}, we must separately handle two cases for each $\Gamma \in \cD_{d,r}^*$: when $\Gamma$-instances are numerous (in which case we say $G$ is $\Gamma$-good), and when they are rare (we say $G$ is $\Gamma$-bad).
Due to estimation errors introduced by sampling (\textsc{Grover}) and counting (\textsc{Count}), it is not feasible to classify inputs based on a single explicit threshold.
Instead, we use two thresholds for each $\Gamma \in \cD_{d,r}^*$, which creates a gray area where $G$ is neither $\Gamma$-good nor $\Gamma$-bad.

\begin{definition}
\label{def:type}
    Given an $n$-vertex $d$-bounded-degree digraph $G$ and parameters $r,\delta$, define %
    $\boldsymbol{\gamma}\in\{1,0,*\}^{D_{d,r}}$ as its type vector where the $\Gamma$-entry is
    \[
       \boldsymbol{\gamma}_\Gamma = 
       \begin{cases}
           1, & \sum_{\Gamma'\succeq \Gamma}\mu_{\Gamma,\Gamma'}\cnt_{\Gamma'}\geq \alpha_\Gamma n;\\
           0, & \sum_{\Gamma'\succeq \Gamma}\mu_{\Gamma,\Gamma'}\cnt_{\Gamma'}< \beta_\Gamma n; \\
           *, &\text{otherwise}.
       \end{cases}
    \]
    Writing $i=|E(\Gamma)|$, the threshold values $\alpha_\Gamma,\beta_\Gamma$ depend only on $d,r$ and $\delta$ and are defined as
    \[
    \alpha_\Gamma = \frac{\underline{c^X_\Gamma}}{(1-\eps)^{i-1}} + 2i\mu_\Gamma \eta, \quad \beta_\Gamma = \frac{(1-2\eps)\underline{c^X_\Gamma}}{(1+\eps)^{i-1}} - 2i\mu_\Gamma \eta.
    \]
\end{definition}

\begin{definition}[$\delta$-explicit]
\label{def:explicit}
    Given an $n$-vertex $d$-bounded-degree digraph $G$ and parameters $r,\delta$, we say $G$ is $\delta$-explicit if $\boldsymbol{\gamma} \in\{1,0\}^{D_{d,r}}$, and inexplicit if there exists $\Gamma\in\cD_{d,r}^*$ such that $\boldsymbol{\gamma}_\Gamma = *$.
\end{definition}

\subsection{Analysis of \texorpdfstring{\Cref{alg:ApproxP}}{EstDisc}}
\label{sec:pfP}
In this section, we analyze \Cref{alg:ApproxP}, and prove the following performance guarantee. 

\begin{theorem}
\label{thm:disc}    
Given an $n$-vertex $d$-bounded-degree digraph $G$, integers $d,r\geq1$ and an error parameter $\delta \in (0,1)$, if $\boldsymbol{\gamma}\in\{0,1\}^{D_{d,r}}$, then for sufficiently large $n$, \Cref{alg:ApproxP} outputs a list of estimates $
\{\widetilde{\cnt}_\Gamma\}_{\Gamma\in \cD_{d,r}^*}$ satisfying
\[
\sum_{\Gamma\in \cD_{d,r}^*}\abs{\widetilde{\cnt}_{\Gamma} - \cnt_{\Gamma}} \leq \delta n,
\]
with probability at least $2/3$. On every execution it uses at most
$O_{\delta,d,r}\left(n^{1/2-1/(2(2^{m_{d,r}}-1))}\right)$ quantum queries.
\end{theorem}

In the following, we begin in \Cref{sec:isotype} by introducing several new properties of isomorphism types.
In \Cref{sssec:def}, we define key notations and events used throughout the analysis.
Next, in \Cref{sssec:property}, we establish how these events relate to the behavior of the algorithm under different types of input graphs, with full proofs deferred to \Cref{sssec:pfproperty}.
Finally, in \Cref{sssec:pfDisc}, we use these events to prove the correctness and analyze the query complexity of \Cref{alg:ApproxP}.

\subsubsection{Properties of isomorphism types}
\label{sec:isotype}
Given two isomorphism types $\Gamma \preceq \Gamma'$ with $i$ and $i'$ edges respectively, %
{we define the set of \emph{vertex-and-$i$-edges tuples}
$\cW_{\Gamma,\Gamma'}$ as the collection of all sequences consisting of a vertex and $i$ edges that generate $\Gamma_{[j]}$ sequentially for each $j \leq i$, i.e., }
\[
\cW_{\Gamma,\Gamma'} := \{(v,e_1,\dots,e_i)\in \{\rt(\Gamma')\}\times\perm{E(\Gamma')}{i}:\forall j\in[i], \Gamma'[\{{v}, e_1,\dots,e_j\}] \cong \Gamma_{[j]}\}.
\]

Then we further define
\[
\mu_{\Gamma,\Gamma'} := \abs{\cW_{\Gamma,\Gamma'}}.
\]
For $\Gamma\npreceq\Gamma'$, set $\cW_{\Gamma,\Gamma'}=\varnothing$
and $\mu_{\Gamma,\Gamma'}=0$.

We denote $W\times e$ as the vertex-and-$(i+1)$-edges tuple that appends $e$ at the end of $W$, i.e., if $W=(v,e_1,\dots,e_i)$, then $W\times e = (v,e_1,\dots,e_i,e)$.
Place tuples in $\cW_{\Gamma_{[j]},\Gamma'}$ (abbreviated as $\cW_j$ when context is clear) at level $j$.
For any tuples $W_1,W_2$ located at adjacent levels, connect them if there exists an edge $e$ such that $W_1\times e =W_2$.
Observe that there is exactly one tuple at level $0$, and that for any tuple at level $j>0$, there exists exactly one tuple at level $j-1$ connected to it.
Thus, under these rules, the tuples form a tree structure. (See \Cref{fig:tree} as an example.) 

Then we partition $\cW_{\Gamma,\Gamma'}$ into some equivalence classes, for $W\in \cW_{\Gamma,\Gamma'}$, it belongs to
\[
[W] = \{\phi(W):\phi\in\Phi_{\Gamma'}\}, \quad \text{where } \phi(W)=(\phi(v),\phi(e_1),\dots,\phi(e_i)) \text{ for } W=(v,e_1,\dots,e_i).
\]
{That is, $[W]$ contains all the vertex-and-$i$-edges tuples that can be obtained from $W$ via some automorphism $\phi$ in $\Phi_{\Gamma'}$.}
{
Equivalence classes are extremely useful, since each $\Gamma$‑instance can be assigned to a unique equivalence class, thereby enabling us to compute exactly how much over‑counting is incurred.
}

Define $\kappa_{[W_1],[W_2]}$ to be the number of edges that extend a fixed tuple $W_1$ to a tuple in $[W_2]$, i.e., %
\[
\kappa_{[W_1],[W_2]} = \abs{\{e\in E(\Gamma'): W_1\times e\in [W_2]\}}.
\]
Note that $\kappa_{[W_1],[W_2]}>0$ only when $W_2$ contains exactly one more edge than $W_1$, as otherwise, there is no edge $e$ with $W_1\times e\in[W_2]$. And $\kappa_{[W_1],[W_2]}$ is well-defined; {that is, it depends only on the equivalence classes of $[W_1],[W_2]$, and not on the specific representatives $W_1$ and $W_2$ themselves.} %
This holds because for any $[W_1]=[W_1']$, there exists a bijection $\phi\in\Phi_{\Gamma'}$ such that $\phi(W_1)=W_1'$. Then, for each edge $e$ such that $W_1\times e\in[W_2]$, there is a unique corresponding edge $\phi(e)$ such that $W_1'\times \phi(e)=\phi(W_1\times e)\in [W_2]$.  This implies that $\kappa_{[W_1],[W_2]} \leq \kappa_{[W_1'],[W_2]}$.
By a symmetric argument, we also have $\kappa_{[W_1],[W_2]} \geq \kappa_{[W_1'],[W_2]}$, and hence $\kappa_{[W_1],[W_2]} = \kappa_{[W_1'],[W_2]}$.

\subsubsection{Useful definitions}
\label{sssec:def}

\paragraph{Execution used in the analysis.}
All stage variables and probability statements below refer to the stopped
run from \Cref{lem:query-threshold}. It follows \Cref{alg:ApproxP} without
the query threshold and stops after the first inaccurate count, or after
the first completed sampling step whose concentration or separation
condition fails. These conditions are defined in \Cref{def:_events}.
The results of the failing step are retained. We call a type \emph{visited}
when its counting step is reached; this step always terminates.
A sampling step is reached only after an accurate count passes the
threshold. In that case
\[
 X_\Gamma\geq(1-3\eps/2)\underline{c^X_\Gamma}t_{i-1}>0,
\]
so every Grover call terminates almost surely. Thus this execution is
well-defined even on outcomes that fail a required condition.
The query threshold of the actual algorithm is handled separately in
\Cref{sssec:pfDisc}.

\paragraph{Counts within each full disc type.}
Whenever the sampling step at Line~\ref{line:_H} is performed, partition
its indexed multiset $\cH_\Gamma$ by the full disc type at each root.
For $\Gamma'\succeq\Gamma$, write
\[
\cH_{\Gamma,\Gamma'} = \{H\in\cH_\Gamma: \disc_r(\rt(H)) \cong \Gamma'\},
\quad S_{\Gamma,\Gamma'}=\abs{\Supp{\cH_{\Gamma,\Gamma'}}}.
\]
Moreover, we partition $\cH_{\Gamma,\Gamma'}$ according to the equivalence classes of $\cW_{\Gamma,\Gamma'}$.
Each stored instance $H$ carries the order in which its edges were revealed,
and $\Gamma'[W]$ carries the edge order of the tuple $W$.
In comparisons $H\cong_\phi\Gamma'[W]$ below, the restriction of $\phi$
is required to send the $j$-th recorded edge of $H$ to the $j$-th edge of $W$.
For $H+e$, the new edge is appended to this order.
Formally, for each $W\in \cW_{\Gamma,\Gamma'}$ we denote
\begin{equation}
\label{eq:defS}
    \cH^{[W]}_{\Gamma,\Gamma'} = \{H\in\cH_{\Gamma,\Gamma'}:\exists \phi, \disc_r(\rt(H)) \cong_\phi \Gamma' \text{ and } H\cong_{\phi}\Gamma'[W]\},
\quad   S^{[W]}_{\Gamma,\Gamma'}=\abs{\Supp{\cH^{[W]}_{\Gamma,\Gamma'}}}.
\end{equation}

If an instance belongs to the classes of both $W_1$ and $W_2$, composing
the two full-disc isomorphisms gives an automorphism sending $W_1$ to
$W_2$. Thus $[W_1]=[W_2]$, and these classes form a partition.
Only types $\Gamma'\succeq\Gamma$ can occur. At the initial type, set
$S_{\Gamma_0,\Gamma'}=S^{[W_0]}_{\Gamma_0,\Gamma'}=\cnt_{\Gamma'}$.

We also denote the set of edges that complete $\Gamma$-instances while the actual type is of $\Gamma'$ by
\[
    E^{[W]}_{\Gamma,\Gamma'} = \{e\in E(G):\exists W_{i-1}\in \cW_{\Gamma_{[i-1]},\Gamma'}, \exists H \in\cH^{[W_{i-1}]}_{\Gamma_{[i-1]},\Gamma'}, \exists \phi, \disc_r(\rt(H))\cong_\phi\Gamma', H+ e\cong_\phi \Gamma'[W] \},
\]
When sampling is performed, let $T^{[W]}_{\Gamma,\Gamma'}$ be the
submultiset of sampled edges in this set, retaining multiplicities, and
let $\widehat S^{[W]}_{\Gamma,\Gamma'}$ be its size:
\[
    T^{[W]}_{\Gamma,\Gamma'} = T_\Gamma\cap E^{[W]}_{\Gamma,\Gamma'},
    \quad \widehat{S}^{[W]}_{\Gamma,\Gamma'} = \abs{T^{[W]}_{\Gamma,\Gamma'}}.
\]
Under the separation conditions below, this sample count equals the
number $S^{[W]}_{\Gamma,\Gamma'}$ of distinct stored instances.

\begin{definition}
\label{def:_events}
For $\Gamma\in\cD_{d,r}^i$, $1\leq i\leq m$, define the events below
on the stopped run. The count conditions are evaluated only when
$\Gamma$ is visited, and the three sampling conditions are evaluated
only when its sampling step is performed. Each event is defined to hold
when its corresponding step is not performed. In particular, all
sampling events hold by convention after a failed threshold check or
at $i=m$, without assigning values to unused sampling variables.
Statements about $X_\Gamma$, $\tilde X_\Gamma$ or sampled instances
always concern executions that reach their defining step.

\begin{itemize}
\item Write $\cH_\Gamma=(H_a)_{a=1}^{\ell_\Gamma}$ as an indexed multiset.
$\cE^H_\Gamma$ is the event $\dist(\rt(H_a),\rt(H_b))>2r$ for every pair
of distinct indices $a\neq b$. In particular, identical repeated samples
violate the event.
This condition makes the full $r$-discs of distinct sampled roots disjoint.
\item $\cE^{\tilde{X}}_\Gamma: |\tilde{X}_\Gamma - X_\Gamma| \leq \tfrac12\eps \underline{c^X_\Gamma} t_{i-1}$. 
Half the relative error allowance is reserved for rounding the sample size.
\item $\cE^{X<}_\Gamma: X_\Gamma \leq \overline{c^X_\Gamma} t_{i-1}$. This event says that $X_\Gamma$ is upper bounded by a constant factor of $t_{i-1}$. 
\item $\cE^{X>}_\Gamma: X_\Gamma \geq \underline{c^X_\Gamma} t_{i-1}$. This event says that $X_\Gamma$ is lower bounded by a constant factor of $t_{i-1}$. 
\item $\cE^{\widehat{S}}_\Gamma : \widehat{S}^{[W]}_{\Gamma,\Gamma'} \in (1\pm\eps)\frac{t_i}{t_{i-1}} \sum_{W_{i-1}}\frac{\kappa_{[W_{i-1}],[W]}}{\abs{[W_{i-1}]}} S^{[W_{i-1}]}_{\Gamma_{[i-1]},\Gamma'} \pm\eta t_i \quad\forall \Gamma'\succeq \Gamma \text{ and } W\in \cW_{\Gamma,\Gamma'}$. 
The sum runs over $W_{i-1}\in\cW_{\Gamma_{[i-1]},\Gamma'}$.
This bounds each sampled class count using the stored parent counts.
\item $\cE^{{S}}_\Gamma: \widehat{S}^{[W]}_{\Gamma,\Gamma'}= {S}^{[W]}_{\Gamma,\Gamma'}\quad\forall \Gamma'\succeq \Gamma \text{ and } W\in \cW_{\Gamma,\Gamma'}$. This event ensures the equivalence between ${S}^{[W]}_{\Gamma,\Gamma'}$ and $\widehat{S}^{[W]}_{\Gamma,\Gamma'}$. Combining it with $\cE^{\widehat{S}}_\Gamma$, we can obtain an iterative bound of ${S}^{[W]}_{\Gamma,\Gamma'}$.
\item $\cE^\eq_\Gamma : X_\Gamma \in (1\pm\eps)^{i-1} \left(\frac{t_{i-1}}{t_0} \sum_{\Gamma'\succeq \Gamma} \mu_{\Gamma,\Gamma'} \cnt_{\Gamma'} \pm 2i\mu_{\Gamma} \eta t_{i-1}\right)$.
This relates the marked-edge count to the full disc counts.
\item $\cE^\pass_\Gamma: \tilde{X}_\Gamma \geq (1-\eps)\underline{c^X_\Gamma }t_{i-1}$. This is the threshold check at Line~\ref{line:0}.
\end{itemize}    
\end{definition}

Write $\cE^X_\Gamma=\cE^{X<}_\Gamma\cap\cE^{X>}_\Gamma$ for the
two-sided bound on the marked-edge count, and set
\[
 \cE_\Gamma=\cE^H_\Gamma\cap\cE^{\tilde X}_\Gamma\cap\cE^X_\Gamma
 \cap\cE^S_\Gamma\cap\cE^{\widehat S}_\Gamma\cap\cE^\eq_\Gamma,
 \qquad \cF_\Gamma=\bigcap_{j=1}^i\cE_{\Gamma_{[j]}}.
\]
At the initial type $\Gamma_{[0]}=\Gamma_0$, the elementary events,
$\cE_{\Gamma_0}$ and $\cF_{\Gamma_0}$ are defined to be sure events.
At the first stage, uniqueness of an edge's eligible parent follows
from the root's prescribed endpoint, rather than from separation of
the initial vertices.

Every stage claim below is uniform over each full history that reaches
the specified step and satisfies the stated conditions. Before counting,
this fixes $f_\Gamma$ and $X_\Gamma$. Before sampling, we also fix the
current accurate count, and hence $\ell_\Gamma$; probabilities then refer
only to the fresh Grover calls. For $i<m$, sufficiently large $n$ ensures
$1\leq\eps\underline{c^X_\Gamma}t_i/2$ for all types. At a sampling step
satisfying $\cE^X_\Gamma\cap\cE^{\tilde X}_\Gamma$, this gives
\begin{equation}\label{eq:disc-rounding}
 (1-\eps)\frac{t_i}{t_{i-1}}\leq\frac{\ell_\Gamma}{X_\Gamma}
 \leq(1+\eps)\frac{t_i}{t_{i-1}}.
\end{equation}
This is the same floor estimate as \Cref{eq:warm-rounding}.

\begin{figure}
    \centering
    \includegraphics[width=1.0\linewidth]{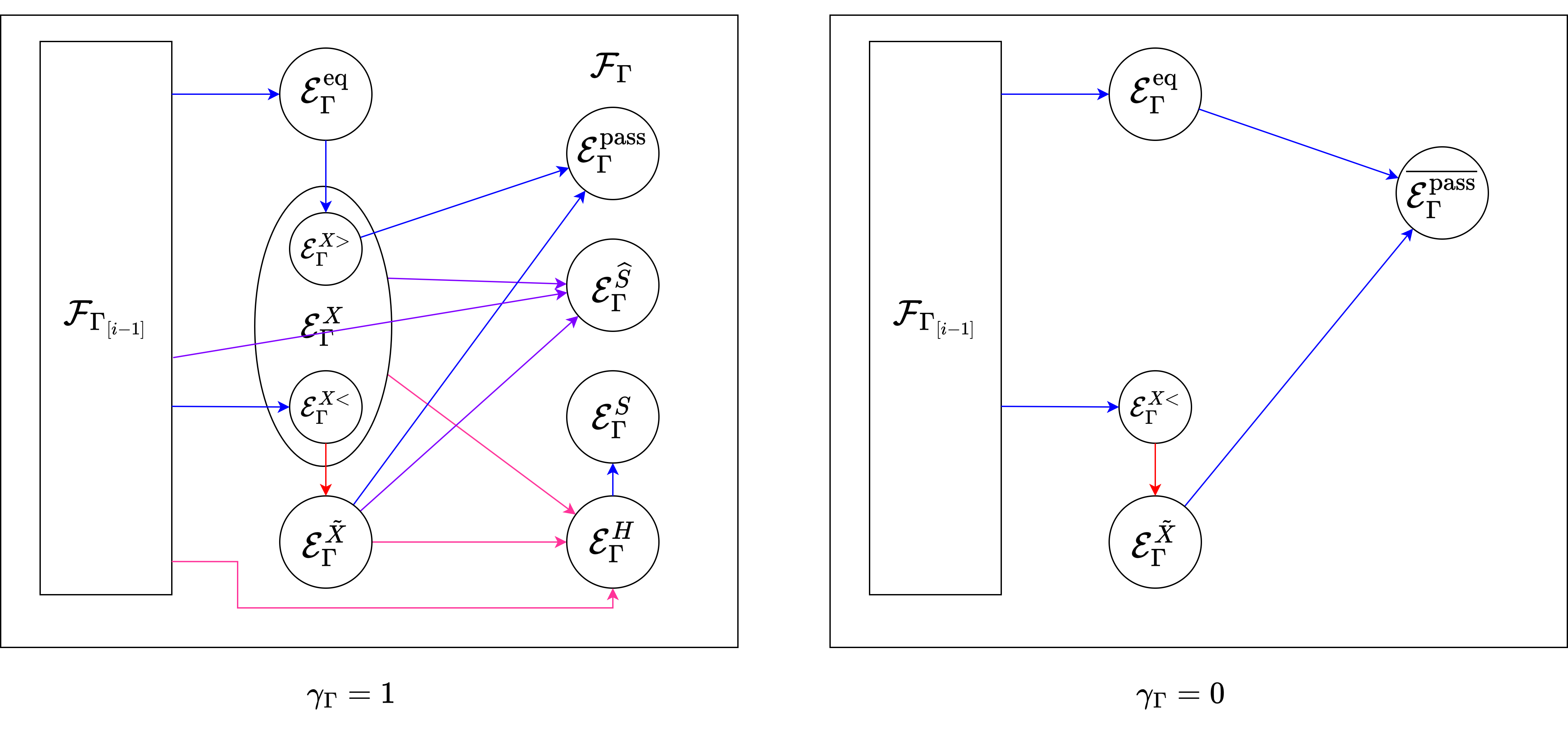}
    
    \caption{Illustration on how events depend on each other defined in \Cref{def:_events} when $\boldsymbol{\gamma}_\Gamma=1$ and $\boldsymbol{\gamma}_\Gamma=0$. The blue, red, purple and pink arrows mean event holds with probability at least $1$, $1-1/(100m_{d,r}D_{d,r})$, $1 - c^S_\Gamma n^{-(2^{m_{d,r}-i}-1)/(2^{m_{d,r}}-1)}$ and $1 - c^H_\Gamma n^{-1/(2^{m_{d,r}}-1)}$ respectively. Sampling arrows apply only to nonterminal stages and
are conditioned on the current accurate count. Within a stage their
failure probabilities are combined by a union bound, as detailed in
\Cref{clm:Eintersection}.}
    \label{fig:dep2}
\end{figure}

\subsubsection{Some central properties}
\label{sssec:property}

Fix a full entry history reaching $\Gamma\in\cD_{d,r}^i$ and satisfying $\cF_{\Gamma_{[i-1]}}$. Given an accurate count, the threshold check passes if $G$ is $\Gamma$-good, and truncates the branch if $G$ is $\Gamma$-bad. If $i<m$ and $G$ is $\Gamma$-good, the sampling conditions hold with high probability, conditional on this history and the current accurate count.
See \Cref{fig:dep2} as an illustration for \Cref{clm:E^X}.

\begin{restatable}{claim}{EventEX}
\label{clm:E^X}
    Fix a full entry history reaching a type $\Gamma\in\cD_{d,r}^i$
    and satisfying $\cF_{\Gamma_{[i-1]}}$. Then
    \begin{itemize}
        \item $\cE^\eq_\Gamma$ and $\cE^{X<}_\Gamma$ always hold;
        \item $\cE^{\tilde{X}}_\Gamma$ holds with probability at least $1-1/(100m_{d,r}D_{d,r})$;
        \item if $\boldsymbol{\gamma}_{\Gamma} = 1$, then
        \begin{itemize}
            \item $\cE^{X>}_\Gamma$ always holds;
            \item if $\cE^{\tilde{X}}_\Gamma$ also holds, then $\cE^\pass_\Gamma$ will always hold;
            \item for $i<m$, if $\cE^{\tilde X}_\Gamma$ also holds, then $\cE^{\widehat{S}}_\Gamma$ holds with probability at least $1-c^S_\Gamma/t_i$;
            \item for $i<m$, if $\cE^{\tilde X}_\Gamma$ also holds, then $\cE^H_\Gamma$ holds with probability at least
            \[1-c^H_\Gamma n^{-1/(2^{m_{d,r}}-1)};\]
            \item if $\cE^H_\Gamma$ also holds, then $\cE^S_\Gamma$ will always hold.
        \end{itemize}
        \item if $\boldsymbol{\gamma}_{\Gamma} = 0$, then 
        \begin{itemize}
            \item if $\cE^{\tilde{X}}_\Gamma$ also holds, then $\cE^\pass_\Gamma$ will never hold.
        \end{itemize}
    \end{itemize}
\end{restatable}

{
On a good execution, a rare type $\Gamma$ truncates precisely the
types $\Gamma'$ whose chosen prefix chain contains it,
$\Gamma'_{[|E(\Gamma)|]}=\Gamma$, including $\Gamma$ itself.
Thus every proper prefix of a visited type is abundant, while the current type may be abundant or rare.
}

The stopped run defined before \Cref{lem:query-threshold} stops after the
first count violating $\cE^{\tilde X}_\Gamma$, or after the first sampling
step violating $\cE^{\widehat S}_\Gamma$ or $\cE^H_\Gamma$.

\begin{restatable}{claim}{EventIntersection}\label{clm:Eintersection}
For a $\delta$-explicit graph and sufficiently large $n$, with probability
at least $0.98$ the stopped run completes the algorithm and satisfies all
of the following conditions:
every visited type is counted accurately; every visited abundant type
passes the threshold and satisfies its elementary events; every visited
rare type fails the threshold and has its prefix descendants skipped.
We call an execution satisfying these conditions a \emph{good execution}.
\end{restatable}

\subsubsection{Proof of \texorpdfstring{\Cref{thm:disc}}{EstDisc guarantee}}
\label{sssec:pfDisc}
The upper triangular matrix $M=(\mu_{\Gamma,\Gamma'})$ has positive
diagonal: each canonical edge ordering of $\Gamma$ is one of the tuples
counted by $\mu_{\Gamma,\Gamma}$. It is therefore invertible. Write
\[
 H=\|M^{-1}\|_1,\qquad U=d(2m)^m,\qquad V=Dm!,\qquad
 A_\Gamma=\sum_{\Gamma'\succeq\Gamma}\mu_{\Gamma,\Gamma'}\cnt_{\Gamma'}.
\]
Choose the constants used throughout the algorithm to be
\begin{equation}\label{eq:eps&eta}
\begin{split}
 \eps_{d,r}&=\min\left\{\frac1{100m},
                   \frac1{32HDm(U+V)}\right\},\\
 \eta_{d,r}&=\min\left\{\frac1{16HDmV},
                   \frac{d\eps_{d,r}}{4D(m!)^2}\right\}.
\end{split}
\end{equation}
In particular $\eps\leq1/(100m)$. \Cref{clm:para} implies
$0<\underline{c^X_\Gamma}<\overline{c^X_\Gamma}\leq U$ and
$\mu_\Gamma\leq V$.

\begin{lemma}\label{lem:_norm}
On a good execution as in \Cref{clm:Eintersection},
\[
 \|\tilde{\boldsymbol{x}}-M\boldsymbol{\cnt}\|_1
 \leq D\bigl(8m(U+V)\eps+4mV\eta\bigr)n.
\]
\end{lemma}
\begin{proof}
Consider first a type $\Gamma$ with $i$ edges whose entire prefix chain
is abundant. It is visited and passes its threshold. By
$\cE^\eq_\Gamma$ and $\cE^{\tilde X}_\Gamma$,
\[
 \left|\frac{n\tilde X_\Gamma}{t_{i-1}}-A_\Gamma\right|
 \leq 2m\eps A_\Gamma+4m\mu_\Gamma\eta n
                        +\frac\eps2\underline{c^X_\Gamma}n.
\]
Here $(1+\eps)^m\leq2$, $(1+\eps)^m-1\leq2m\eps$ and
$1-(1-\eps)^m\leq m\eps$; these follow from $m\eps\leq1/100$.
Since $A_\Gamma\leq Vn$, this error is bounded by the claimed
per-coordinate bound.

Otherwise let $\Lambda=\Gamma_{[j]}$ be the first rare prefix, possibly
$\Gamma$ itself. Its truncation sets $\tilde{\boldsymbol{x}}_\Gamma=0$.
Every ordered $\Gamma$-instance has a $\Lambda$ prefix, and every
$\Gamma'$ containing $\Gamma$ has $\mu_{\Lambda,\Gamma'}\geq1$.
Therefore
\[
 A_\Gamma\leq m!A_\Lambda<m!\beta_\Lambda n
 \leq m!\underline{c^X_\Lambda}n
 \leq8mU\eps n.
\]
The last inequality uses the definition of $\underline{c^X_\Lambda}$
and $1+\eps\leq2$. This bound is uniform over the first rare prefix,
so it also bounds every skipped descendant coordinate. Summing the
per-coordinate bounds proves the lemma.
\end{proof}
Multiplication by $M^{-1}$ and \Cref{eq:eps&eta} now give
\[
 \|\widetilde{\boldsymbol{\cnt}}-\boldsymbol{\cnt}\|_1
 \leq HD\bigl(8m(U+V)\eps+4mV\eta\bigr)n\leq\delta n.
\]

\paragraph{Query threshold.}\label{sec:QCP}
With a universal graph-query overhead constant $K$ as in the warm-up, set
\begin{equation}\label{eq:query-threshold-disc}
\begin{aligned}
 L_{\rm disc}&=K\sum_{i=1}^m\sum_{\Gamma\in\cD_{d,r}^i}
 \Bigl[b(1/(100mD))(c^M_\Gamma+1)
       +\mathbf1_{i<m}(1+\eps)\overline{c^X_\Gamma}
                       \sqrt{d/\underline{c^X_\Gamma}}\,t_i\Bigr]
                       \sqrt{n/t_{i-1}},\\
 B_{\rm disc}&=\lceil100L_{\rm disc}\rceil.
\end{aligned}
\end{equation}
Let $T$ count the graph queries in the stopped run. Every counting step
uses at most its counting term in $L_{\rm disc}$, even when its output is
the first inaccurate count. A sampling step is reached only after all
earlier conditions and the current counting accuracy hold. On an explicit
input its type is then abundant, so
$X_\Gamma\geq\underline{c^X_\Gamma}t_{i-1}>0$ and
$\ell_\Gamma\leq(1+\eps)\overline{c^X_\Gamma}t_i$.
For every fixed history with these properties, the conditional expected
number of queries in the entire sampling step is at most its sampling
term in $L_{\rm disc}$, even if the resulting samples violate concentration
or separation. Steps not reached and skipped types make no queries.
Averaging over the earlier outcomes and summing over types therefore
gives $\E[T]\leq L_{\rm disc}$.
There are constantly many types, and
\[
 L_{\rm disc}=O_{\delta,d,r}\left(n^{1/2-1/(2(2^m-1))}\right).
\]
\Cref{lem:query-threshold,lem:_norm} and \Cref{clm:Eintersection} give success probability
at least $0.98-0.01>2/3$ with at most $B_{\rm disc}$ queries on every execution. This
proves \Cref{thm:disc}. The counter remains a deterministic bound even
on inputs that are not explicit.

\subsection{Analysis of \texorpdfstring{\Cref{alg:ApproxP*}}{EstDisc*}}
Now we prove the following performance guarantee of \Cref{alg:ApproxP*}.
\label{sec:pfP*}

\algorithmfull*
\begin{proof}
    By \Cref{clm:alpha&beta}, we first choose two constants $\alpha,\beta$ only dependent on $d$ and $r$, but not dependent on $\delta$, such that for each $\Gamma\in\cD_{d,r}^*$,
    \[
    \alpha\delta\geq \alpha_\Gamma\geq \beta_\Gamma\geq \beta\delta.
    \]

    Then we choose $\delta_i = (\beta/\alpha)^{i-1} \delta$.
    Note that $\alpha_\Gamma, \beta_\Gamma$ are dependent on $\delta$ (see \Cref{def:type}).
    To distinguish between different $\delta_i$, we denote $\alpha_\Gamma$ as $\alpha_\Gamma^\delta$ and $\beta_\Gamma$ as $\beta_\Gamma^\delta$ respectively.
    Now we prove that for different $\delta_i,\delta_j$ ($i>j$), the intervals $[\beta_\Gamma^{\delta_i},\alpha_\Gamma^{\delta_i})$ and $[\beta_\Gamma^{\delta_j},\alpha_\Gamma^{\delta_j})$ are disjoint.
    \[
        [\beta_\Gamma^{\delta_i},\alpha_\Gamma^{\delta_i}) \subseteq [\beta \delta_i, \alpha\delta_i) = [(\beta/\alpha)^{i-1} \beta\delta, (\beta/\alpha)^{i-1} \alpha\delta).
    \]
    \[
        [\beta_\Gamma^{\delta_j},\alpha_\Gamma^{\delta_j}) \subseteq [\beta \delta_j, \alpha\delta_j) = [(\beta/\alpha)^{j-1} \beta\delta, (\beta/\alpha)^{j-1} \alpha\delta).
    \]
    \[
        (\beta/\alpha)^{i-1} \alpha\delta \leq (\beta/\alpha)^j \alpha\delta = (\beta/\alpha)^{j-1} \beta \delta.
    \]
    Therefore for each input graph $G$ and $\Gamma\in\cD_{d,r}^*$, there is at most one error parameter $\delta_i$ such that $\sum_{\Gamma'\succeq \Gamma}\mu_{\Gamma,\Gamma'}\cnt_{\Gamma'}\in[\beta_\Gamma n, \alpha_\Gamma n)$, i.e. $\boldsymbol{\gamma}_\Gamma=*$.
    That means there are at most $D_{d,r}$ error parameters such that there exists some $\Gamma$ such that $\boldsymbol{\gamma}_\Gamma=*$.
    In other words, if we choose an index $i$ from $1$ to $100D_{d,r}$ uniformly at random, then with probability at least $0.99$ we have $\boldsymbol{\gamma}\in \{0,1\}^{D_{d,r}}$.
    See \Cref{fig:table} as an illustration for explanation.

    By \Cref{thm:disc}, the total success probability is at least $2/3\times 0.99\geq 5/9$, and the query complexity is $O_{\delta,d,r}\left(n^{1/2-1/(2(2^{m_{d,r}}-1))}\right)$.
    Take the maximum of the graph-size thresholds and query-bound constants over
    the finite list of $100D_{d,r}$ choices. This gives a uniform
    deterministic query bound depending only on $\delta,d,r$, including
    choices for which the input is not explicit.
\end{proof}

    \begin{figure}[ht]
        \centering
        \includegraphics[width=0.8\linewidth]{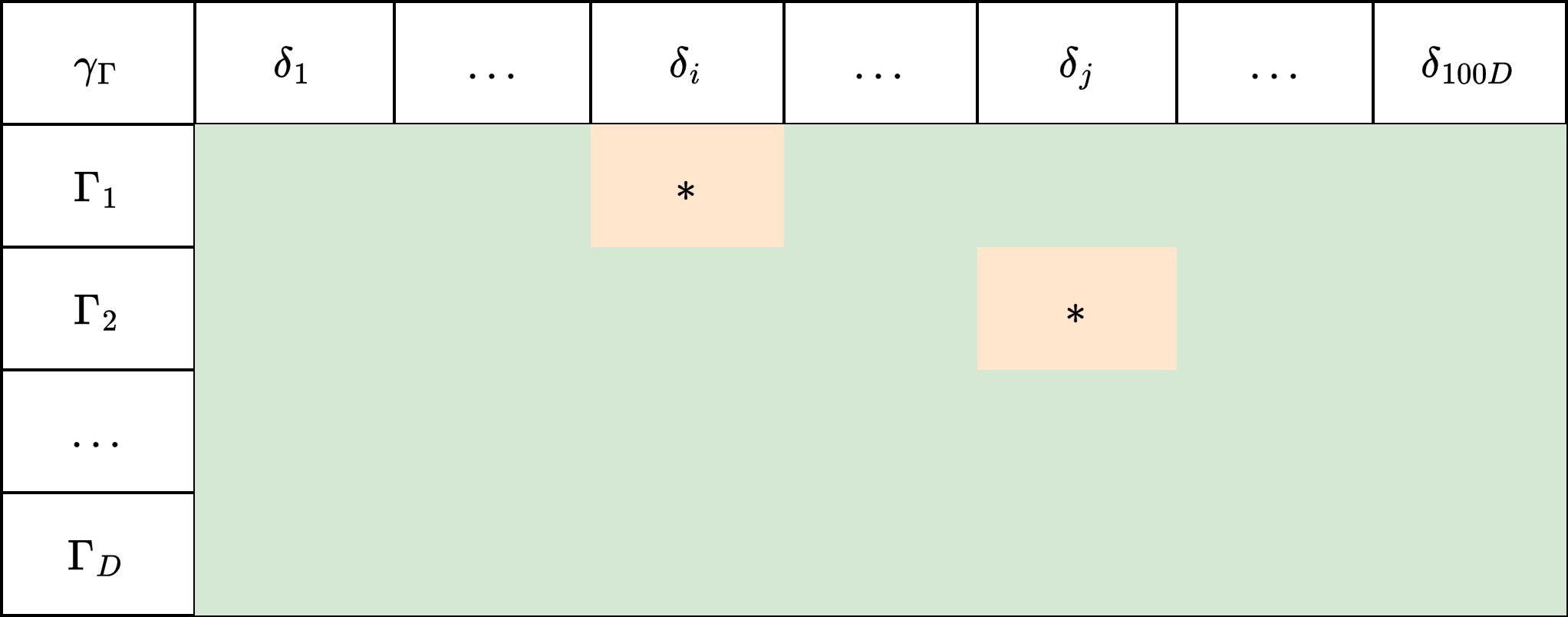}
        \caption{We choose $\delta_i$ properly, such that all intervals $[\beta_\Gamma^{\delta_i},\alpha_\Gamma^{\delta_i})$ are disjoint. Therefore, in each row, at most one element equals $*$ (marked as pink), and the other elements equal $0$ or $1$ (marked as green). That means that there are at most $D_{d,r}$ (abbreviated as $D$ in the figure) non-all-green columns, which is only $1\%$ of all columns. }
        \label{fig:table}
    \end{figure}

\subsection{Proof of central properties}
\label{sssec:pfproperty}

To prove \Cref{clm:E^X} and \Cref{clm:Eintersection},
we need the help of the following claims.

\subsubsection{Some useful claims}
{
For each  $i \in [m_{d,r}]$ and $\Gamma\in\cD_{d,r}^i$, if there are $\Theta(t_{i-1})$ edges that are able to complete some $\Gamma$-instance and
we only sample $\Theta(t_i)=o(\sqrt{t_{i-1}})$ of them, then these sampled edges should be far away from each other with high probability by birthday paradox.
That implies the obtained $\Gamma$-instances should also be far away from each other, that is, $\cH_\Gamma$ has pairwise separated roots.
}

\begin{restatable}{claim}{EventEHi}\label{clm:_E^H_i}
At a reached sampling step for $\Gamma\in\cD_{d,r}^i$, $i<m$, fix
the full history and current count satisfying
$\cE^X_\Gamma\cap\cE^{\tilde X}_\Gamma\cap
\cE^H_{\Gamma_{[i-1]}}$. Then the fresh samples satisfy
\[
\Pr[\cE^H_\Gamma]\geq1-c^H_\Gamma n^{-1/(2^m-1)},
\]
where $c^H_\Gamma$ depends only on $\delta,d,r$.
\end{restatable}

{
Since the function $f_\Gamma$ is dependent on the edges sampled in previous iterations, if the previous estimates are good, then $X_\Gamma=|f_\Gamma^{-1}(1)|$ is $t_{i-1}$-upper-bounded. 
}

\begin{restatable}{claim}{EventXGammai}\label{clm:XGammai}
    Fix an entry history reaching $\Gamma\in\cD_{d,r}^i$ in which
    all counts at its proper prefixes satisfy $\cE^{\tilde X}_{\Gamma_{[j]}}$.
    Then
    \[
    X_\Gamma \leq \overline{c^X_\Gamma} t_{i-1}
    \]
    with $\overline{c^X_\Gamma}$ as defined in \Cref{alg:ApproxP}.
\end{restatable}

{
Given a Boolean function $f_\Gamma$, by \Cref{lem:count}, quantum counting can return a good estimate $\tilde{X}_\Gamma$ of $X_\Gamma=|f_\Gamma^{-1}(1)|$ with high probability as long as an upper bound of $X_\Gamma$ is known.
}

\begin{restatable}{claim}{EventEXii}
\label{clm:_E^-X_i}
Fix an entry history reaching $\Gamma\in\cD_{d,r}^i$ and satisfying
$\cE^{X<}_\Gamma$. Over the fresh randomness of \Count,
\[
\Pr[\cE^{\tilde X}_\Gamma]\geq1-\frac1{100mD}.
\]
\end{restatable}

{
Recall that we collect some $\Gamma$-instances after performing Grover search.
Because of the uniformity of Grover search, we can bound the number of collected $\Gamma$-instances in each equivalence class via concentration inequalities.
}

\begin{restatable}{claim}{EventSii}
\label{clm:_E^S_i}
At a reached sampling step for $\Gamma\in\cD_{d,r}^i$, $i<m$, fix
the full history and current count satisfying
$\cE^H_{\Gamma_{[i-1]}}\cap\cE^{\tilde X}_\Gamma\cap\cE^X_\Gamma$.
Then
\[
 \Pr[\cE^{\widehat S}_\Gamma]\geq1-c^S_\Gamma/t_i,
\]
where the probability is over the fresh samples and $c^S_\Gamma$
depends only on $\delta,d,r$.
\end{restatable}

From the event $\cE^{\widehat{S}}_\Gamma$, we have bounded the variable $\widehat{S}^{[W]}_{\Gamma,\Gamma'}$ by $S^{[W_{i-1}]}_{\Gamma_{[i-1]},\Gamma'}$ well.
It remains to pass this bound to ${S}^{[W]}_{\Gamma,\Gamma'}$.
\begin{restatable}{claim}{EventS}
\label{clm:_E^S}
    At a completed sampling step for $\Gamma\in\cD_{d,r}^i$, $i<m$, if both
    $\cE^H_{\Gamma_{[i-1]}}$ and $\cE^H_\Gamma$ hold, then
    $\cE^S_\Gamma$ holds. At $i=1$ use endpoint uniqueness in place
    of a separation assertion about the initial vertices.
\end{restatable}

{
Based on the separation of $\cH_\Gamma$ (stated in \Cref{clm:_E^H_i}) and the concentration bound for each equivalence class (stated in \Cref{clm:_E^S_i}), we can construct a linear inequality between $X_\Gamma$ and the counts to be estimated $\{\cnt_{\Gamma'}\}_{\Gamma'\succeq \Gamma}$
}.

\begin{restatable}{claim}{EventEeqi}
\label{clm:_E^eq_i}
Fix an entry history reaching $\Gamma\in\cD_{d,r}^i$. If
$\cE^H_{\Gamma_{[j]}}\cap\cE^S_{\Gamma_{[j]}}\cap
\cE^{\widehat S}_{\Gamma_{[j]}}$ holds for every proper prefix
$1\leq j<i$, then $\cE^\eq_\Gamma$ holds.
\end{restatable}

\subsubsection{Proofs of these useful claims}
From here, we give the detailed proofs of all claims stated above.

\EventEHi*
\begin{claimproof}
Freeze the entire entry history and the current accurate count.
For $i=1$, the fixed one-edge type determines whether the root is the
head or tail of the sampled edge. At most $m_{d,2r+1}$ edges have their
corresponding root within distance $2r$ of a fixed root. Thus a union
bound over distinct sample indices gives failure probability at most
$\binom{\ell_\Gamma}{2}m_{d,2r+1}/X_\Gamma$.
For $i>1$, distinct parent samples have roots at distance greater than
$2r$. Their induced $r$-discs are disjoint, so a marked edge has a unique
eligible parent. Each parent contributes at most $m$ marked edges.
After fixing the first sample in a pair, the probability that the second
extends the same parent is therefore at most $m/X_\Gamma$.
Different parents already have separated roots. In both cases the
failure probability is at most
\[
 \ell_\Gamma^2\frac{\max\{m,m_{d,2r+1}\}}{X_\Gamma}
 \leq (1+\eps)^2\overline{c^X_\Gamma}\max\{m,m_{d,2r+1}\}
             \frac{t_i^2}{t_{i-1}}.
\]
The last fraction equals $n^{-1/(2^m-1)}$. This proves the claim with
the displayed constant as $c^H_\Gamma$. Repeated copies of an identical
instance are included among the failed pairs.
\end{claimproof}

\EventXGammai*
\begin{claimproof}
We induct on $i$. For $i=1$, every edge realizes the prescribed
one-edge type at its corresponding endpoint, so $X_\Gamma=|E(G)|\leq dn$.
For $i>1$, let $\Lambda=\Gamma_{[i-1]}$. Since the current step is
reached, the parent sampling step was performed. Each stored parent
contributes at most $m$ marked edges, and hence
\[
 X_\Gamma\leq m\ell_\Lambda
 \leq m\frac{t_{i-1}}{t_{i-2}}\tilde X_\Lambda.
\]
The induction hypothesis gives
$X_\Lambda\leq\overline{c^X_\Lambda}t_{i-2}$.
Together with $\cE^{\tilde X}_\Lambda$, this yields
\[
 X_\Gamma\leq m t_{i-1}
   \left(\overline{c^X_\Lambda}+\frac\eps2\underline{c^X_\Lambda}\right)
 \leq\overline{c^X_\Gamma}t_{i-1},
\]
by the parameter recursion in \Cref{alg:ApproxP}.
\end{claimproof}

\EventEXii*
\begin{claimproof}
Freeze an entry history satisfying $\cE^{X<}_\Gamma$. With
$K_\Gamma=\lceil c^M_\Gamma\sqrt{n/t_{i-1}}\rceil$,
\Cref{cor:count} gives failure probability at most $1/(100mD)$ and
\[
 |\tilde X_\Gamma-X_\Gamma|
 \leq\frac{2\pi\sqrt{\overline{c^X_\Gamma}t_{i-1}dn}}{K_\Gamma}
       +\frac{\pi^2dn}{K_\Gamma^2}
 \leq\frac\eps2\underline{c^X_\Gamma}t_{i-1}.
\]
Each summand is at most $\eps\underline{c^X_\Gamma}t_{i-1}/4$ by
the two defining lower bounds on $c^M_\Gamma$.
\end{claimproof}

\EventSii*
\begin{claimproof}
Freeze the full history and the current accurate count.
For $i>1$, $\cE^H_{\Gamma_{[i-1]}}$ ensures that no edge expands two
parent samples. For $i=1$, uniqueness follows instead because the fixed
one-edge type specifies which endpoint is the root. No separation of
$\cH_{\Gamma_0}$ is assumed.
    For each $W_{i-1}\in \cW_{\Gamma_{[i-1]},\Gamma'}$ and  $H\in\cH^{[W_{i-1}]}_{\Gamma_{[i-1]},\Gamma'}$, 
    $H$ contributes exactly $\kappa_{[W_{i-1}],[W]}$ edges to $E^{[W]}_{\Gamma,\Gamma'}$, thus
    \begin{equation}
    \label{eq:_E}
        \abs{E^{[W]}_{\Gamma,\Gamma'}} = \sum_{W_{i-1}\in \cW_{\Gamma_{[i-1]},\Gamma'}}\frac{1}{\abs{[W_{i-1}]}} \kappa_{[W_{i-1}],[W]}S^{[W_{i-1}]}_{\Gamma_{[i-1]},\Gamma'}.
    \end{equation}

    Since we apply \textsc{Grover} on $f_\Gamma$ for $\ell_\Gamma$ independent samples, each uniformly distributed over $f_\Gamma^{-1}(1)$, the probability that a sample falls into $E^{[W]}_{\Gamma,\Gamma'}$ is $\abs{E^{[W]}_{\Gamma,\Gamma'}}/ X_\Gamma$.
    That means \[
    \widehat{S}^{[W]}_{\Gamma,\Gamma'}\sim \cB(\ell_\Gamma, \abs{E^{[W]}_{\Gamma,\Gamma'}}/ X_\Gamma).\]

    Now, we can bound
    \begin{align*}
        p^{[W]}_{\Gamma,\Gamma'} &:= \Pr\left[\widehat{S}^{[W]}_{\Gamma,\Gamma'} \in (1\pm\eps) \frac{t_i}{t_{i-1}} \abs{E^{[W]}_{\Gamma,\Gamma'}}\pm\eta t_i \right] \geq \Pr\left[\widehat{S}^{[W]}_{\Gamma,\Gamma'} \in \ell_\Gamma \abs{E^{[W]}_{\Gamma,\Gamma'}} / X_\Gamma \pm \eta t_i\right] \\
        &\geq 1- 2\exp{\left(-2\eta^2 t_i^2 / \ell_\Gamma\right)} \geq 1-2\exp{\left(-2\eta^2 t_i / ((1+\eps)\overline{c^X_\Gamma})\right)}
        \geq 1-\frac{(1+\eps)\overline{c^X_\Gamma}}{\eta^2 t_i}.
    \end{align*}
    The first inequality uses the rounded ratio bound \Cref{eq:disc-rounding}, the second inequality holds because of  \Cref{lem:B}, the third inequality holds because of $\cE^X_\Gamma$ again, and the last inequality holds because of the fact $e^{-x}\leq 1/x$.

    Using the union bound for $\Gamma'\succeq \Gamma$ and $W\in \cW_{\Gamma,\Gamma'}$, we have,
    \[
        \Pr\left[\cE^{\widehat{S}}_\Gamma\right] \geq 1 - 2 \sum_{\Gamma'\succeq \Gamma}\sum_{W\in\cW_{\Gamma,\Gamma'}}\frac{(1+\eps)\overline{c^X_\Gamma}}{\eta^2 t_i} = 1 -  2 \sum_{\Gamma'\succeq \Gamma}\mu_{\Gamma,\Gamma'}\frac{(1+\eps)\overline{c^X_\Gamma}}{\eta^2 t_i}.
    \]

    Set $c^S_\Gamma = 2(1+\eps) \mu_\Gamma{\overline{c^X_\Gamma}}/{\eta^2}$ where we recall that $\mu_\Gamma = \sum_{\Gamma'\succeq \Gamma}\mu_{\Gamma,\Gamma'}$, and the claim holds.
\end{claimproof}

\EventS*
\begin{claimproof}
Parent separation (or endpoint uniqueness at $i=1$) makes every
marked edge's eligible parent unique. Current separation makes all
returned indexed instances distinct. Hence sampled edges in each
$T^{[W]}_{\Gamma,\Gamma'}$ correspond one-to-one with the stored
instances in $\cH^{[W]}_{\Gamma,\Gamma'}$, and
    \[
    \widehat{S}^{[W]}_{\Gamma,\Gamma'} = {S}^{[W]}_{\Gamma,\Gamma'}.
    \]
\end{claimproof}

\EventEeqi*
\begin{claimproof}
We give the error propagation explicitly. Fix a target type $\Gamma$
with $i$ edges and a full disc type $\Gamma'\succeq\Gamma$.
At level $j$ group the ordered tuples in $\cW_{\Gamma_{[j]},\Gamma'}$
into automorphism classes. Let $B_j$ be the nonnegative matrix whose
entry in row $C$ and column $P$ is $B_j(C,P)=\kappa_{P,C}$, for a child class $C$ and parent class $P$; all other entries are zero. Each child class has exactly one parent class, and
\[
 |C|=|P|\kappa_{P,C}.
\]
Indeed truncation is equivariant under automorphisms, and every tuple
in the child class has exactly one parent. Thus applying
$B_j\cdots B_1$ to the singleton root class gives the vector of class
sizes. At the last level their sum is $\mu_{\Gamma,\Gamma'}$.

For this fixed $\Gamma'$, let $s_j(C)=S^C_{\Gamma_{[j]},\Gamma'}/t_j$ and
$s_0=\cnt_{\Gamma'}/n$. The sampling and distinctness events say,
coordinatewise for $1\leq j<i$,
\[
 (1-\eps)B_js_{j-1}-\eta\boldsymbol1
 \leq s_j\leq(1+\eps)B_js_{j-1}+\eta\boldsymbol1.
\]
At the final counting level, the partition of marked edges gives
\begin{equation}\label{eq:X_Gamma}
 \frac{X_\Gamma}{t_{i-1}}
       =\sum_{\Gamma'\succeq\Gamma}\boldsymbol1^T B_i s_{i-1}.
\end{equation}
This includes $i=1$ by the endpoint uniqueness just established.
Since every class size is at least one, propagating an all-ones error
vector from any intermediate level to the last and summing its entries
has total weight at most $\mu_{\Gamma,\Gamma'}$. Hence expansion of
the recurrences gives upper error at most
$(i-1)(1+\eps)^{i-1}\mu_{\Gamma,\Gamma'}\eta$, and lower error at most
$(i-1)\mu_{\Gamma,\Gamma'}\eta$.
Our parameter choice ensures $(1-\eps)^{i-1}\geq1/2$, so both are
covered by the respective endpoints of
\[
 (1\pm\eps)^{i-1}
 \left(\mu_{\Gamma,\Gamma'}\frac{\cnt_{\Gamma'}}n
                      \pm2i\mu_{\Gamma,\Gamma'}\eta\right).
\]
Summing over $\Gamma'$ proves $\cE^\eq_\Gamma$. The factor $2$ is
needed to cover the lower-endpoint additive errors as well as the upper
ones; it is also present in the thresholds and parameter choices.
\end{claimproof}

\subsubsection{Proof of \texorpdfstring{\Cref{clm:E^X}}{the stage guarantee}}
\EventEX*
\begin{claimproof}
    Fix the full entry history in the claim. The events
    $\cE^\eq_\Gamma$ and $\cE^{X<}_\Gamma$ follow from
    \Cref{clm:_E^eq_i,clm:XGammai}.
    
    \textbf{Case 1:} When $\boldsymbol{\gamma}_\Gamma=1$, i.e.  $\sum_{\Gamma'\succeq \Gamma}\mu_{\Gamma,\Gamma'}\cnt_{\Gamma'}\geq \alpha_\Gamma n$, then event $\cE^{X>}_\Gamma$ always holds because
    \[
    X_\Gamma\geq (1-\eps)^{i-1}(\alpha_\Gamma - 2i\mu_\Gamma \eta) t_{i-1}= \underline{c^X_\Gamma} t_{i-1}.
    \]

    The equality follows from \Cref{def:type}.
    Then by \Cref{clm:_E^-X_i}, with probability at least $1-1/(100m_{d,r}D_{d,r})$,
    \[
    \tilde{X}_\Gamma \geq X_\Gamma - \tfrac12\eps \underline{c^X_\Gamma} t_{i-1} \geq (1-\eps) \underline{c^X_\Gamma} t_{i-1}.
    \]
    
    Thus, $\tilde{X}_\Gamma$ will pass the threshold check at Line \ref{line:0}.
    For $i<m$, the concentration and separation bounds follow from
    \Cref{clm:_E^S_i,clm:_E^H_i}; \Cref{clm:_E^S} then gives $\cE^S_\Gamma$
    whenever the separation event holds. At $i=m$, all three sampling
    events hold by definition. This proves the first case.

    \textbf{Case 2:} When $\boldsymbol{\gamma}_\Gamma=0$, i.e. $\sum_{\Gamma'\succeq \Gamma}\mu_{\Gamma,\Gamma'}\cnt_{\Gamma'}< \beta_\Gamma n$, then 
    \[
    X_\Gamma< (1+\eps)^{i-1}(\beta_\Gamma + 2i\mu_\Gamma \eta) t_{i-1} \leq (1-2\eps) \underline{c^X_\Gamma} t_{i-1}.
    \]
    
    The last inequality follows from \Cref{def:type}. Then by \Cref{clm:_E^-X_i}, with probability at least $1-1/(100m_{d,r}D_{d,r})$,
    \[
    \tilde{X}_\Gamma \leq X_\Gamma + \tfrac12\eps \underline{c^X_\Gamma} t_{i-1} < (1-\eps) \underline{c^X_\Gamma} t_{i-1}.
    \]
    
    Thus, $\tilde{X}_\Gamma$ will fail the threshold check at Line \ref{line:0}.
\end{claimproof}

\subsubsection{Proof of \texorpdfstring{\Cref{clm:Eintersection}}{simultaneous good events}}
\begin{proof}
Consider the stopped run defined before \Cref{lem:query-threshold}, in
the algorithm's fixed processing order. It stops after the first count
violating $\cE^{\tilde X}_\Gamma$ or the first sampling step violating
$\cE^{\widehat S}_\Gamma$ or $\cE^H_\Gamma$. Before such a failure,
every proper prefix of a visited type is abundant.
Fix all earlier outcomes satisfying these conditions. The deterministic
parts of \Cref{clm:E^X} give $\cE^\eq_\Gamma$ and the upper count bound;
the conditional counting failure probability is at most $1/(100mD)$.
If the count is accurate and the type is rare, the threshold fails and
the algorithm skips exactly its prefix descendants. There is no sampling
step in this case. If the type is abundant, its threshold passes; for
$i<m$, condition further on this accurate count and take a union bound
for concentration and separation. The conditional failure probability
at this iteration is at most
\[
 \frac1{100mD}+\frac{c^S_\Gamma}{t_i}
                +c^H_\Gamma n^{-1/(2^m-1)}.
\]
At $i=m$ only the first term occurs. A union bound over the disjoint
possibilities for the first failure, and then over at most $D$ visited
types, gives total failure probability at most
$1/(100m)+o(1)\leq0.02$ for sufficiently large $n$.
There is no independence assumption between concentration and separation
or between different levels. Since the conditional bounds are uniform
over full histories, they also apply when other prefix branches have
already been processed. This proves the claim.
\end{proof}

\section{Testing Properties in the Quantum  Unidirectional Model}
\label{sec:test}
In this section, we complete the proof of our main theorem, stated below. We note that \Cref{thm:main-classical-quantum} follows directly from this result.

\begin{theorem}[full version of \Cref{thm:main-classical-quantum}]
\label{thm:full-main}
   Assume $d\geq1$; the case $d=0$ requires no queries. If a digraph property is $\eps$-testable with two-sided error and classical query complexity $q=O_{\eps,d}(1)$ in the $d$-bounded-degree digraph bidirectional model, then it is also $\eps$-testable with two-sided error and quantum query complexity $O_{\eps,d}(n^{1/2-1/(2(2^{m_{d,r}}-1))})$ in the $d$-bounded-degree digraph unidirectional model for a universal constant $c>0$, where $r=cq$ and $m_{d,r} = 2d((2d)^{r+1}-1)/(2d-1)$.
\end{theorem}

To prove the above, we need the following result from \cite{10.1145/2897518.2897575} on the 
canonical testing algorithm for constant query testable properties in the classical bidirectional model.  

\begin{lemma}[Adapted from {\cite[Lemmas 3.1 and 3.2]{10.1145/2897518.2897575}}]
\label{lem:tester}
    Let $\Pi=(\Pi_n)_{n\in\N}$ be a property that is testable with query complexity $q = q(\eps,d)$ in the bidirectional model. Then there exists a universal constant $c > 0$ such that for every $\eps$ and $d$, there is an integer $n_0=n_0(\eps,d)$, and an infinite sequence $\cF = (\cF_n)_{n\geq n_0}$, such that for any $n\geq n_0$, $\cF_n$ is a set of digraphs, each being a union of $k$ disjoint $r$-discs with $k=r=cq$, and for any $n$-vertex digraph $G$,
    \begin{itemize}
        \item if $G$ satisfies $\Pi_n$, then with probability at most $5/12$ the formal disjoint union of the $r$-disc types of $k$ uniformly sampled (without replacement) vertices is isomorphic to one of the members in $\cF_n$;
        \item if $G$ is $\eps$-far from satisfying $\Pi_n$, then with probability at least $7/12$ the formal disjoint union of these $r$-disc types is isomorphic to one of the members in $\cF_n$.
    \end{itemize}
\end{lemma}

\begin{proof}
Use the amplified canonical tester from the proof of
\cite[Lemma 3.2]{10.1145/2897518.2897575}, whose rejection family
$\cF'_n$ has rejection probability at most $1/3$ on graphs satisfying
$\Pi_n$ and at least $2/3$ on graphs that are $\eps$-far.
Let $\cF_n$ be the subfamily consisting of $k$ disjoint $r$-discs.
Put $s_{d,L}=\sum_{j=0}^L(2d)^j$ and enlarge $n_0$ so that
\[
 \frac{\binom{k}{2}s_{d,2r+1}}{n-k+1}\leq\frac1{12}
 \qquad\text{for all }n\geq n_0.
\]
The left-hand side bounds the probability that two sampled roots have
weak distance at most $2r+1$. Outside this event, the physical
neighborhoods are disjoint with no connecting edges, so their union
agrees with the formal disjoint union. The rejection probabilities
therefore differ by at most $1/12$, giving $1/3+1/12=5/12$ and
$2/3-1/12=7/12$.
\end{proof}

Given our \Cref{thm:disc*} and the above lemma, we are now ready to complete the proof of \Cref{thm:full-main}. We note that the remainder of the argument closely follows the proof in \cite{10.1145/2897518.2897575}, with the only difference being that we replace their classical algorithm for estimating disc-type counts with our quantum algorithm from \Cref{thm:disc*}. For completeness, we include a sketch of the proof below.%

\paragraph{Proof (Sketch) of \Cref{thm:full-main}}
We may assume $d,q\geq1$, since an edgeless input or a zero-query
decision is handled directly.
Let $c$ and $n_0 = n_0(\eps, d)$ be the constants from Lemma~\ref{lem:tester}, and define the sample count $k=cq$, the disc radius $r=cq$, and $n_1 = (1 + k) \cdot 48 \cdot (2k (D_{d,r}+1))^k$.
Increase $n_0$ if necessary so that \Cref{thm:disc*} applies at accuracy
$\delta/(D_{d,r}+1)$, with $\delta$ as in \Cref{alg:test}.
Since the property $\Pi_n$ can be tested trivially with a constant number of queries for all $n < \max\{n_0, n_1\}$, we assume throughout that $n \geq \max\{n_0, n_1\}$.

Let $\cF_n$ denote the set of subgraphs (each formed by the union of $k$ disjoint $r$-discs) whose occurrence causes the canonical tester to reject, as ensured by Lemma~\ref{lem:tester}.
Without loss of generality, we may represent each $F \in \cF_n$ as a multiset of $r$-discs:
\[
F = \{\Delta_1,\dots,\Delta_k\}
\]
where each $\Delta_i$ is an $r$-disc.

For a real number $a$ and an integer $s\geq0$, we use the generalized
binomial coefficient
\[
 \binom{a}{s}=\frac{\prod_{j=0}^{s-1}(a-j)}{s!},
 \qquad\binom{a}{0}=1.
\]
This agrees with the usual binomial coefficient for nonnegative integer
counts, including the value $0$ when $a<s$, and also defines the formula
for real or negative estimates.

Then, \Cref{alg:test} is shown to correctly test the property $\Pi_n$ under the model that permits only outgoing neighbor queries.

\begin{algorithm}[ht]
    \caption{Testing property $P$}
    \label{alg:test}
    \begin{algorithmic}[1] %
        \Procedure{TestP}{$G, n, d, q, \eps$} \Comment{$n\geq\max\{n_0,n_1\}$}
        \State Let $k=r=cq$ and $\delta = \frac{1}{48\cdot(2k(D_{d,r}+1))^k}$.
        \State $\{X_0,\dots,X_{D_{d,r}}\} = \textsc{AmplifiedCompleteEstDisc*}(G,n,d,r,\delta)$.
        \For {$F = \{\Delta_1,\dots,\Delta_k\}\in\cF_n$}
            \For {$i=0\dots D_{d,r}$}
                \State Let $x_i$ be the number of copies among $\{\Delta_j\}_{j=1}^k$ that are of the same isomorphic type as $\Gamma_i$
            \EndFor
            \State Let $\est(F)\ ={\prod_{i=0}^{D_{d,r}}\binom{X_i}{x_i}}/{\binom{n}{k}}$.
        \EndFor
        \If {$\sum_{F\in\cF_n}\est(F) < 1/2$}
            \State \Return \textbf{Accept}
        \Else \State \Return \textbf{Reject}
        \EndIf
        \EndProcedure
    \end{algorithmic}
\end{algorithm}

    Algorithm \textsc{EstDisc*} 
    ensures that each returned value $X_i$ will be very close to $\cnt_{\Gamma_i}$ (c.f. \Cref{thm:disc*}) with probability at least $5/9$.
    We first amplify the success probability by repeating it a constant number of times and taking coordinate-wise medians, so that all disc counts are simultaneously correct with probability at least $0.99$.
    Then we estimate the number of empty disc type $\Gamma_0$ by
    \[
    \widetilde{\cnt}_{\Gamma_0} = n - \sum_{\Gamma\in\cD_{d,r}^*}\widetilde{\cnt}_\Gamma
    \]
    and rescale the error parameter $\delta$ by dividing $D_{d,r}+1$ to make the additive error of estimating ${\cnt}_{\Gamma_0}$ also be bounded by $\delta n$.
    We call this new algorithm \textsc{AmplifiedCompleteEstDisc*}. It changes the query complexity only by a constant factor and achieves the success probability requirement $2/3$ for property testing.
    In the remainder of the proof, we condition on the simultaneous accuracy event $\abs{X_i-\cnt_{\Gamma_i}}\leq\delta n$ for every $i\in[D_{d,r}]_0$, which holds with probability at least $0.99$.
    Therefore for every $F = \{\Delta_1,\dots,\Delta_k\}\in\cF_n$ and the relevant $x_0,\dots,x_{D_{d,r}}$, we will study $\prob(\Delta_1, . . . , \Delta_k) :=  {\prod_{i=0}^{D_{d,r}}\binom{\cnt_{\Gamma_i}}{x_i}}/{\binom{n}{k}} $, from which we will obtain the required bounds  for $\sum_{F\in\cF_n}\est(F)$.

    Observe that the multiset of $r$-disc types of $k$ vertices sampled uniformly without replacement has the \emph{multivariate hypergeometric distribution} with parameters $n, \cnt_{\Gamma_0}, \dots , \cnt_{\Gamma_{D_{d,r}}}, k$. Thus, if there are exactly $x_i$ copies of type $\Gamma_i$ in $\{\Delta_1,\dots,\Delta_k\}$, with $x_0+\dots+x_{D_{d,r}}=k$, then the probability of this type multiset, or equivalently its formal disjoint union, is
    \[
    \prob(\Delta_1,\ldots,\Delta_k)
    =\frac{\prod_{i=0}^{D_{d,r}}\binom{\cnt_{\Gamma_i}}{x_i}}{\binom{n}{k}}.
    \]

    Writing $\prob(F):=\prob(\Delta_1,\ldots,\Delta_k)$ for $F=\{\Delta_1,\ldots,\Delta_k\}$, we prove that $\est(F)$ and $\prob(F)$ are very close via the following claims.

    \begin{claim}
    \label{clm:ps_1}
        For any $i\in[D_{d,r}]_0$, if $\abs{X_i-\cnt_{\Gamma_i}} \leq \delta n$, it holds that $\abs{\binom{X_i}{x_i}-\binom{\cnt_{\Gamma_i}}{x_i}}\leq 4\delta n ^{x_i}$.
    \end{claim}
    \begin{claim}
    \label{clm:ps_2}
        Let $\cI=\{i\in[D_{d,r}]_0:x_i>0\}$. For any $i\in[D_{d,r}]_0$, if $\abs{X_i-\cnt_{\Gamma_i}} \leq \delta n$, then the following inequalities hold:
        \begin{align*}
            &\prod_{i\in\cI}\left(\binom{\cnt_{\Gamma_i}}{x_i}+4\delta n^{x_i}\right) < \prod_{i\in\cI} \binom{\cnt_{\Gamma_i}}{x_i} +4\delta 2^k n^k, \\
            &\prod_{i\in\cI}\left(\binom{\cnt_{\Gamma_i}}{x_i}-4\delta n^{x_i}\right) > \prod_{i\in\cI} \binom{\cnt_{\Gamma_i}}{x_i} -4\delta 2^k n^k.
        \end{align*}
    \end{claim}

    Expanding the product error using \Cref{clm:ps_1} gives the following claim.
    \begin{claim}
    \label{clm:ps_3}
        If $\abs{X_i-\cnt_{\Gamma_i}} \leq \delta n$ for every $i\in[D_{d,r}]_0$, then $\abs{\est(F)-\prob(F)}< 4\delta (2k)^k$ for every $F\in\cF_n$.
    \end{claim}

    These claims follow the argument of \cite[Claims 5.1--5.3]{10.1145/2897518.2897575}. To accommodate real-valued estimates, note that $n\geq n_1=(k+1)/\delta$ and $(1+2\delta)^k<2$. For $1\leq s\leq k$, the derivative of $\binom{a}{s}$ on $[-\delta n,(1+\delta)n]$ has absolute value at most $2n^{s-1}$, so the mean value theorem proves \Cref{clm:ps_1}; the case $s=0$ is immediate. Write $a_i=\binom{\cnt_{\Gamma_i}}{x_i}$ and $e_i=\binom{X_i}{x_i}-a_i$. Since $\abs{a_i}\leq n^{x_i}$, $\abs{e_i}\leq4\delta n^{x_i}$ and $4\delta\leq1$, expansion gives
    \[
    \left|\prod_{i\in\cI}(a_i+e_i)-\prod_{i\in\cI}a_i\right|
    \leq n^k\sum_{\emptyset\ne S\subseteq\cI}(4\delta)^{|S|}
    <4\delta 2^k n^k.
    \]
    The same expansion with $e_i=\pm4\delta n^{x_i}$ proves \Cref{clm:ps_2}. Dividing by $\binom nk\geq(n/k)^k$ proves \Cref{clm:ps_3}, without multiplying possibly negative lower bounds.

    Now it is sufficient to claim:
    \begin{itemize}
        \item if $G$ satisfies $\Pi_n$, then $\sum_{F\in\cF_n}\est(F)<1/2$, and
        \item if $G$ is $\eps$-far from $\Pi_n$, then $\sum_{F\in\cF_n}\est(F)\geq 1/2$.
    \end{itemize}

    Let $G$ be a $d$-bounded-degree digraph satisfying $\Pi$. 
    Then, by \Cref{lem:tester}, the formal disjoint union of the $r$-disc types of $k$ vertices sampled uniformly at random without replacement belongs to $\cF_n$ with probability at most $5/12$, that is, $\sum_{F\in\cF_n}\prob(F)\leq 5/12$.
    Therefore, by \Cref{clm:ps_3},
    \[
    \begin{aligned}
    \sum_{F\in\cF_n}\est(F)
    &\leq \sum_{F\in\cF_n}\prob(F) + \sum_{F\in\cF_n}\abs{\est(F)-\prob(F)}\\
    &< \sum_{F\in\cF_n}\prob(F) + (D_{d,r}+1)^k\cdot 4\delta (2k)^k
    \leq \frac{5}{12}+ \frac{1}{12}=\frac{1}{2}.
    \end{aligned}
    \]

    Similarly, by \Cref{lem:tester}, if $G$ is $\eps$-far from $\Pi$, then the formal disjoint union of the $r$-disc types of $k$ vertices sampled uniformly at random without replacement belongs to $\cF_n$ with probability at least $7/12$.
    Hence, \Cref{clm:ps_3} gives,
    \[
    \begin{aligned}
    \sum_{F\in\cF_n}\est(F)
    &\geq \sum_{F\in\cF_n}\prob(F) - \sum_{F\in\cF_n}\abs{\est(F)-\prob(F)}\\
    &> \sum_{F\in\cF_n}\prob(F) - (D_{d,r}+1)^k\cdot 4\delta (2k)^k
    \geq \frac{7}{12} - \frac{1}{12}=\frac{1}{2}.
    \end{aligned}
    \]
    
    That completes the proof.

\section{Lower Bound}
\label{sec:LB}
In this section, we prove our lower bound for testing $k$-star-freeness (\Cref{def:kstarfree}) in digraphs with bounded maximum in- and out-degree, i.e., \Cref{thm:k_star_lowerbound}, which we restate below.

\Kstar*

We firstly consider a related and slightly harder problem called $k$-occurrence-freeness, which can be reduced to $k$-star-freeness.
We then establish its quantum lower bound using the well-known (dual) polynomial method.

Let us start by introducing some new notation %
and definitions.

\subsection{Notation and definitions}

For a bit-string $x\in\bits^n$, we will use $|x|$ to denote the Hamming weight of vector $x\in\bits^n$, i.e.
    \[|x|=|\{i\in[n]:x_i=-1\}|.\]
Let $H^n_{\leq k}$, $H^n_{\geq k}$ and $H^n_{I}$ denote the sets of length-$n$ bit-strings whose Hamming weights are at most $k$, at least $k$ and within interval $I$, respectively. That is,
    \[ 
    \begin{aligned}
        &H^n_{\leq k}=\{x\in\bits^n:|x|\leq k\}, \quad H^n_{\ge k}=\{x\in\bits^n:|x|\ge k\},\\
        & \text{ and } \quad H^n_I=\{x\in\bits^n:|x|\in I\}.
    \end{aligned}\]

We denote by $\vecone^n$ the length-$n$ bit-string consisting entirely of $1$s, and by $-\vecone^n$ its negation. We define some fundamental Boolean functions that will be used later.
\begin{itemize}
    \item The \emph{threshold function} $\THR_N^k:\bits^N\to\bits$ outputs $-1$ if the Hamming weight of the input bit-string is at least $k$, and $1$ otherwise;
    \item The \emph{promised threshold function}  $\BTHR_N^k:H^N_{\leq k}\to\bits$ agrees with $\THR_N^k$ on inputs in $H^N_{\leq k}$, and is undefined otherwise;
    \item The \emph{OR function} $\OR_R:\bits^R\to\bits$ outputs $1$ if the input equals $\vecone^R$, %
    and $-1$ otherwise.
    \item The \emph{gap OR function} $\GapOR_R^{\eps N}: H^R_{\geq \eps N}\cup \{\vecone^R\}\to\bits$ agrees with $\OR_R$ on inputs in $H^R_{\geq\eps N}\cup\{\vecone^R\}$, and is undefined otherwise.
\end{itemize}

\subsection{Relating \texorpdfstring{$k$}{k}-star-freeness to \texorpdfstring{$k$}{k}-occurrence-freeness}
We will relate the problem of testing $k$-star-freeness in $k$-bounded-degree digraphs to that of testing $k$-occurrence-freeness in functions defined as \Cref{def:koccurrencefree}. \cite{peng2023optimalseparationpropertytesting} (see also \cite{hellweg2012property}) provides a local reduction from  testing $k$-occurrence-freeness of a sequence to the problem of testing $k$-star-freeness in the unidirectional model. %
\begin{lemma}[\cite{peng2023optimalseparationpropertytesting,hellweg2012property}]
\label{obs:reduction}
Given an instance $f:[N]\to[R]$ of the $k$-occurrence-freeness testing problem where $R=cN$ and $c\geq1$, we define the directed graph $G_f=(V,E)$ with $V=\{v_1,\dots,v_{N+R}\}$ and $E=\{(v_i,v_{N+j}) : f(i)=j\}$. Then the following properties hold:  
    \begin{itemize}
        \item If $f$ is $k$-occurrence-free, then $G_f$ is $k$-star-free;
        \item If $f$ is $\varepsilon$-far from $k$-occurrence-free, then $G_f$ is $\frac{\eps}{(c+1)k}$-far from $k$-star-free;
        \item If $f$ is $k$-bounded, then $G_f$ is $k$-bounded-degree;
        \item Each quantum outgoing-neighbor query in $G_f$ can be simulated by at most two queries to $f$.
    \end{itemize}
\end{lemma}
\begin{proof}
    The first item is immediate. For the second item: 
    since $f$ is $\eps$-far from $k$-occurrence-freeness, $G_f$ must contain at least $\eps N$ edge-disjoint $k$-stars with distinct centers.
    To make $G_f$ $k$-star free, at least one edge must be removed from each $k$-star. Therefore, the number of edges that must be removed is at least $\eps N =\frac{\eps N}{k(R+N)}\cdot k\cdot |V| \geq \frac{\eps}{(c+1)k}\cdot k\cdot |V|$, which establishes the claimed distance. 

    For the third item: since $f$ is $k$-bounded, the maximum in-degree of $G_f$ is at most $k$, and the maximum out-degree is at most $1$.
    Hence $G_f$ is $k$-bounded-degree.

    For the last item: let $\cO_f$ be the standard query oracle to function $f$, i.e.,
    \[
    \cO_f(s)=f(s)\quad \forall s\in[N].
    \]
 Then a unidirectional outgoing-neighbor oracle $\cO^\out_{G_f}$ for $G_f$ can be constructed as follows: 
    \[
    \cO_{G_f}^\out(v_s,i)=
    \begin{cases}
        v_{N+\cO_f(s)}, & s\leq N \text{ and } i=1; \\
        \perp, & \text{otherwise}.
    \end{cases}
    \]
    To implement this quantum oracle, compute $f(s)$ in an auxiliary register, add the indicated vertex name to the answer register modulo $|V(G_f)|+1$, and uncompute the auxiliary register with an inverse query. For query locations outside $s\leq N$ and $i=1$, use a fixed valid input to $f$ and leave the answer register unchanged. This uses at most two queries to $f$.
\end{proof}

The above lemma implies that an $\eps$-tester for $k$-star-freeness can be used to construct a $(c+1)k\eps$-tester for $k$-occurrence-freeness. Consequently, a lower bound for the $k$-occurrence-freeness testing problem also gives a lower bound of $k$-star-freeness testing problem in $k$-bounded-degree digraphs. 

We now state the main technical result of this section: a lower bound for the $k$-occurrence-freeness testing problem.

\begin{theorem}
\label{thm:lowerbound}
Let $ k\geq 2$ and $0<\eps<1/(4^{k-1}\lceil 20 (2k)^{k/2}\rceil)$ be constants and let $N = \lceil 20 (2k)^{k/2}\rceil R$ 
be a large enough positive integer.
Then the quantum query complexity of $\eps$-testing $k$-occurrence-freeness 
for a $k$-bounded function $f:[N]\to[R+N]$ is $\Omega_{k,\eps}(N^{1/2 - 1/(2k)}/\ln^4 N)$.
\end{theorem}

Once \Cref{thm:lowerbound} is established, we can prove the main lower bound stated in \Cref{thm:k_star_lowerbound} using the reduction \Cref{obs:reduction} as follows.

\begin{proof}[Proof of \Cref{thm:k_star_lowerbound}]
Write $c_k=\lceil20(2k)^{k/2}\rceil$ and fix the graph proximity
$0<\eps_g<(32k)^{-k/2-1}$. We have
\[
 \eps_g<\frac1{k(2c_k+1)4^{k-1}}.
\]
Indeed $c_k\leq21(2k)^{k/2}$ gives
$k(2c_k+1)4^{k-1}\leq(43/128)(32k)^{k/2+1}$.
Choose and fix
\[
 \frac{\eps_g k(2c_k+1)}{c_k}<\eps_f<\frac1{4^{k-1}c_k}.
\]
For any sufficiently large vertex count $n$, let
$R=\lfloor n/(2c_k+1)\rfloor$, $N=c_kR$, and $n_*=2N+R$.
Given the $k$-bounded function $f:[N]\to[R+N]$ of
\Cref{thm:lowerbound}, form $G_f$ as in \Cref{obs:reduction} and add
$n-n_*<2c_k+1$ isolated vertices. A free function produces a
$k$-star-free graph. An $\eps_f$-far function produces at least
$\eps_f N$ edge-disjoint $k$-stars, so making the padded graph free
requires at least $\eps_f N$ edge deletions. The strict choice of
$\eps_f$ ensures $\eps_f N>\eps_g kn$ for all sufficiently large $n$.
Thus the padded graph is $\eps_g$-far, with both degrees at most $k$.
Padding is computable without oracle calls, and each outgoing query
still requires at most two queries to $f$. Since $N=\Theta_k(n)$,
\Cref{thm:lowerbound} yields the claimed bound for every sufficiently
large $n$, not only the sizes $n=(2c_k+1)R$.
\end{proof}

\subsection{Proof of \texorpdfstring{\Cref{thm:lowerbound}}{lower bound theorem}}
In the following, we prove \Cref{thm:lowerbound}, using the standard polynomial method.
We begin in \Cref{sec:Prop2Func}, where we express the property of $k$-occurrence-freeness as a Boolean function, denoted by $\OCCU^k$.
In \Cref{sec:QLBfromABD}, we show that the approximate degree of $\OCCU^k$ implies a quantum lower bound for testing $k$-occurrence-freeness.
Next, in \Cref{sec:dummy-augment}, we introduce the dummy-augmented function $\dOCCU^k$, which we prove to be nearly as hard as $\OCCU^k$.
In \Cref{sec:BC}, we show that $\dOCCU^k$ can be expressed as the block composition $\GapOR \circ \BTHR^k$ over binary-encoded inputs.
In \Cref{sec:dual-poly}, we review the dual polynomial method, which relates the pure high degree of a dual witness to the approximate degree.
In \Cref{sec:dual-BC}, we construct dual witnesses for $\GapOR$ and $\THR^k$, and show that their composition yields a dual witness for $\GapOR \circ \BTHR^k$, even without explicitly constructing one for $\BTHR^k$.
Then, in \Cref{sec:zero-out}, we modify the composed witness to have support on inputs of total Hamming weight at most $N$, controlling the loss in pure high degree by the zeroing-out technique.
Finally, in \Cref{sec:pfLB}, we complete the proof of \Cref{thm:lowerbound} by combining the results established throughout this section.

\subsubsection{Expressing the property as a function}
\label{sec:Prop2Func}
Although the $k$-occurrence-freeness testing problem is defined on the function $f:[N]\to [R]$, it is often more convenient to view the input as an integer string $s\in[R]^N$ obtained from the truth table of $f$, where $s_i = f(i)$ for each $i \in [N]$.
We will use a function $\OCCU_{N,R}^{k,\eps}$ to indicate whether the input is $k$-occurrence-free or $\eps$-far from $k$-occurrence-freeness. 
Before that, we need to first clarify its promised domain, i.e. the set of all valid inputs. The inputs are restricted by two conditions: 
\begin{enumerate}
    \item \emph{Bounded-promise:} each value in $[R]$ occurs at most $k$ times in the string, i.e. for any $r\in[R]$,
    \begin{equation}
    \label{eq:rule_bounded} 
        \#_r(s) := |\{i\in[N]:s_i=r\}|\leq k.
    \end{equation}
    \item \emph{Gap-promise:} either no value in $[R]$ occurs exactly $k$ times, or at least $\eps N$ distinct values occur exactly $k$ times, i.e.
    \begin{equation}
    \label{eq:rule_gap} 
        |\{r\in[R]:\#_r(s)=k\}| \in \{0\} \cup[\eps N,\infty).
    \end{equation}
\end{enumerate}
Hence, the promised domain is given by 
\[
D_{\OCCU_{N,R}^{k,\eps}}=\{s\in[R]^N: s \text{ satisfies \Cref{eq:rule_bounded} and \Cref{eq:rule_gap}}\}.
\]

We now define the behavior of the function.
\begin{definition}
\label{def:occu}
The \emph{occurrence function} $\OCCU_{N,R}^{k,\eps }:D_{\OCCU_{N,R}^{k,\eps}}\to\{-1,1\}$ is defined as follows: it takes value $1$ if no element occurs exactly $k$ times in $s$ (i.e., all elements occur fewer than $k$ times), and value $-1$ if at least $\eps N$ distinct elements occur exactly $k$ times in $s$. It is undefined otherwise.
\end{definition}

\subsubsection{Query lower bound from the approximate bounded degree}
\label{sec:QLBfromABD}
For $R\geq2$, we also view $\OCCU_{N,R}^{k,\eps}$ as a partial Boolean function on $\bits^n$ by encoding each element of $[R]$ using $b=\lceil\log_2 R\rceil$ bits, where $n=Nb$. Unused binary codewords lie outside the promised domain. More generally, an alphabet of size $L$ is encoded using $\lceil\log_2 L\rceil$ bits.

One can then use the minimum-degree polynomial that approximates $\OCCU_{N,R}^{k,\eps}$ to derive a query lower bound for any quantum algorithm that computes this function, which is the central idea of the well-known polynomial method. We need the following definition. 

\begin{definition}[Approximate bounded degree]
\label{def:adeg}
Let  $F:D_F\subseteq \bits^n \to\bits$ and $\delta>0$. 
The \emph{$\delta$-approximate bounded degree $\adeg_\delta(F)$} is the least degree of a real polynomial $p:\{-1,1\}^n\to\mathbb{R}$ that $\delta$-approximates $F$, i.e.,
 $$\forall x\in D_F:\ |F(x)-p(x)|\leq\delta \quad \text{ and }\quad\forall x\in\{-1,1\}^n\setminus D_F:\ |p(x)| \leq 1+\delta.$$
\end{definition}

We say that a quantum query algorithm $Q$ computes a Boolean function $F$ with error $\delta$ if 
\[\Pr[Q(x)\neq F(x)]\leq \delta\] for all valid inputs $x$.
The \emph{quantum query complexity} $Q_\delta(F)$ is the minimum number $T$ of value queries required by any quantum algorithm that computes $F$ with error at most $\delta$. 
Map unused binary codewords to an arbitrary fixed element of $[R]$, so that the value oracle is unitary on every Boolean input. Each matrix entry of a value query depends on a single $b$-bit input block and therefore has a multilinear polynomial representation of degree at most $b$. Following the polynomial-method argument of \cite{beals_poly}, after $T$ value queries every amplitude has degree at most $Tb$. Thus the acceptance probability $\acc(x):=\Pr[Q(x)=-1]$ has degree at most $2Tb$.
Consequently, if a quantum query algorithm $Q$ computes a Boolean function $F$ with error $\delta$, then the polynomial 
\[p(x)=1-2\acc(x)\] (1) approximates $F$ within additive error $2\delta$ for all $x\in D_F$, (2) remains bounded in the interval $[-1,1]$ for all $x\notin D_F$ since $\acc(x)$ is a probability, and (3) has degree at most $2Tb$.
The following lemma follows immediately by taking  $F=\OCCU_{N,R}^{k,\eps}$. 

\begin{lemma}[\cite{beals_poly}]
\label{lem:poly_method}
It holds that 
    $Q_\delta(\OCCU_{N,R}^{k,\eps})\geq \adeg_{2\delta}(\OCCU_{N,R}^{k,\eps})/(2\lceil\log_2 R\rceil)$.
\end{lemma}

\subsubsection{Dummy augmentation}
\label{sec:dummy-augment}
For technical reasons, %
it is more convenient to study the dummy-augmented version of $\OCCU_{N,R}^{k,\varepsilon} $, where the input sequence may contain the integer $0$, but those zeros are simply ignored whenever they occur. In other words, while any of the items $ s_1, s_2, \dots, s_N $ may take the dummy value $0$, the number of zeros is unrestricted, and the input is not considered non-$k$-occurrence-free even if more than $k$ zeros occur. This gives rise to a dummy-augmented version of the occurrence function with a correspondingly dummy-augmented domain. 
\[
D_{\dOCCU_{N,R}^{k,\eps}}=\{s\in[R]_0^N: s \text{ satisfies \Cref{eq:rule_bounded} and \Cref{eq:rule_gap}}\}.
\]
\begin{definition}
\label{def:dOCCU}
    Define the function $\dOCCU_{N,R}^{k,\eps}:D_{\dOCCU_{N,R}^{k,\eps}}\to\{-1,1\}$ as follows: it outputs $1$ if no integer (other than $0$) occurs exactly $k$ times in $s$ (i.e. every integer in $s$ except $0$ occurs fewer than $k$ times); it outputs $-1$ if there are at least $\eps N$ distinct nonzero integers, each occurring exactly $k$ times in $s$; it is undefined otherwise.
\end{definition}

Related dummy-symbol reductions appear for element distinctness in {\cite[Proposition 2.13]{bun2018polynomial}} and for collision freeness in {\cite[Proposition 24]{apers2024quantumpropertytestingsparse}}. We use the following bounded-occurrence variant, whose proof includes the binary encoding, to reduce to bounding the approximate bounded degree of $\dOCCU_{N,R}^{k,\eps}$.
\begin{lemma}
\label{lem:dummy_augment}
    Let $k\geq 2$, $\delta > 0$ and $N,R\in\mathbb{N}$. Then
\[ \adeg_\delta(\dOCCU_{N, R}^{k,\eps}) \le \adeg_\delta(\OCCU_{N , R + N}^{k,\eps}) \cdot \lceil\log_2(R+1)\rceil.
\]
\end{lemma}

\begin{proof}
Let
\[
 p : \bits^{N \cdot \lceil\log_2(R+N)\rceil} \to \R
\]
be a polynomial of degree $d$ that $\delta$-approximates $\OCCU_{N, R+N}^{k,\eps}$. We will use $p$ to construct a polynomial of degree $d\cdot \lceil\log_2(R+1)\rceil$ that $\delta$-approximates %
$\dOCCU_{N, R}^{k,\eps}$. An input to $\dOCCU_{N, R}^{k,\eps}$ is an integer string $s\in[R]_0^N$, with each entry represented using $\lceil\log_2(R+1)\rceil$ bits. Define a  family of transformations $T_i : [R]_0 \hookrightarrow [R + N]$ by
\[T_i(s) = \begin{cases}
R + i & \text{ if } s = 0;\\
s & \text{ otherwise}.
\end{cases}\]
Extend each $T_i$ arbitrarily to unused binary codewords. As a mapping between binary representations, $T=(T_1,\dots,T_N)$ can then be realized by a vector of polynomials of degree at most $\lceil\log_2(R+1)\rceil$. For every $s\in D_{\dOCCU_{N,R}^{k,\eps}}$, observe that
\[\dOCCU_{N, R}^{k,\eps}(s_1, \dots, s_N) = \OCCU_{N, R+N}^{k,\eps}(T_1(s_1), \dots, T_N(s_N)).\]
Hence, the polynomial
\[p(T_1(s_1), \dots, T_N(s_N))\]
is a polynomial of degree $d \cdot \lceil\log_2(R+1)\rceil$ that $\delta$-approximates $\dOCCU_{N, R}^{k,\eps}$.
\end{proof}

\subsubsection{Block composition}
\label{sec:BC}

In this section we firstly introduce the block composition of two general Boolean functions. 

\begin{definition}[Block composition]
For Boolean functions $f:D_f\subseteq\bits^R\to\bits$ and $g:D_g\subseteq\bits^N\to\bits$, we define the \emph{block composition} $f \circ g : D_{f\circ g} \to \bits$ by \[
(f\circ g)(x_1, \dots, x_R) = f(g(x_1), \dots, g(x_R))\]
where \[
D_{f\circ g}=\{x\in(D_g)^R:(g(x_1), \dots, g(x_R))\in D_f\}.
\]
\end{definition}

For $s\in[R]_0^N$, define its incidence encoding
$x(s)=(x_{ij}(s))_{i\in[N],j\in[R]}\in\bits^{NR}$ by setting
$x_{ij}(s)=-1$ if $s_i=j$ and $x_{ij}(s)=1$ otherwise.
Write $x_j(s)=(x_{1j}(s),\ldots,x_{Nj}(s))$ for the block corresponding
to the value $j$. A dummy entry $s_i=0$ gives an all-$1$ row.
Given $f$ and $g$, set
$D_F=\{s\in[R]_0^N:x(s)\in D_{f\circ g}\}$ and define
\begin{equation}
\label{eq:bc}
    F(s)=(f\circ g)(x_1(s),\ldots,x_R(s)),\qquad s\in D_F.
\end{equation}
For $\adeg_\delta(F)$, we encode each symbol of $[R]_0$ using
$b=\lceil\log_2(R+1)\rceil$ bits; inputs containing unused codewords
lie outside the promised domain.
Every incidence encoding has total Hamming weight at most $N$.
We extend from these encodings to all inputs in
$H^{NR}_{\leq N}\cap D_{f\circ g}$ by defining
\begin{equation}
\label{eq:<=N}
    F^{\leq N}(x_1,\dots,x_R) = (f\circ g)(x_1,\dots,x_R), \text{ for all } x\in H^{NR}_{\leq N} \cap D_{f\circ g}.
\end{equation}

\begin{definition}[Double-promise approximate degree]
\label{def:dpdeg}
    Let $f:D_f\subseteq\bits^R\to\bits$ be a partial Boolean function, $g:D_g\subseteq\bits^N\to\bits$ be a partial symmetric Boolean function, meaning that both its domain and its values are invariant under permutations of its $N$ coordinates. For $F^{\leq N}$ defined above and any $\delta>0$,
    the \emph{double-promise $\delta$-approximate degree}, denoted by $\dpdeg_\delta(F^{\leq N})$, is the minimum degree of a real polynomial $p:\{-1,1\}^{NR}\to\mathbb{R}$ satisfying
    \begin{align*}
        &|p(x)-F^{\leq N}(x)|\leq \delta \text{ for all } x\in H^{NR}_{\leq N} \cap D_{f\circ g},\\
        &|p(x)|\leq 1+\delta  \text{ for all } x\in H^{NR}_{\leq N} \setminus D_{f\circ g}.
    \end{align*}
\end{definition}

We use the following extension of \cite[Theorem 6.5]{bun2018polynomial}
to partial symmetric inner functions $g$. The proof below also handles
the binary symbol encoding: approximating $F$ is at least as hard as
approximating $F^{\le N}$.

\begin{lemma}[Adapted from {\cite[Theorem 6.5]{bun2018polynomial}}]
\label{lem:dp-encode}
    Let $f$ and $g$ be partial Boolean functions, with $g$ symmetric in both its domain and its values, and let $F$ and $F^{\leq N}$ be defined above. Then for any $\delta>0$,
    \[
    \adeg_\delta(F)\geq \dpdeg_\delta(F^{\leq N}).
    \]
\end{lemma}
\begin{proof}
Let $p$ be a degree-$a$ polynomial satisfying the conditions for
$\adeg_\delta(F)$ under the binary encoding of $[R]_0$.
On valid symbol lists, average $p$ over all permutations of the $N$
positions. Since both the domain and the values of $g$ are symmetric,
this average preserves the approximation on $D_F$ and the bound
$1+\delta$ outside $D_F$.

Restricted to valid symbol encodings, a binary monomial touching $h$
distinct positions expands into a linear combination of products
$\prod_{t=1}^h\mathbf{1}[s_{i_t}=j_t]$, where $h\leq\min\{a,N\}$.
Let $z_j$ count the occurrences of symbol $j$, and let
$m_j=|\{t:j_t=j\}|$. The permutation average of such a product is
\[
 \frac{\prod_{j=0}^R(z_j)_{m_j}}{(N)_h},
 \qquad \sum_{j=0}^R m_j=h,
\]
where $(z)_m=z(z-1)\cdots(z-m+1)$ and $(z)_0=1$.
The denominator is nonzero since $h\leq N$, and the numerator has
degree at most $h$. Thus the averaged polynomial agrees, on every
valid list, with a polynomial $q(z_0,\ldots,z_R)$ of degree at most $a$.

Define
\[
 P(x)=q\left(N-\sum_{j=1}^R|x_j|,|x_1|,\ldots,|x_R|\right).
\]
Each Hamming weight is affine in the Boolean variables, so
$\deg P\leq a$. For every $x\in H^{NR}_{\leq N}$, these substituted
counts are nonnegative integers summing to $N$, and hence are realized
by a list $s\in[R]_0^N$. The blocks $x_j$ and $x_j(s)$ have equal
Hamming weights. Symmetry of the domain and values of $g$ therefore
implies $x\in D_{f\circ g}$ if and only if $s\in D_F$, with equal
function values whenever defined. Consequently $P$ satisfies both
conditions for $\dpdeg_\delta(F^{\leq N})$, proving the inequality.
\end{proof}

\subsubsection{Dual polynomial method}
\label{sec:dual-poly}
To bound the approximate degree of a function, it suffices to bound the maximum pure high degree of the dual witness of this function (see {\cite[Proposition 2.3 and Proposition 6.6]{bun2018polynomial}}). We need the following definitions and tools. 

For a function $\psi:\bits^n\to \R$, we say it is normalized if $\|\psi\|_1=1$.

\begin{definition}[Dual witness for the (double-promised) Boolean function]
\label{def:dual_witness}
    Let $\delta>0$, the Boolean function $F : D_F\subseteq \bits^n \to \{-1,1\}$ and the double-promised Boolean function $F^{\le N}$ be defined above.
    Then
    \begin{itemize}
        \item a normalized function $\psi \colon \bits^n \to \R$ is called a \emph{dual witness} for $F$ if
        \begin{align} &\sum_{x \in D_F}\psi(x) \cdot F(x) - \sum_{x \in \bits^n \setminus D_F} |\psi(x)| > \delta,
        \end{align}
        \item and a normalized function $\psi \colon \bits^{NR} \to \R$ is called a \emph{dual witness} for $F^{\le N}$ if
        \begin{align} 
        \label{eq:zero-out}\forall x\in\{-1,1\}^{NR}\setminus H_{\le N}^{NR},\quad \psi(x)=0,\\
        \label{eq:corr}\text{and } \sum_{x\in H^{NR}_{\leq N} \cap D_{f\circ g}}\psi(x)\cdot F^{\le N}(x)-\sum_{x\in H^{NR}_{\leq N} \setminus D_{f\circ g}}|\psi(x)|>\delta.
        \end{align}
    \end{itemize}
\end{definition}

The pure high degree of a function $\psi$ is defined as follows and it is known that the maximum pure high degree of the dual witness is a lower bound on the approximate degree of the original function. 
\begin{definition}[Pure high degree]
\label{def:phd}
    A function $\psi:\{-1,1\}^n\to\mathbb{R}$ has pure high degree at least $\Delta$ if for every polynomial $p:\{-1,1\}^n\to\mathbb{R}$ with $\deg(p)<\Delta$ it satisfies $\sum_{x\in\{-1,1\}^n}p(x)\psi(x)=0$.
    We denote this as $\phd(\psi)\ge \Delta$.
    For a function $\omega:[T]_0\to\R$, we use the same notation:
    $\phd(\omega)\geq\Delta$ means that
    $\sum_{t=0}^T p(t)\omega(t)=0$ for every univariate polynomial $p$
    of degree less than $\Delta$.
\end{definition}

\begin{proposition}[{\cite[Proposition 6.6]{bun2018polynomial}}] \label{prop:dual_dpdeg}
    Let $\delta>0$, $f$ and $g$ be (partial) Boolean functions and let $F^{\le N}$ be defined from \Cref{eq:<=N} and $\psi$ be a dual witness for $F^{\le N}$, then $\dpdeg_{\delta}(F^{\le N})\ge \phd(\psi)$. 
\end{proposition}

\subsubsection{Dual block composition}
\label{sec:dual-BC}
There is a standard approach (see \cite{shi2009quantum, lee2009note, sherstov2013intersection}) to construct the dual witness for a block-composed function $f\circ g$ by combining dual witnesses for $f$ and $g$ respectively.

\begin{definition}[Dual block composition] \label{def:dual_block_comp}
    The dual block composition of two functions $\phi:\{-1,1\}^R\to\mathbb{R}$ and $\psi:\{-1,1\}^N\to\mathbb{R}$ is a function $\phi\star \psi:\{-1,1\}^{NR}\to\mathbb{R}$ defined as 
    $$(\phi\star \psi)(x_1,\dots,x_R)=2^R \,\phi(\sgn(\psi(x_1)),\dots,\sgn(\psi(x_R)))\prod_{r\in[R]}|\psi(x_r)|$$
    where $x=(x_1,\dots,x_R)$, $x_i\in\{-1,1\}^N$, for $i\in [R]$; and $\sgn(z)$ equals $-1$ if $z<0$ and equals $1$ otherwise.
\end{definition}

\begin{proposition}[\cite{sherstov2013intersection,bun2019nearly}]
\label{prop:dual_block_composition}
    The dual block composition satisfies the following properties:
    \begin{itemize}
        \item \emph{Preservation of $\ell_1$-norm:} If $\|\phi\|_1=\|\psi\|_1 = 1$ and $\phd(\psi)\geq 1$, then $\|\phi\star \psi\|_1=1$.
        \item \emph{Multiplicativity of high degree:} $\phd(\phi\star \psi)\geq \phd(\phi)\cdot \phd(\psi)$.
    \end{itemize}
\end{proposition}

We use constructions from \cite{bun2018polynomial} as starting points for our dual witnesses for $\GapOR_R^{\eps N}$ and $\THR_N^k$.
We record the required properties in \Cref{prop:phi} and \Cref{prop:psi}, including our additional observations.
\begin{proposition}[Adapted from {\cite[Section 6.1.2]{bun2018polynomial}}] \label{prop:phi}
Assume $0<\eps N\leq R$. Let $\phi:\{-1,1\}^R\to\mathbb{R}$ be such that
$\phi(-\vecone^R)=-1/2$, $\phi(\vecone^R)=1/2$, and $\phi(x)=0$ for all $x\in\{-1,1\}^R\setminus\{-\vecone^R, \vecone^R\}$.
Then
$\|\phi\|_1=1$, $\phd(\phi)\geq 1$,
and  $$\sum_{x\in H^R_{\geq\eps N}\cup\{\vecone^R\}}
\phi(x)\GapOR_R^{\eps N}(x)=1.$$ 
\end{proposition}

\begin{proof}
The two nonzero values of $\phi$ have absolute values summing to $1$ and sum to $0$, so $\|\phi\|_1=1$ and $\phd(\phi)\geq1$.
Since $0<\eps N\leq R$, both endpoints belong to the promise, and their correlation contributions are $\frac12\cdot1$ and $(-\frac12)\cdot(-1)$, giving correlation $1$.
\end{proof}

\begin{definition}[Error regions]
\label{def:err_reg}
    For a (partial) Boolean function $f:D_f\subseteq \bits^n \to \bits$ and a function $\psi: \bits^n \to\R$, we define $D_+$ and $D_-$, the false positive and false negative regions, respectively, by 
    \begin{align*}
        D_+(\psi, f) := \{x \in D_f : \psi(x) > 0, f(x) = -1\}; \\
        D_-(\psi, f) := \{x \in D_f : \psi(x) < 0, f(x) = +1\}.
    \end{align*}
\end{definition}

The following proposition gives an inner witness with the pure high degree and sign and tail bounds needed by the composition.
\begin{proposition}[A witness with controlled support and tails]
\label{prop:psi}
Fix $k\geq2$, let $c_k=\lceil20(2k)^{k/2}\rceil$ and $N=c_kR$.
For sufficiently large $R$, there is one function $\omega:[R]_0\to\R$
with zero sum and unit $\ell_1$-norm whose symmetrization
\[
 \psi(x)=\begin{cases}\omega(|x|)/\binom N{|x|},&|x|\leq R,\\
                         0,&|x|>R\end{cases}
\]
has unit norm and pure high degree $\Omega_k(N^{1/2-1/(2k)})$.
The same witness satisfies
\begin{align}
 \label{eq:errk}
 \sum_{x\in D_+(\psi,\THR_N^k)}|\psi(x)|&\leq\frac1{48N},&
 \sum_{x\in D_-(\psi,\THR_N^k)}|\psi(x)|&\leq\frac12-\frac2{4^k},\\
 \label{eq:errk+1}
 \sum_{|x|>k}|\psi(x)|&\leq\frac1{48N},&
 \sum_{x\in D_-(\psi,\THR_N^{k+1})}|\psi(x)|&\geq\frac12-\frac1{48N}.
\end{align}
Moreover, with $U=256N$, $a_k=(2k)^k$ and
$b_k=1/(4\sqrt{kR U^{1/k}})$,
\begin{equation}\label{eq:witness-decay}
 |\omega(t)|\leq a_k e^{-b_kt}/t^2\qquad(1\leq t\leq R).
\end{equation}
\end{proposition}
The construction, following the sparse-support polynomial of
\cite[Proposition 5.3]{bun2018polynomial} with a separate auxiliary
parameter $U$, is proved in \Cref{sec:pfpsi}. The cutoff is $R$, while
the Boolean block length is $N>R$; these parameters serve different roles.
We next prove the promised correlation directly from these sign and tail
bounds.

\begin{lemma}
\label{lem:corr}
Let $f=\GapOR_{R}^{\eps N}$, $g= \BTHR_N^k$, $N=\lceil 20 (2k)^{k/2}\rceil R$ for large enough $R$ and $0<\eps<1/(4^{k-1}\lceil 20 (2k)^{k/2}\rceil)$. The functions $\phi$ from \Cref{prop:phi} and $\psi$ from \Cref{prop:psi} satisfy
$$\sum_{x\in D_{f\circ g}}(\phi\star\psi)(x)\cdot (f\circ g)(x)-\sum_{x\in\{-1,1\}^{NR}\setminus D_{f\circ g}}|(\phi\star\psi)(x)|\ge 9/10.$$
\end{lemma}
\begin{proof}
Put $\xi=\phi\star\psi$. Since $\psi$ has zero sum and unit norm, its
positive and negative masses are each $1/2$. Thus
\[
 \lambda_+(x)=2|\psi(x)|\mathbf1_{\psi(x)>0},\qquad
 \lambda_-(x)=2|\psi(x)|\mathbf1_{\psi(x)<0}
\]
are probability distributions, and $\xi$ is the signed mixture
$\tfrac12\lambda_+^{\otimes R}-\tfrac12\lambda_-^{\otimes R}$.
The two product distributions have disjoint supports.
Under $\lambda_+$, a block has weight at least $k$ with probability at
most $1/(24N)$. Hence the product input is a valid $+1$ input except
with probability at most $R/(24N)$.
Under $\lambda_-$, a block has weight at least $k$ with probability
at least $p_0=4^{1-k}$, and weight greater than $k$ with probability
at most $1/(24N)$. By a union bound and one-sided Hoeffding, the
product input is a valid $-1$ input except with probability at most
\[
 \frac{R}{24N}+\exp\bigl(-2R(p_0-\eps c_k)^2\bigr).
\]
Here $\eps c_k<p_0$ by hypothesis. On the complementary event all
blocks have weight at most $k$ and at least $\eps N$ blocks have
weight exactly $k$, so both promises hold. If the bad probabilities
in the two halves are $a_+$ and $a_-$, the promised correlation minus
the outside-promise mass is at least $1-a_+-a_-$: each bad input
contributes at worst $-1$ rather than $+1$ before the half-mixture
weight. Consequently the expression in the lemma is at least
\[
 1-\frac{R}{12N}-\exp\bigl(-2R(p_0-\eps c_k)^2\bigr)\geq\frac9{10}
\]
for sufficiently large $R$. This proves the claim with the same
witness that will be used in zeroing out.
\end{proof}

\subsubsection{Zeroing out the mass outside \texorpdfstring{$H_{\leq N}^{NR}$}{H<=N}}
\label{sec:zero-out}
Let $f=\GapOR_{R}^{\eps N}$, $g= \BTHR_N^k$ and $F^{\le N}$ be defined by \Cref{eq:<=N}.
Recall \Cref{eq:zero-out}, the mass of a dual witness for $F^{\leq N}$ must be zero on inputs outside $H_{\leq N}^{NR}$.
The following lemma guarantees the existence of a modified function similar to $\phi\star\psi$, except that its mass outside $H_{\leq N}^{NR}$ is zeroed out, while the pure high degree decreases by at most a polylogarithmic factor.
\begin{lemma}[{\cite[Proposition 2.22]{bun2018polynomial}}]
\label{lem:zeroing}
    Let $N = \lceil20 (2k)^{k/2}\rceil R$, and let $\phi$ and $\psi$ be functions from \Cref{prop:phi} and \Cref{prop:psi}, then there exists a function $\zeta:\bits^{NR}\to\mathbb{R}$ such that
    \begin{align*}
    \forall x\in\{-1,1\}^{NR}\setminus H_{\le N}^{NR},\quad \zeta(x)=0;& \\
    \|\zeta-\phi\star\psi\|_1\le 2/9;& \\
    \|\zeta\|_1= 1  \quad \text{ and } \quad
      \phd(\zeta)\geq\Omega_k(\sqrt{N^{1-1/k}}/\ln^2N).
    \end{align*}
\end{lemma}
\begin{proof}
We verify all hypotheses of \cite[Proposition 2.22]{bun2018polynomial}.
Its univariate support parameter is $M=R$, not the Boolean block
length $N$. Take $\alpha=a_k=(2k)^k$ and
$\beta=b_k=1/(4\sqrt{kR(256N)^{1/k}})$ from \Cref{prop:psi}.
Then $\omega$ has zero sum, unit norm, and the required decay
$|\omega(t)|\leq\alpha e^{-\beta t}/t^2$ on $[R]$.
For fixed $k$ and large $R$, $1\leq\alpha\leq R^2$ and
\[
 \frac{4\ln^2R}{\sqrt\alpha R}<\beta<1,
 \qquad N=\lceil20\sqrt\alpha\rceil R.
\]
The first strict inequality follows because its ratio grows as
$\Theta_k(R^{1/2-1/(2k)}/\ln^2R)$. By
\Cref{prop:dual_block_composition,prop:psi}, choose an integer
$D_0=\Omega_k(N^{1/2-1/(2k)})<N$ with
$\phd(\phi\star\psi)\geq D_0$. The cited proposition supplies a
normalized $\zeta$ supported on $H_{\leq N}^{NR}$, within $2/9$ in
$\ell_1$ distance, and with pure high degree at least
$\min\{D_0,\Delta\}$, where
\[
 \Delta\geq\frac{\beta\sqrt\alpha R}{4\ln^2R}
       =\Omega_k\left(\frac{N^{1/2-1/(2k)}}{\ln^2N}\right).
\]
This proves all the stated properties.
\end{proof}

Then we prove $\zeta$, the function modified from $\phi\star\psi$, is exactly the dual witness for $F^{\leq N}$.
\begin{lemma}
\label{lem:zeta}
    Let $f=\GapOR_{R}^{\eps N}$, $g= \BTHR_N^k$ and $F^{\le N}$ be defined by \Cref{eq:<=N}, then the function $\zeta:\{-1,1\}^{NR}\to\mathbb{R}$ from \Cref{lem:zeroing} satisfies that
\begin{align}
    \label{eq:zeta1}\forall x\in\{-1,1\}^{NR}\setminus H_{\le N}^{NR},\quad \zeta(x)=0;\\
    \label{eq:zeta2}\sum_{x\in H^{NR}_{\leq N} \cap D_{f\circ g}}\zeta(x)F^{\le N}(x)-\sum_{x\in H^{NR}_{\leq N} \setminus D_{f\circ g}}|\zeta(x)|>2/3;\\
    \label{eq:zeta3} \|\zeta\|_1= 1  \quad \text{ and } \quad
      \phd(\zeta)\geq\Omega_k(\sqrt{N^{1-1/k}}/\ln^2N).
\end{align}
\end{lemma}
\begin{proof}
    \Cref{eq:zeta1} and \Cref{eq:zeta3} hold directly from \Cref{lem:zeroing}.
    For \Cref{eq:zeta2}, by \Cref{lem:corr},
    $$\sum_{x\in D_{f\circ g}}(\phi\star\psi)(x)\cdot (f\circ g)(x)-\sum_{x\in\{-1,1\}^{NR}\setminus D_{f\circ g}}|(\phi\star\psi)(x)|\ge 9/10.$$

By combining it with $\|\zeta-\phi\star\psi\|_1\le 2/9$ from \Cref{lem:zeroing}, we obtain
    \[
    \sum_{x\in D_{f\circ g}}\zeta(x)\cdot (f\circ g)(x)-\sum_{x\in\{-1,1\}^{NR}\setminus D_{f\circ g}}|\zeta(x)|\ge 9/10 -2/9 >2/3.
    \]

    Since $\zeta(x)=0$ for $x\notin H^{NR}_{\leq N}$ and $F^{\leq N}(x) = (f\circ g)(x)$ for $x\in H^{NR}_{\leq N} \cap D_{f\circ g}$, the left hand side equals
    \[
    \sum_{x\in H^{NR}_{\leq N} \cap D_{f\circ g}}\zeta(x)F^{\le N}(x)-\sum_{x\in H^{NR}_{\leq N} \setminus D_{f\circ g}}|\zeta(x)|.
    \]
    Hence the lemma holds.
\end{proof}

\subsubsection{Completing the proof of \texorpdfstring{\Cref{thm:lowerbound}}{lower bound theorem}}
\label{sec:pfLB}
Let $f=\GapOR_{R}^{\eps N}$ and $g= \BTHR_N^k$. Note that $F=\dOCCU_{N,R}^{k,\eps}$ satisfies \Cref{eq:bc}, i.e.,
\[
\dOCCU_{N,R}^{k,\eps}(s_1,\dots,s_N) = (\GapOR_{R}^{\eps N}\circ \BTHR_N^k)(x_1,\dots,x_R), \text{ for all } s\in D_{\dOCCU_{N,R}^{k,\eps}},
\]
where $x=x(s)$ is the incidence encoding of $s$.
Therefore, combining \Cref{lem:poly_method}, \Cref{lem:dummy_augment}, \Cref{lem:dp-encode}, \Cref{prop:dual_dpdeg} and \Cref{lem:zeta}, for parameters satisfying
\[
 k\geq2,\qquad
 0<\eps<\frac{1}{4^{k-1}\lceil20(2k)^{k/2}\rceil},\qquad
 N=\lceil20(2k)^{k/2}\rceil R,
\]
and sufficiently large $R$, we obtain
\[
Q_{1/3}(\OCCU_{N,R+N}^{k,\eps}) = \Omega_{k,\eps}(\sqrt{N^{1-1/k}}/\ln^4N),
\]
which completes the proof of \Cref{thm:lowerbound}. The two logarithmic
factors in \Cref{lem:zeroing}, the factor $\lceil\log_2(R+1)\rceil$ in
\Cref{lem:dummy_augment}, and the factor
$2\lceil\log_2(R+N)\rceil$ in \Cref{lem:poly_method} account for all
four logarithmic factors. The lower threshold on $R$ may depend on
both $k$ and the fixed proximity parameter.

\subsubsection{Proof of \texorpdfstring{\Cref{prop:psi}}{key proposition}}
\label{sec:pfpsi}
We give the sparse-support construction explicitly to ensure that its
degree, signs, tail mass and decay hold simultaneously. It is the
construction of \cite[Proposition 5.3]{bun2018polynomial}, with cutoff
$T=R$ and auxiliary parameter $U=256N$.
Set $c=2k\lceil U^{1/k}\rceil$, $h=\lfloor\sqrt{R/c}\rfloor$, and
\[
 S=\{0,1,\ldots,k\}\cup\{cj^2:1\leq j\leq h\}.
\]
For large $R$, $h\geq1$ and these sets are disjoint. Put $w(t)=0$
outside $S$, and on $S$ set
\[
 w(t)=\frac{\sigma}{\prod_{u\in S\setminus\{t\}}(t-u)},
 \qquad \omega=w/\|w\|_1,
\]
where $\sigma\in\{-1,1\}$ makes $w(k)<0$.
The leading-coefficient identity for Lagrange interpolation gives
$\sum_t w(t)p(t)=0$ for every polynomial $p$ of degree less than
$|S|-1=k+h$. In particular $\omega$ has zero sum and
\[
 \phd(\omega)\geq k+h=\Omega\bigl(\sqrt{R/(kU^{1/k})}\bigr)
                           =\Omega_k(N^{1/2-1/(2k)}).
\]

For $0\leq t\leq k$, cancellation of products gives
\[
 \frac{|w(t)|}{|w(k)|}
 =\binom kt\prod_{j=1}^h\frac{cj^2-k}{cj^2-t}\leq\binom kt.
\]
For $1\leq j\leq h$, use
$\prod_{l=1,l\ne j}^h|j^2-l^2|=(h-j)!(h+j)!/(2j^2)$ to obtain
\begin{align*}
 \frac{|w(cj^2)|}{|w(k)|}
 &\leq\frac{2k!}{(c-k)^k j^{2k}}
                         \frac{(h!)^2}{(h-j)!(h+j)!}\\
 &\leq\frac{(2k)^k}{c^k j^4}\exp(-j^2/(2h)).
\end{align*}
Here $c\geq2k$, $2k!\leq k^k$ for $k\geq2$, and
$\prod_{l=0}^{j-1}(h-l)/(h+j-l)\leq e^{-j^2/(2h)}$.
Since $c^k\geq(2k)^kU$ and $\sum_{j\geq1}j^{-4}<2$,
\begin{equation}\label{eq:omega-tail}
 \sum_{t>k}|\omega(t)|\leq\frac2U\leq\frac1{48N},\qquad
 \frac{\|w\|_1}{|w(k)|}\leq2^k+\frac2U<\frac{4^k}{2}.
\end{equation}
Thus $|\omega(k)|>2/4^k$. Zero sum and unit norm give total negative
mass $1/2$, so its mass below $k$ is at most $1/2-2/4^k$, and its
mass at most $k$ is at least $1/2-1/(48N)$. All positive mass at
weight at least $k$ lies strictly above $k$, where \Cref{eq:omega-tail}
applies. These prove both error bounds in \Cref{prop:psi}.

To check decay, for $t=cj^2$ the preceding ratio bound and $k\geq2$
give $|\omega(t)|\leq(2k)^k e^{-t/(2ch)}/t^2$.
As $c\leq4kU^{1/k}$ and $h\leq\sqrt{R/c}$,
$2ch\leq4\sqrt{kR U^{1/k}}$, establishing \Cref{eq:witness-decay}
at these points. For $1\leq t\leq k$, its right hand side is at least
$(2k)^k e^{-1/4}/k^2\geq1$, whereas $|\omega(t)|\leq1$.
All other positive $t$ have zero mass.

Finally symmetrize on $N$ Boolean coordinates as stated in
\Cref{prop:psi}, extending $\omega$ by zero above $R$.
Averaging a degree-$a$ multilinear polynomial over the weight-$t$
layer gives a univariate polynomial in $t$ of degree at most $a$.
Thus symmetrization preserves the pure high degree. The mass of each
layer is exactly $|\omega(t)|$, so it also preserves normalization,
the sign inequalities and the tail bound. This completes the proof.

\section*{Acknowledgments}
{The work is supported in part by Quantum Science and Technology - National Science and Technology Major Project (Grant No. 2021ZD0302901), NSFC Grant 62272431 and NSFC Grant 62672464.}

\bibliographystyle{alpha}
\bibliography{polished_version/refs}

\clearpage

\appendix
\begin{center}\huge\bf Appendix \end{center}
\section{Useful Tools}\label{sec:useful_tools}%

This section collects the concentration and binomial-inversion results
used in the main proofs.

\begin{lemma}[Median trick]\label{lem:mt}
Suppose an algorithm uses at most $M$ queries and returns a value in
$[a,b]$ with probability at least $p>1/2$. Repeat it independently an
odd number $L$ of times and return the median. The failure probability
is at most $\exp(-2L(p-1/2)^2)$, and the query bound is $LM$.
For $p\geq0.6$, the choice $L=b(\eta)=2\lceil250\ln(2/\eta)\rceil+1$
makes this probability at most $\eta$.
\end{lemma}
\begin{proof}
Let $Z_j$ indicate that run $j$ returns a value in $[a,b]$. These
indicators are independent and $\mathbb E\sum_jZ_j\geq pL$.
If the median is outside $[a,b]$, then $\sum_jZ_j\leq L/2$.
The one-sided Hoeffding inequality therefore gives
\[
 \Pr\left[\sum_jZ_j\leq L/2\right]
 \leq\exp\left(-2(\mathbb E\sum_jZ_j-L/2)^2/L\right)
 \leq\exp(-2L(p-1/2)^2).
\]
\end{proof}

\begin{proposition}
\label{lem:B}
    Let $n\geq1$, $p\in[0,1]$, $0<\eps<1$, and $a>0$, and let $X\sim\cB(n,p)$. Then \begin{align*}
        &\Pr\left[X\notin(1\pm \eps)np\right] \leq 2\exp{(-\eps^2np/3)}; \\
        &\Pr\left[X\notin np\pm a\right] \leq 2\exp{(-2a^2/n)}.
    \end{align*}
\end{proposition}
\begin{proof}
The bounds follow from the multiplicative Chernoff and additive
Hoeffding inequalities, respectively, using $\E[X]=np$.
\end{proof}

\begin{proposition}
\label{lem:inverse}
    Let $M\in\R^{n\times n}$ be an upper triangular matrix, where $[M]_{ij} = \binom{j}{i}$, then $M^{-1}$ is also an upper triangular matrix with $[M^{-1}]_{ij} = (-1)^{j-i}\binom{j}{i}$.
\end{proposition}
\begin{proof}
    Notice that
    \[
    \sum_{j=i}^k (-1)^{k-j}\binom{k}{j}\binom{j}{i} = 
    \begin{cases}
        1, &i=k;\\
        0, &i\neq k,
    \end{cases}
    \]
    which gives the claimed inverse.
\end{proof}

\section{Deferred Proofs from \texorpdfstring{\Cref{sec:star}}{Warm-up Section}}
\label{sec:Appen_claims}\label{sec:pfLem}
\begin{proof}[Proof of \Cref{lem:norm}]
Put $A_i=\sum_{k=i}^d\binom{k}{i}\cnt_k$ and
$Z_i=n\tilde X_i/(i!t_{i-1})$. The accurate count and
\Cref{eq:warm-recursion} imply
\[
 (1-\eps)^i(1-\eta)^{i-1}A_i
 \leq Z_i\leq(1+\eps)^i(1+\eta)^{i-1}A_i.
\]
Since $d(\eps+\eta)=5\delta'/48<1/2$, the inequalities
$e^x-1\leq2x$ for $0\leq x\leq1$ and
$\prod_j(1-x_j)\geq1-\sum_jx_j$ for $0\leq x_j\leq1$ give
$|Z_i-A_i|\leq4d(\eps+\eta)A_i$.
Also $A_i\leq2^dn$, since the degree classes are disjoint.
Binomial inversion (\Cref{lem:inverse}) and
$\sum_{k=1}^i\binom{i}{k}\leq2^d$ now give
\begin{align*}
 \sum_{k=1}^d|\widetilde{\cnt}_k-\cnt_k|
 &\leq\sum_{i=1}^d\sum_{k=1}^i\binom{i}{k}|Z_i-A_i|\\
 &\leq4d^2 4^d(\eps+\eta)n
 =\frac{5\,2^d}{12d!d^d}\delta n\leq\delta n.
\end{align*}
For the last inequality, the coefficient is $5/6$ at $d=1$ and decreases
for $d\geq1$. This proves the stated error bound.
\end{proof}

\section{Deferred proofs from \texorpdfstring{\Cref{sec:disc}}{r-Disc Section}}
\begin{claim}\label{clm:para}
For $1\leq i\leq m$ and $\Gamma\in\cD_{d,r}^i$,
\begin{align*}
 dm^{i-1}&\leq\overline{c^X_\Gamma}\leq d[m(1+\eps)]^{i-1},\\
 \frac{4id\eps m^{i-1}}{2^i m!}
 &\leq\underline{c^X_\Gamma}\leq\frac{4id\eps m^{i-1}}{m!},\\
 \mu_{\Gamma,\Gamma'}&\leq m!,\qquad \mu_\Gamma\leq Dm!.
\end{align*}
\end{claim}
\begin{proof}
Because $\eps\leq1/(100m)$, the ratio
$\underline{c^X_\Gamma}/\overline{c^X_\Gamma}
=4i\eps(1+\eps)/((1+2\eps)^im!)$ lies in $(0,1)$.
Induction in the defining recurrence proves the first line. Substituting
these bounds into the formula for $\underline{c^X_\Gamma}$ proves the
second line, using $1+2\eps\leq2$ and
$(1+\eps)^i\leq(1+2\eps)^i$. Finally an ordered tuple has at most $m$
distinct edges, so there are at most $m!$ such tuples for any fixed pair
of types. Sum over at most $D$ full types for the last bound.
\end{proof}

\begin{claim}\label{clm:alpha&beta}
There are constants $0<\beta<\alpha$, depending only on $d,r$, such that
$\alpha\delta\geq\alpha_\Gamma>\beta_\Gamma\geq\beta\delta$ for all
types and all $\delta\in(0,1)$.
\end{claim}
\begin{proof}
Use $U,V$ and the constants of \Cref{eq:eps&eta}, and take
\[
 \alpha=\frac{16mU\eps_{d,r}}{m!}+2mV\eta_{d,r},\qquad
 \beta=\frac{d\eps_{d,r}}{2m!}.
\]
Since $(1-\eps)^{-(i-1)}\leq2$ and
$\underline{c^X_\Gamma}\leq8mU\eps/m!$, the claimed upper bound follows.
For the lower bound, $\eps\leq1/(100m)$ gives
\[
 \frac{1-2\eps}{(1+\eps)^{i-1}}
 \geq (1-2\eps)(1-(i-1)\eps)\geq\frac12.
\]
Using \Cref{clm:para} and $m\geq2$,
\begin{align*}
 \beta_\Gamma
 &\geq\frac{2id\eps m^{i-1}}{2^i m!}-2iDm!\eta\\
 &\geq i\left(\frac{d\eps}{m!}-2Dm!\eta\right)
 \geq\frac{d\eps}{2m!}=\beta\delta>0.
\end{align*}
The last inequality uses $\eta_{d,r}\leq d\eps_{d,r}/(4D(m!)^2)$.
The strict separation $\alpha_\Gamma>\beta_\Gamma$ is immediate from
their definitions and $\eps>0$. All bounds use the actual maximum
stage $m$, so they hold uniformly for every $i\leq m$.
\end{proof}

\section{An Application: Approximating the Number of Subgraphs}
\label{sec:app}
In this section, we show that \Cref{alg:ApproxP*}  implies the algorithm that approximates the number of $H$-instances
in a digraph $G$ in the quantum unidirectional model, where $H$ is a connected $d$-bounded-degree digraph of radius at most $r$ and an $H$-instance is any (not necessarily induced) subgraph isomorphic to $H$.

\paragraph{The algorithm {EstSubgraph}} \Cref{alg:ApproxSubgraph} firstly invokes \Cref{alg:ApproxP*} as a subroutine with smaller error parameter and obtains estimates of all isomorphism types $\{\widetilde{\cnt}_\Gamma\}_{\Gamma\in\cD_{d,r}^*}$.
Then it chooses a valid root for $H$ and finds the corresponding isomorphism type $\Gamma_H$.
At the final step, it outputs an estimate $\widetilde{\#H}$ based on the estimates $\widetilde{\cnt}_{\Gamma'}$ for those $\Gamma'\succeq {\Gamma_H}$.

\begin{algorithm}[ht]
        \caption{Estimating the number of subgraph copies}
        \label{alg:ApproxSubgraph}
        \begin{algorithmic}[1] 
            \Procedure{EstSubgraph}{$G, n, d, r ,  \delta, H$}
                \If{$|V(H)|=1$}
                \State \Return $n$.
                \EndIf
                \State $\{\widetilde{\cnt}_\Gamma\}_{\Gamma\in\cD_{d,r}^*} = $ \textsc{EstDisc*}$(G,n,d,r,\delta/m_{d,r}!)$.
                \State Choose a vertex $v\in V(H)$ of distance at most $r$ to all the other vertices in $H$. \label{line:chooseroot}
                \State Let $H_v$ be the rooted graph same as $H$ but with $v$ marked as root. 
                \State Let $\Gamma_H$ be the isomorphism type such that $H_v\cong \Gamma_H$.
                \State Let $c_H = \abs{\{u\in V(H): H_u\cong H_v\}}$ where $H_u$ is the rooted graph same as $H$ but with $u$ marked as root. \label{line:cH}
                \State \Return $\widetilde{\#H} = \frac{1}{c_H}\sum_{\Gamma'\succeq \Gamma_H} \frac{\mu_{\Gamma_H,\Gamma'}}{\mu_{\Gamma_H,\Gamma_H}}{\widetilde{\cnt}_{\Gamma'}}$.
            \EndProcedure
        \end{algorithmic}
    \end{algorithm}
    
\begin{restatable}{theorem}{subgraph}\label{thm:subgraph}
    Given an $n$-vertex $d$-bounded-degree digraph $G$, a $d$-bounded-degree connected digraph $H$ of radius at most $r$, and a constant $\delta \in (0,1)$, then for sufficiently large $n$, there is a quantum algorithm that outputs an estimate $\widetilde{\#H}$ for the number $\#H$ of $H$-instances in $G$  satisfying
    \[
    \abs{\widetilde{\#H} - \#H} \leq \delta n,
    \]
    with probability at least $5/9$. For $d,r\geq1$, its deterministic
    unidirectional query bound is
    \[
      O_{\delta,d,r}\left(n^{1/2-1/(2(2^{m_{d,r}}-1))}\right),
      \qquad m_{d,r}=\frac{2d((2d)^{r+1}-1)}{2d-1}.
    \]
\end{restatable}

\begin{proof}
    If $H$ has one vertex, its count is exactly $n$ and the algorithm uses zero queries. Otherwise $H$ is connected and has an edge. Since the radius of $H$ is at most $r$, there exists a vertex $v$ that can be chosen as root in Line \ref{line:chooseroot}.
    Then there exists a unique isomorphism type $\Gamma_H\in\cD_{d,r}^*$ such that $H_v\cong \Gamma_H$.
    The parameter $c_H$ defined in Line \ref{line:cH} represents the number of vertices in $V(H)$ that are equivalent to $v$, while each of them can be marked as root equivalently.
    
    Then we have the following key observation:
    \begin{equation}
    \label{eq:obs}
        \#H = \frac{1}{c_H}\sum_{\Gamma'\succeq \Gamma_H} \frac{\mu_{\Gamma_H,\Gamma'}}{\mu_{\Gamma_H,\Gamma_H}}\cnt_{\Gamma'}.
    \end{equation}

    Inside a rooted disc of type $\Gamma'\succeq\Gamma_H$, each rooted copy of $\Gamma_H$ has exactly $\mu_{\Gamma_H,\Gamma_H}$ edge orderings consistent with its prescribed prefix types. Thus $\mu_{\Gamma_H,\Gamma'}/\mu_{\Gamma_H,\Gamma_H}$ counts these rooted copies.

    For each $H$-instance $H^*$, there are $c_H$ vertices that can be marked as root equivalently, i.e., $H^*$ can be counted by $c_H$ different $r$-discs because of the different choices of the root.
    That means we count it $c_H$ times, thus we need to multiply by $1/c_H$.

    Based on observation (\ref{eq:obs}),
    we can bound the error between $\widetilde{\#H}$ and $\#H$ as
    \begin{align*}
        \abs{\widetilde{\#H}-\#H} &\leq \frac{1}{c_H}\sum_{\Gamma'\succeq \Gamma_H} \frac{\mu_{\Gamma_H,\Gamma'}}{\mu_{\Gamma_H,\Gamma_H}}\abs{{\widetilde{\cnt}_{\Gamma'}}-\cnt_{\Gamma'}} \leq \frac{1}{c_H} \left(\max_{\Gamma'\succeq \Gamma_H} \frac{\mu_{\Gamma_H,\Gamma'}}{\mu_{\Gamma_H,\Gamma_H}} \right) \sum_{\Gamma'\succeq \Gamma_H} \abs{{\widetilde{\cnt}_{\Gamma'}}-\cnt_{\Gamma'}} \\
        &\leq \frac{1}{c_H} \left(\max_{\Gamma'\succeq \Gamma_H} \frac{\mu_{\Gamma_H,\Gamma'}}{\mu_{\Gamma_H,\Gamma_H}} \right) \sum_{\Gamma'\in\cD_{d,r}^*} \abs{{\widetilde{\cnt}_{\Gamma'}}-\cnt_{\Gamma'}} \leq \frac{1}{c_H} \left(\max_{\Gamma'\succeq \Gamma_H} \frac{\mu_{\Gamma_H,\Gamma'}}{\mu_{\Gamma_H,\Gamma_H}} \right)  \frac{\delta}{m_{d,r}!} n \\
        &\leq \delta n.
    \end{align*}
    The second last inequality holds because of \Cref{thm:disc*}. The last inequality holds since $c_H\geq 1$, $\mu_{\Gamma_H,\Gamma_H}\geq 1$ and $\mu_{\Gamma_H,\Gamma'}\leq m_{d,r}!$ (see \Cref{clm:para}).

    Notice that \Cref{alg:ApproxSubgraph} invokes \Cref{alg:ApproxP*} only once with the error parameter shrinking by a constant multiple,
    thus the success probability and the query complexity are identical to \Cref{alg:ApproxP*}.
\end{proof}

\section{A Simple \texorpdfstring{$\Omega(n^{1/3})$}{} Lower Bound}
\label{sec:BKLB}
The standard quantum collision lower bound already yields the following
baseline result without a support-size estimation reduction.
\begin{theorem}\label{thm:BKLB}
For every fixed $0<\eps<1/8$, testing $2$-star-freeness on
$n$-vertex $2$-bounded-degree digraphs requires
$\Omega(n^{1/3})$ quantum queries in the unidirectional model.
\end{theorem}
\begin{proof}
For even $L$, distinguishing a one-to-one function $f:[L]\to[L]$ from
a two-to-one function requires $\Omega(L^{1/3})$ quantum queries,
including for this alphabet size \cite{ambainis2005polynomial}.
Form the bipartite digraph with vertices $u_1,\ldots,u_L,v_1,\ldots,v_L$
and edges $(u_j,v_{f(j)})$. Its maximum out-degree is one and its
maximum in-degree is two. In the one-to-one case it is $2$-star-free;
in the two-to-one case it contains $L/2$ edge-disjoint $2$-stars.
Each outgoing query is simulated coherently with at most two calls to
$f$ as in \Cref{obs:reduction}.

For every sufficiently large $n$, choose $L=2\lfloor n/4\rfloor$ and add
$n-2L<4$ isolated vertices. At least $L/2$ edges still have to be
deleted, and $L/(4n)\to1/8$. Thus for any fixed $\eps<1/8$ these
graphs are $\eps$-far for all sufficiently large $n$. A tester for
them would solve collision with the same asymptotic number of queries.
Since $L=\Theta(n)$, the claimed bound follows.
\end{proof}

\end{document}